\documentclass[11pt,a4paper]{article}
\usepackage{amsmath,amssymb,amsthm}
\usepackage{report-style}
\usepackage{multirow}
\usepackage{xcolor}

\newtheorem{theorem}{Theorem}[section]

\newtheorem{assumption}[theorem]{Assumption}

\definecolor{liftgreen}{HTML}{187A3E}

\usepackage{bm}

\ExplSyntaxOn
\clist_map_inline:nn
  {a,b,c,d,e,f,g,h,i,j,k,l,m,n,o,p,q,r,s,t,u,v,w,x,y,z}
  {
    \str_if_eq:nnF {#1} {v}
      {\cs_new:cpn {v#1} {\bm{#1}}}
    \cs_new:cpn {ev#1}  {#1}
    \cs_new:cpn {rv#1}  {\mathbf{#1}}
    \cs_new:cpn {erv#1} {\textnormal{#1}}
    \str_case:nnF {#1} {{m}{} {q}{}}
      {\cs_new:cpn {r#1} {\textnormal{#1}}}
  }
\clist_map_inline:nn
  {A,B,C,D,E,F,G,H,I,J,K,L,M,N,O,P,Q,R,S,T,U,V,W,X,Y,Z}
  {
    \cs_new:cpn {m#1}   {\bm{#1}}
    \cs_new:cpn {em#1}  {#1}
    \cs_new:cpn {rm#1}  {\mathbf{#1}}
    \cs_new:cpn {erm#1} {\textnormal{#1}}
    \cs_new:cpn {t#1}   {\tens{#1}}
    \cs_new:cpn {et#1}  {\etens{#1}}
    \cs_new:cpn {g#1}   {\mathcal{#1}}
    \str_if_eq:nnF {#1} {E}
      {\cs_new:cpn {s#1} {\mathbb{#1}}}
  }
\ExplSyntaxOff

\DeclareMathAlphabet{\mathsfit}{\encodingdefault}{\sfdefault}{m}{it}
\SetMathAlphabet{\mathsfit}{bold}{\encodingdefault}{\sfdefault}{bx}{n}
\newcommand{\tens}[1]{\bm{\mathsfit{#1}}}
\newcommand{\etens}[1]{\mathsfit{#1}}

\newcommand{\modelname}{TransX}
\newcommand{\reporttitle}{\modelname \ Technical Report}
\newcommand{\reportsubtitle}{Scaling Sequence Transduction for Recommendation}
\newcommand{\reportauthors}{Da Xu\texorpdfstring{\ifdefempty{\reportcontact}{}{\footnotemark[1]}}{}, Liyan Fang, Divya Venugopalan, Sunny Hsu, Xukai Wang, Rishav Chowdhury, Cindy Liang, Nishant Lakshmikanth}
\newcommand{\reportaffiliation}{LinkedIn}
\newcommand{\reportcontact}{daxu5180@gmail.com}

\newcommand{\reportversion}{}
\newcommand{\reportshorttitle}{\modelname: Technical Report}

\newcommand{\projecturl}{}
\newcommand{\codeurl}{}
\newcommand{\modelurl}{}

\draftnotestrue
\showcontentstrue
\showcoverfiguretrue

\hypersetup{
  pdftitle={\reporttitle: \reportsubtitle},
  pdfauthor={\reportauthors},
  pdfsubject={TransX technical report}
}

\begin{document}
\makereporttitle
\begin{abstract}

We present TransX, a sequence-transduction framework and co-designed infrastructure for scaling industrial recommendation with explicit control over model-compute allocation and execution amortization. TransX separates sequence encoding, behavior memory, and candidate-conditioned fusion into independently scalable computational roles. Training amortizes causal trajectory computation across behavioral prefixes and fans out memory states to candidates; serving disaggregates bulk prefill, nearline refresh, and online prediction through compact behavior-memory KV.

We empirically observe that scaling returns are anisotropic: encoder scaling may saturate beyond moderate depth at fixed memory resolution and fusion capacity; increasing supervision shifts the allocation frontier toward candidate-conditioned fusion; and the memory resolution needed to preserve quality remains roughly stable as model capacity grows. We train the core TransX model on click and conversion, then post-train the shared behavioral interface for retrieval and preference ranking, with generative recommendation explored offline.

In separate production A/B experiments, the core TransX model improves CTR by 6.02\% and conversion by 4.45\%, while post-trained retrieval and preference-ranking variants improve CTR by 2.96\% and conversion by 1.33\%, respectively, over their corresponding production baselines. Under matched training conditions, the co-designed training infrastructure increases supervised-target throughput by 84\%. 
In serving, disaggregated execution reduces TransX's total serving cost by almost \(80\%\) relative to colocated serving at matched traffic and SLO, allowing the deployed system to sustain approximately \(5\times\) the model FLOPs per request of the HSTU-style production baseline at the same total serving cost. Even at matched per-request model FLOPs, TransX delivers approximately \(4\times\) higher request throughput per H200.
Together, these results show that scaling recommendation model capacity need not proportionally increase online serving cost.
\end{abstract}

\begingroup
\setlength{\intextsep}{4pt}
\begin{figure}[hb]
\centering
\includegraphics[width=\linewidth]{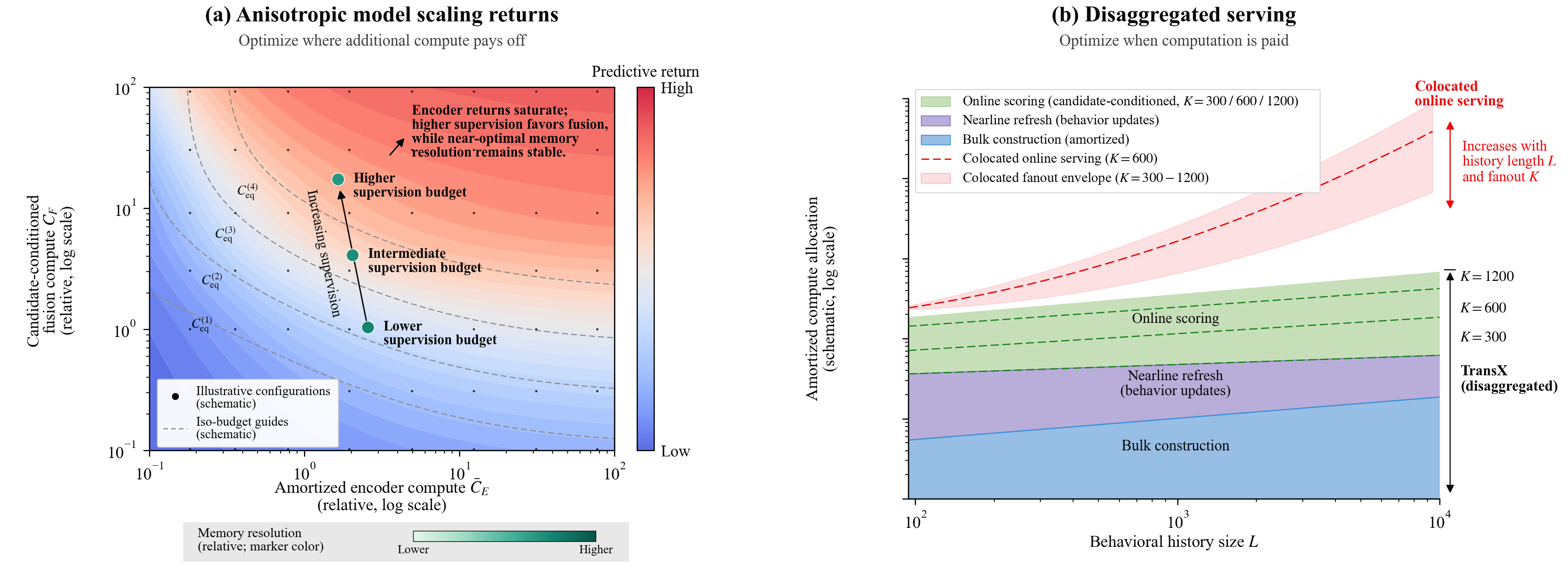}
\caption{
\textbf{TransX decouples recommendation scaling from online serving cost.}
\textbf{(a)} Scaling returns for TransX differ across encoder, behavior-memory, and candidate-conditioned fusion capacity, with higher supervision shifting the validation-selected allocation toward fusion while the memory resolution needed to preserve quality remains roughly stable.
\textbf{(b)} Disaggregated serving amortizes history-dependent computation through offline and nearline execution, concentrating the online path on state access and candidate-conditioned prediction.
}
\label{fig:overview}
\end{figure}
\endgroup

\clearpage

\ifshowcontents
  \begingroup
    \small
    \setlength{\parskip}{0pt}
    \tableofcontents
  \endgroup
  \clearpage
\fi

\section{Introduction}
\label{sec:introduction}

Long behavioral histories provide recommendation systems with evidence
about persistent preferences, evolving interests, and immediate intent~\citep{boka2024survey}.
Recent efforts to scale recommendation models have focused largely on
increasing sequence-handling capacity: consuming longer histories,
allocating more model computation to them, and engineering the serving
stack to sustain the resulting workload~\citep{zhai2024actions,chai2025longer}. When additional computation is
placed on the online serving path, however, its cost scales directly
with online prediction demand. Recommendation scaling is therefore not
only a question of predictive gain, but also of where computation is
placed and how often it must be executed.

This raises two more fundamental questions that are easily obscured by
generic compute scaling:
\begin{quote}
    \emph{Which computational components should scale to advance the
    recommendation-quality frontier, and must advancing that frontier increase
    online serving cost proportionally?}
\end{quote}

Answering these questions leads naturally to a model--infrastructure
co-design problem. The first is a \emph{model-design} question: the
architecture must expose the computational components whose capacities
can be scaled along different axes. The second is an \emph{infrastructure-design}
question: those components should execute at their designated amortization scope and trigger, rather than all at request time.

These two design problems are coupled through the model’s compute–state decomposition. The same computational decomposition that exposes distinct scaling axes also defines the state boundaries across which computation can be amortized, recomputed, and disaggregated.

We develop \textbf{TransX}, a sequence-transduction
framework that instantiates this model--infrastructure co-design
principle. Its model architecture exposes three distinct computational
roles:
\begin{enumerate}
    \item \textbf{Sequence encoding.}
    A causal encoder transforms the evolving behavioral trajectory
    independently of the current candidate, allowing computation over
    overlapping prefixes to be amortized across prediction events.

    \item \textbf{Behavior memory.}
    A bounded global--local memory separates the amount of history
    processed from the behavioral representation exposed to downstream
    prediction, making memory resolution an explicit scaling axis while exposing the memory as a reusable state boundary.

    \item \textbf{Candidate-conditioned fusion.}
    Candidate queries cross-attend to the shared memory as part of the task-specific prediction path. Behavior memory KV is shared across candidate fanout while preserving candidate-specific fusion.
\end{enumerate}
This decomposition makes the amortization boundary explicit: behavioral-state construction over shared prefixes and its incremental refresh are amortized across requests and candidate fanout, while only prediction-conditioned computation scales with request and candidate demand.

The infrastructure preserves and exploits the same decomposition in execution, translating the model design into substantial system-level efficiency gains. During training, variable-length behavioral prefixes from each trajectory share one causal encoding, while candidate fanout at each prefix shares one event state. During serving, bulk (offline) and nearline workers construct and refresh behavioral states on model rollout and behavior updates, while online workers consume bounded, immutable behavior memory KV cache for candidate-conditioned prediction. Thus, history-dependent computation is performed primarily when behavioral state is initialized or refreshed, while the online serving path is concentrated on state access and request- and candidate-conditioned prediction.

This decomposition also lets us study how different compute allocations advance the recommendation-quality frontier. Our scaling experiments reveal unequal marginal returns across these scaling axes: encoder scaling improves prediction at fixed behavior-memory resolution but saturates beyond moderate depth at fixed fusion capacity; increasing supervision shifts the allocation frontier toward candidate-conditioned fusion; and the memory resolution needed to preserve quality remains roughly stable even as overall model capacity grows. Together, these results suggest that efficient scaling is a compute-allocation problem: additional capacity should be allocated according to component-wise marginal returns rather than increased proportionally across all computational roles. The disaggregated infrastructure preserves the benefit of this allocation in execution by amortizing history-dependent computation outside the online serving path, relaxing the coupling between history-side model capacity and online serving cost.

Finally, we reuse the same behavioral memory interface for capability specialization across recommendation tasks. We first train the TransX backbone on response prediction, then post-train capability-specific memory readers and output heads for retrieval, preference ranking, and generative recommendation. Core training learns a reusable behavioral memory interface; post-training specializes how each capability queries it and maps it to task-specific outputs.

The contributions of this work are:
\begin{itemize}
    \item \textbf{Recommendation as request-conditioned sequence transduction.}
    We formulate recommendation as repeated request-level queries
    over an evolving behavior stream, and expose a compute--state
    decomposition into sequence encoding, behavior memory, and
    candidate-conditioned fusion. This decomposition identifies the model states that can be shared across serving events and candidate fanout.

    \item \textbf{Model--infrastructure co-design for amortized execution.}
    The same compute--state decomposition defines the reuse and
    recomputation boundaries preserved by multi-prefix sequence-to-sequence
    training, sequence--fanout parallelism, tiered KV residency, and
    disaggregated serving. 
    Expensive sequence computation can therefore be amortized across overlapping prefixes and executed outside the request path until behavioral state changes.

    \item \textbf{Empirical characterization of scaling returns and compute allocation.}
    Controlled experiments vary behavioral-encoder compute,
    behavior-memory resolution, and candidate-conditioned fusion compute,
    revealing unequal marginal returns across these scaling axes: encoder
    returns saturate beyond moderate depth at fixed fusion capacity,
    increasing supervision shifts the validation-selected allocation
    toward candidate-conditioned fusion, and the memory resolution needed
    to preserve quality remains roughly stable as model capacity grows.

    \item \textbf{Capability specialization over a shared behavior-memory interface.}
    We post-train capability-specific memory readers and output heads for
    retrieval, preference ranking, and generative recommendation while
    preserving the behavior memory interface. Production deployment of the core, retrieval, and ranking variants further validates the framework across recommendation quality, training efficiency, and serving economics.
\end{itemize}

We evaluate TransX in a large-scale social recommendation application
serving approximately 20K requests per second with roughly
0.5K--1.2K candidates per request. In separate production A/B
experiments, TransX-Base improves click-through rate (CTR) by $6.02\%$ and conversion by $4.45\%$ over the production sequential model baseline ; TransX-Retrieval improves CTR by $2.96\%$ over the sequential two-tower production model
baseline; and TransX-Ranking improves conversion by $1.33\%$
over a DLRM learning-to-rank baseline. Under matched training
conditions, the optimized execution pipeline increases supervised-target
throughput by $84\%$. Production serving further demonstrates that
history-dependent computation can be moved off the online serving path while
maintaining bounded behavior memory KV, nearline freshness, and latency-critical
candidate scoring.


Taken together, these results motivate a view of recommendation scaling as a joint problem of compute allocation and amortized/disaggregated execution. The central question is not simply how much computation to add, but where additional capacity yields the greatest predictive return, when each component must be recomputed, and over what scope its state can be reused. TransX realizes this view through a compute--state decomposition that jointly controls model capacity, behavior-memory resolution, and execution placement, enabling the recommendation-quality frontier to advance without proportional growth in online serving cost.
\section{Problem Framing}
\label{sec:framing}

We first formalize the prediction interface underlying industrial recommender systems. A user generates an evolving behavioral trajectory, while the stream of recommendation requests arrives on a separate clock and introduces candidate sets to be scored. The model produces predictions for these candidates before downstream selection
determines which items are exposed and which responses are subsequently observed. 

We begin with the behavior–request–exposure–feedback process and its causal information boundary, then formulate recommendation as
request-conditioned sequence transduction, in which successive user requests
query cutoff-valid prefixes of the user's evolving behavioral trajectory.
Finally, we identify the resulting amortization boundary between
candidate-independent behavioral-state construction and
candidate-conditioned prediction. This boundary provides the formal
basis for the architecture in Section~\ref{sec:model-arch}, the execution
design in Section~\ref{sec:core-training} and \ref{sec:general-infra}, and the compute-allocation analysis in Section~\ref{sec:scaling}.

\subsection{Data and Prediction Interface}
\label{sec:data-interface}

\begin{table}[htb]
    \centering
    \resizebox{\textwidth}{!}{
    \begin{tabular}{|p{4cm}|p{13cm}|}
    \hline
        $u$
        & User index. \\ \hline

        $t$
        & Request index. \\ \hline

        $\tilde t$
        & Behavior timestamp/index, not necessarily aligned with request index $t$. \\ \hline

        $j$
        & Candidate index within a request. \\ \hline

        $k_t,m_t$
        & Number of candidate queries scored in request $t$ and number of candidates exposed, where $m_t\le k_t$. \\ \hline

        $\text{cand}_t^{(u,j)}$
        & Identity of candidate $j$ in request $t$. \\ \hline

        $T_u$ & number of request/serving events considered for user $u$. \\ \hline

        $x_t^{(u,j)}$
        & Prediction-time features for $(u,\text{cand}_t^{(u,j)})$, including
        user attributes, candidate attributes, user--candidate interaction
        features, and request-level context shared across candidates,
        available at scoring time. \\ \hline

        $e_t^{(u,j)}$
        & Candidate-query record $(\text{cand}_t^{(u,j)},x_t^{(u,j)})$. \\ \hline

        $\mathcal{Q}_t^{(u)}$
        & Candidate-query set scored in request $t$. \\ \hline

        $I_t^{(u)}$
        & Ordered indices of candidates selected for exposure. \\ \hline

        $\mathcal{S}_t^{(u)}$
        & Served/exposed slate selected from $\mathcal{Q}_t^{(u)}$. \\ \hline

        $y_t^{(u,j)}\in\mathsf{Y}$
        & Observed response for an exposed candidate where $\mathsf{Y}$ is the outcome space. \\ \hline

        $\mathcal{Y}_t^{(u)}$
        & Feedback aligned with the served slate $\mathcal{S}_t^{(u)}$. \\ \hline

        $h_{\tilde t}^{(u)}$
        & Tokenized user behavior at behavior time $\tilde t$. \\ \hline

        $\mathcal{H}^{(u)}$
        & User-behavior stream. \\ \hline

        $H_{\le t}^{(u)}$
        & Prefix of the behavior stream available when request $t$ is scored. \\ \hline

        $l_t^{(u)}$ & Number of behavior records available to user $u$
        when request $t$ is scored; equivalently,
        $l_t^{(u)} = |H_{\le t}^{(u)}|$. \\ \hline

        $\widehat{\mathcal{P}}_t^{(u)}$ & Collection of predicted response distributions for the candidate queries in request \(t\). \\ \hline
    \end{tabular}
    }
    \caption{Notation used in the problem formulation.}
    \label{tab:notation}
\end{table}

We distinguish the \emph{request}, which defines the candidate set
scored by the model, from the \emph{served event}, which contains the
subset ultimately selected for user exposure. The resulting data flow is
\[
H_{\le t}^{(u)},\ \mathcal{Q}_t^{(u)}
\longrightarrow
\widehat{\mathcal{P}}_t^{(u)}
\longrightarrow
I_t^{(u)},\ \mathcal{S}_t^{(u)}
\longrightarrow
\mathcal{Y}_t^{(u)}.
\]

\paragraph{Candidate-query request.}
At request $t$, the ranking system scores $k_t$ candidate queries:
\begin{equation*}
\label{eq:candidate-query-event}
\mathcal{Q}_t^{(u)}
\coloneqq
\left(
e_t^{(u,1)},\ldots,e_t^{(u,k_t)}
\right),
\qquad
e_t^{(u,j)}
\coloneqq
\left(
\text{cand}_t^{(u,j)},x_t^{(u,j)}
\right),
\end{equation*}
where $\text{cand}_t^{(u,j)}$ identifies candidate $j$ and
$x_t^{(u,j)}$ contains user, the candidate, user--candidate, request, and other features available when the candidate is scored. The ordered request stream is
\[
\mathcal{Q}^{(u)}
=
\left(
\mathcal{Q}_1^{(u)},
\mathcal{Q}_2^{(u)},
\ldots
\right).
\]

\paragraph{Selection and exposure.}
The model produces response predictions for all $k_t$ candidate
queries. The exposure policy then chooses
$m_t\le k_t$ candidates,
\begin{equation*}
\label{eq:selected-indices}
I_t^{(u)}
\coloneqq
\left(
i_{t,1}^{(u)},\ldots,i_{t,m_t}^{(u)}
\right),
\end{equation*}
where each $i_{t,r}^{(u)}\in\{1,\ldots,k_t\}$, and emits the served
slate
\begin{equation*}
\label{eq:served-slate}
\mathcal{S}_t^{(u)}
\coloneqq
\left(
e_t^{(u,i_{t,1})},\ldots,
e_t^{(u,i_{t,m_t})}
\right).
\end{equation*}
Thus, the model produces response predictions, whereas the ranking and
selection policy determines exposure. Features revealed only after
selection or exposure cannot enter the preceding candidate scorer.

\paragraph{Feedback.}
Following exposure, the user may generate explicit or implicit
feedback for candidates in the served event:
\begin{equation*}
\label{eqn:feedback-stream}
\mathcal{Y}_t^{(u)}
\coloneqq
\left(
y_t^{(u,i_{t,1})},\ldots,
y_t^{(u,i_{t,m_t})}
\right),
\qquad
y_t^{(u,i_{t,r})}\in \mathsf{Y}.
\end{equation*}
Feedback positions are therefore aligned with the exposed candidates,
not with all candidates scored in the request. The corresponding
feedback stream is
\[
\mathcal{Y}^{(u)}
=
\left(
\mathcal{Y}_1^{(u)},
\mathcal{Y}_2^{(u)},
\ldots
\right).
\]

\begin{figure}[hbt]
    \centering
    \includegraphics[width=0.9\linewidth]{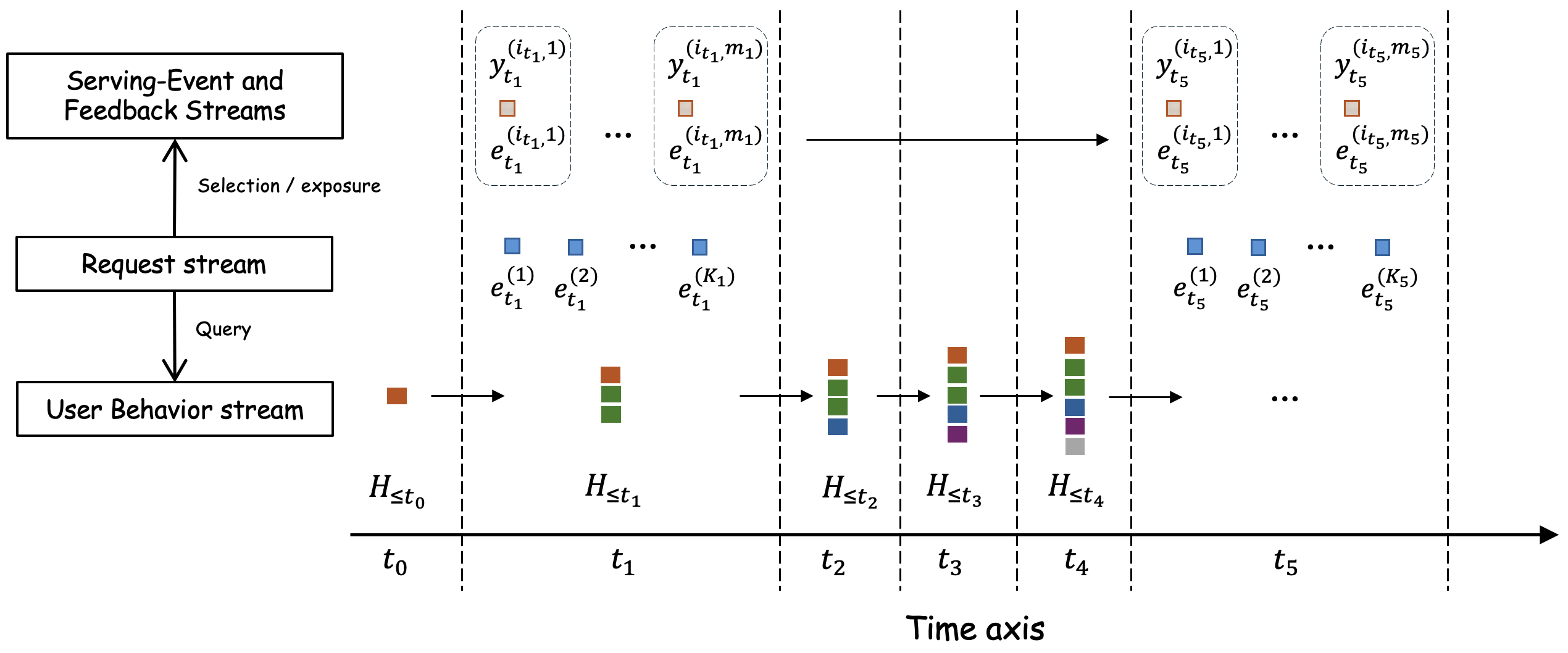}
    \caption{
    Behavior--request--exposure--feedback process.
    Each request queries an independently evolving behavior stream
    through the available prefix $H_{\le t}$; predictions are produced
    for the request candidate set, a downstream policy selects the
    served slate, and feedback is observed only for exposed candidates.
    }
    \label{fig:placeholder}
\end{figure}

\paragraph{User behavior stream.}
Independently of the timing of the target ranking system, a user
generates a broader sequence of activities---for example,
\texttt{search}, \texttt{profile view}, \texttt{connect}, or
interactions with previously served content. We denote the tokenized
behavior stream by
\begin{equation*}
\label{eqn:behavior-stream}
\mathcal{H}^{(u)}
\coloneqq
\left(
\ldots,h_{\tilde t_1}^{(u)},h_{\tilde t_2}^{(u)},\ldots
\right).
\end{equation*}
The behavior clock $\tilde t$ need not align with the request clock
$t$: multiple behaviors may become available between successive
requests, while multiple requests may observe the same behavior
prefix.

We write $H_{\le t}^{(u)}
\preceq
\mathcal{H}^{(u)}$
for the ordered prefix of behavioral records available when request
$t$ is scored. 
The behavior stream and request stream therefore can evolve on different
clocks. The former supplies historical evidence that can be shared
across requests, while each request introduces a new set of
candidate-specific prediction queries. This asynchronous structure is
captured by the request-conditioned sequence-transduction formulation
introduced next.


\subsection{Request-Conditioned Sequence Transduction}
\label{sec:sequence-transduction}

Sequence transduction maps an input sequence to an output sequence
without requiring the two sequences to have the same length or a
position-wise correspondence~\citep{sutskever2014sequence}. Recommendation naturally exhibits this
structure: the input is an evolving behavior sequence, while prediction
outputs occur on a separate sequence of requests.

For user $u$, we write
\begin{equation*}
\label{eq:sequence-transduction}
\left(
\mathcal{H}^{(u)},
\mathcal{Q}_1^{(u)},\ldots,
\mathcal{Q}_{T_u}^{(u)}
\right)
\longrightarrow
\left(
\widehat{\mathcal{P}}_1^{(u)},\ldots,
\widehat{\mathcal{P}}_{T_u}^{(u)}
\right),
\end{equation*}
where
\begin{equation*}
\label{eq:event-response}
\widehat{\mathcal{P}}_t^{(u)}
=
\left(
\widehat p_t^{(u,1)},\ldots,
\widehat p_t^{(u,k_t)}
\right),
\qquad
\widehat p_t^{(u,j)}
=
p_{\theta}
\left(
\,\cdot\,
\middle|
H_{\le t}^{(u)},
e_t^{(u,j)}
\right).
\end{equation*}
Thus, the transducer produces a response prediction for every
candidate query in the request.

Only exposed candidates receive direct response supervision from the
logged event. The user-level training objective is therefore
\begin{equation*}
\label{eq:transduction-training-loss}
\mathcal{L}^{(u)}(\theta)
=
\sum_{t=1}^{T_u}
\sum_{r=1}^{m_t}
\ell
\left(
\widehat p_t^{(u,i_{t,r})},
y_t^{(u,i_{t,r})}
\right).
\end{equation*}
Predictions for unexposed candidates participate in ranking but do
not receive direct logged response supervision from that request.

We call this framing \emph{request-conditioned sequence transduction}:
an evolving behavioral trajectory is repeatedly queried under realized
request and candidate contexts. The behavior stream determines the
historical evidence available at each request, while the request
determines how that evidence must be read for the current prediction
problem.

\subsection{Memory Fanout and the Amortization Boundary}
\label{sec:reuse-boundary}

The transduction formulation exposes a computational boundary between
candidate-independent behavioral state and candidate-conditioned
interaction. We summarize the architecture through four abstract
operators.

Given the behavior prefix,
\begin{equation*}
\label{eq:behavior-encoding}
B_t^{(u)}
=
\operatorname{Enc}
\left(
H_{\le t}^{(u)}
\right),
\end{equation*}
where $\operatorname{Enc}$ produces the contextualized behavioral
representation eligible at request $t$. A memory-construction operator
then forms the reader-facing state
\begin{equation*}
\label{eq:behavior-memory}
M_t^{(u)}
=
\operatorname{Mem}
\left(
B_t^{(u)}
\right).
\end{equation*}
Both stages are independent of the candidates in the current request.

For each candidate query, the resulting candidate-conditioned representation is obtained via:
\begin{equation*}
\label{eq:memory-serving-fusion}
r_t^{(u,j)}
=
\operatorname{Fuse}
\left(
M_t^{(u)},
q_t^{(u,j)}
\right),
\qquad
j=1,\ldots,k_t, 
\end{equation*}
where $q_t^{(u,j)}$ denotes the candidate query. Then the fused representation is passed to the prediction head along with the request-level features to predict the task's target:
\begin{equation*}
\label{eq:response-prediction}
\widehat p_t^{(u,j)}
=
\operatorname{Pred}
\left(
r_t^{(u,j)},
x_t^{(u,j)}
\right).
\end{equation*}

Hence one shared memory fans out to all candidate queries:
\begin{equation*}
\label{eq:candidate-query-fanout}
M_t^{(u)}
\quad
\xRightarrow{\;\mathcal{Q}_t^{(u)}\;}
\quad
\left\{
r_t^{(u,j)}
\right\}_{j=1}^{k_t}
\quad
\xrightarrow{\;\operatorname{Pred}\;}
\quad
\widehat{\mathcal{P}}_t^{(u)}.
\end{equation*}

This candidate-query fanout defines the \emph{request amortization boundary}.
Computation used to construct a valid behavioral memory can be shared
across all candidate queries in the request and, while the memory
version remains valid, across successive requests. Candidate-conditioned
fusion and response prediction scale with request and candidate demand.
The boundary therefore avoids re-encoding history once per candidate
without eliminating candidate-dependent memory access, interaction, or
output computation.

The two resulting amortization scopes are:
\[
\underbrace{
H_{\le t}^{(u)}
\xrightarrow{\;\operatorname{Enc}\;}
B_t^{(u)}
\xrightarrow{\;\operatorname{Mem}\;}
M_t^{(u)}
}_{\substack{
\text{\textbf{Behavioral-memory computation}}\\
\text{prefix modeling capacity}
}}
\qquad
\underbrace{
M_t^{(u)}
\xRightarrow[\mathcal{Q}_t^{(u)}]
{\;\operatorname{Fuse}\;}
\left\{
r_t^{(u,j)}
\right\}_{j=1}^{k_t}
\xrightarrow{\;\operatorname{Pred}\;}
\widehat{\mathcal{P}}_t^{(u)}
}_{\substack{
\text{\textbf{Interaction computation}}\\
\text{request- / candidate-conditioned modeling capacity}
}}.
\]
The ranking and exposure process follows this predictive computation:
\[
\widehat{\mathcal{P}}_t^{(u)}
\longrightarrow
\left(
I_t^{(u)},\mathcal{S}_t^{(u)}
\right)
\longrightarrow
\mathcal{Y}_t^{(u)}.
\]

\begin{remarkbox}{Computational implication}
The computational consequence of the amortization boundary becomes
clear when expressed at the workload level. Let $\lambda_b$ denote the
rate of behavioral updates, $\lambda_q$ the rate of recommendation
requests, and $\bar K$ the mean candidate fanout per request. Let
$C_{\Delta}$ denote the cost of incrementally advancing behavioral state
and publishing a compatible reader snapshot, $C_{\mathrm{req}}$ the
request-level cost of resolving and preparing that shared state, and
$C_{\mathrm{cand}}$ the candidate-conditioned fusion and prediction
cost. The steady-state computation attributable to one request is then
approximately
\[
\bar C_{\mathrm{TransX}}
\approx
\frac{\lambda_b}{\lambda_q} C_{\Delta}
+
C_{\mathrm{req}}
+
\bar K C_{\mathrm{cand}}.
\]
By contrast, a request-synchronous design that reconstructs
history-conditioned state whenever a prediction request arrives incurs
approximately
\[
\bar C_{\mathrm{sync}}
\approx
C_{\mathrm{hist}}(l)
+
C_{\mathrm{req}}
+
\bar K C_{\mathrm{cand}},
\]
where $C_{\mathrm{hist}}(l)$ grows with the history processed on the
online serving path. Thus, cross-request amortization replaces history
computation paid at request rate with state maintenance paid at
behavior-update rate, while candidate-fanout amortization keeps shared
history-side computation outside the $\bar K$ multiplier. The resulting
advantage grows when behavioral state supports multiple subsequent
requests and when each request fans out to many candidate predictions,
subject to the additional storage, transfer, freshness, and recovery
costs required to maintain materialized state. These expressions are
intended as a workload decomposition rather than a complete serving-cost
model.
\end{remarkbox}

\section{Model Architecture}
\label{sec:model-arch}

The central architectural challenge is to expose computation at the granularity at which capacity can be scaled selectively and execution can be amortized, without sacrificing candidate-specific access to behavioral evidence. TransX addresses this through three coupled components: a causal sequence encoder that transforms the evolving behavioral history, a bounded behavior memory that controls the state exposed to prediction, and candidate-conditioned fusion that reads this shared state for each candidate. Figure~\ref{fig:architecture} illustrates the resulting architecture.

\begin{figure}[htb]
    \centering
    \includegraphics[width=0.85\linewidth]{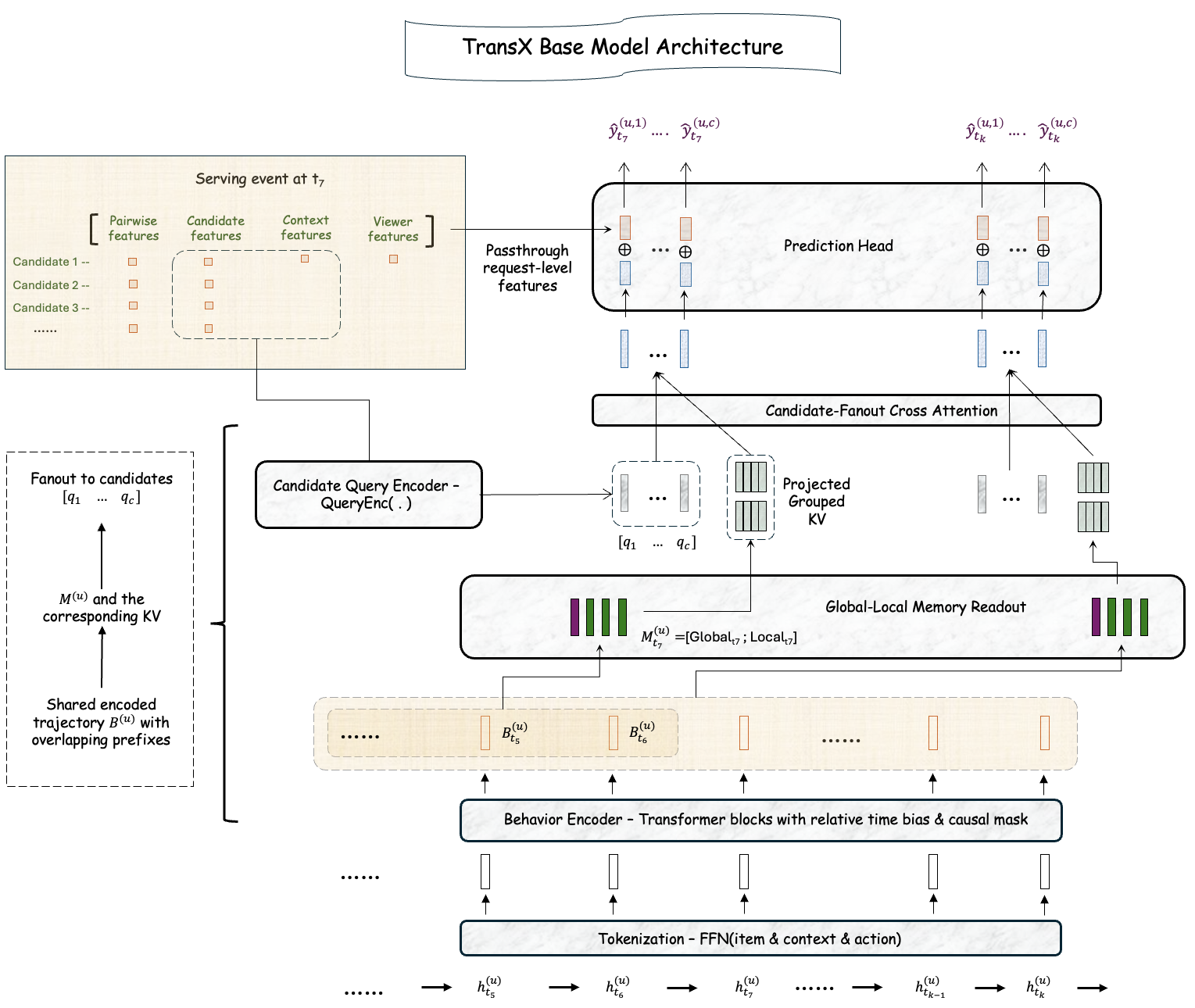}
    \caption{TransX's Encoder-Transducer architecture overview. A causal behavior encoder produces a shared trajectory representation reused across serving events. At each event, a cutoff-specific global–local memory is projected once into shared KV state, while candidate-specific queries independently cross-attend to that memory through candidate-fanout cross-attention. The fused candidate representations, together with passthrough prediction features, are mapped to per-candidate outputs.}
    \label{fig:architecture}
\end{figure}

\subsection{Sequence Tokenization and Encoding}
\label{sec:encoding}

For notational simplicity, we omit the superscript \(u\) when the user index is clear from context. The behavior encoder maps the available history
$H_{\le t}=(h_1,\ldots,h_{l_t})$
to contextualized behavior representations. Historical interaction
context is encoded with each past behavior, whereas current user,
request, and candidate context enters only through the serving-time
query and prediction paths.

\paragraph{Contextualized user behavior tokens.}
A behavior record represents an observed interaction together with
the context under which it occurred:
\begin{equation}
\label{eq:behavior-record}
h_i
=
\left(
o_i,\,
a_i,\,
c_i^{\mathrm{hist}},\,
v_i,\,
p_i,\,
\tilde t_i,\,
\tau_i^{\mathrm{avail}}
\right),
\end{equation}
where $o_i$ denotes the interacted entity and its 
attributes, $a_i$ the action type,
$c_i^{\mathrm{hist}}$ the historical interaction context,
$v_i$ numerical features, $p_i$ available pretrained
representations, $\tilde t_i$ the event timestamp, and
$\tau_i^{\mathrm{avail}}$ the time at which the record becomes
eligible for sequence encoding. Historical context may include,
for example, the interaction surface, entry point, exposure position,
or selected user--entity and user-state attributes associated with
the event.

A record becomes eligible once its tokenized fields are available;
subsequent corrections or enrichments require an explicit
representation refresh.

\paragraph{Shared feature embedding and tokenization.}
Categorical fields use field-specific embedding tables, with
within-field pooling for multi-valued attributes. Action type is
represented by a dedicated embedding. Numerical features are
transformed and normalized, while pretrained representations are
mapped through source-specific learned projections. Missing values
use explicit \texttt{missing} indicators.

We define a shared record-to-token mapping
\begin{equation}
\label{eq:shared-tokenization}
\operatorname{Tok}_{\phi}(x)
=
\operatorname{LN}
\left(
f_{\phi}
\left(
\operatorname{Features}_{\phi}(x)
\right)
\right)
\in\mathbb{R}^{d},
\end{equation}
where
\begin{equation}
\operatorname{Features}_{\phi}(h_i)
=
\left[
e_i^{\mathrm{ent}}
\Vert
e_i^{\mathrm{act}}
\Vert
e_i^{\mathrm{ctx}}
\Vert
\phi_{\mathrm{num}}(v_i)
\Vert
\phi_{\mathrm{pre}}(p_i)
\right].
\end{equation}
We then write
$z_i^{(0)}=\operatorname{Tok}_{\phi}(h_i)$.
The same tokenization interface and embedding tables are reused by
the candidate-query path for fields with matching namespaces and
semantics.

\paragraph{Temporal causal self-attention.}
The token sequence is contextualized using standard Transformer
blocks with multi-head causal self-attention, position-wise
feed-forward networks, residual connections, and layer
normalization~\citep{vaswani2017attention}. For attention head $r$, we use
\begin{equation}
\label{eq:behavior-temporal-attention}
s_{ij}^{(r)}
=
\frac{
(q_i^{(r)})^\top k_j^{(r)}
}{
\sqrt{d_h}
}
+
b_r(\tilde t_i-\tilde t_j)
+
m_{ij},
\end{equation}
where $d_h$ is the query/key dimension, $b_r(\cdot)$ is a learned relative-time bias~\citep{xu2019self} and
\begin{equation}
m_{ij}
=
\begin{cases}
0, & j\le i \text{ and } j \text{ is valid},\\
-\infty, & \text{otherwise}.
\end{cases}
\end{equation}
The causal order follows record availability. Because occurrence
timestamps need not be monotone under this ordering,
$b_r(\tilde t_i-\tilde t_j)$ is defined over signed time differences.
The causal mask enforces the information boundary and excludes
padding, while the relative-time term represents elapsed event time.

For $l_t>0$, the behavior sequence encoder produces
\begin{equation}
\label{eq:behavior-encoder-output}
B_t
=
\operatorname{Enc}
\left(
z_1^{(0)},\ldots,z_{l_t}^{(0)}
\right)
=
\left[
b_1,\ldots,b_{l_t}
\right]^\top
\in\mathbb{R}^{l_t\times d}.
\end{equation}
Under fixed tokenizer/encoder versions and append-consistent
availability ordering, appending later records leaves existing prefix
states unchanged, enabling their reuse across serving events.

\paragraph{Empty histories.}
When no valid behavior memory is available, the request is routed to
the designated non-sequential fallback model using serving-time
features only.

\subsection{Behavior Memory}
\label{sec:behavior-memory}

Exposing all $l_t$ states to every candidate reader would couple
computation on the online serving path, state transfer, and cache footprint directly
to history length. TransX therefore constructs a bounded,
candidate-independent \emph{behavior memory}
\[
M_t
=
\mathcal M\!\left(B_t\right)
\in
\mathbb R^{n_t\times d_m},
\qquad
n_t \ll l_t ,
\]
where $\mathcal M$ is the memory construction mechanism, $n_t$ denotes the number of behavior-memory states available at request $t$, and \(d_m\) is the behavior-memory state dimension. This separates the amount of behavioral history encoded from the number of memory states exposed to prediction.

\paragraph{Global--local memory.}
Our primary realization combines a high-resolution recent window with
lower-resolution summaries of older behavior. Let $W$ denote the local
window size. The local branch retains the most recent encoded states,
\[
L_t
=
B_t
[
l_t-\min(W,l_t)+1:l_t,\:
],
\]
while the preceding history is partitioned into causal supports and
compressed into a bounded set of global summary states,
\[
G_t
=
\mathrm{PoolBlocks}
\!\left(
B_t[1:l_t-W,\:]
\right).
\]
For notational simplicity, we describe the global branch in the main text as block-wise pooling over the history preceding the local window. Appendix~\ref{app:blocked-prefix-scan} details how the corresponding cutoff-specific block summaries are efficiently computed across events using a blocked prefix scan.
Pooling may be fixed or learned, and can be applied to chunks to form block-wise summaries. After projection to a common memory width,
\[
M_t
=
\psi_{\mathrm{mem}}
\left(
[
G_t;
L_t
]
\right).
\]
The reader therefore receives precise access to recent behavior together
with compressed coverage of longer-range history, while the number of
reader-visible states remains bounded.

This construction separates two quantities that are otherwise easily
conflated: \emph{history support}, which determines how far into the
behavioral trajectory information may be drawn, and \emph{memory
resolution}, which determines how much of the resulting representation
is exposed to the reader. A long behavioral history can therefore be
encoded without requiring a proportionally longer reader sequence.
Section~\ref{sec:scaling} studies this distinction explicitly when
scaling behavioral computation and behavior-memory resolution.

\paragraph{Causal eligibility.}
Every reader-visible memory state is constructed only from behavior
available at the corresponding prediction cutoff. For each memory slot
$M_{t,k}$,
\[
\operatorname{supp}(M_{t,k})
\subseteq
\mathcal \{1,\ldots,l_t\}.
\]
Thus neither local retention nor global pooling can expose future
observations to an earlier prediction.

\subsection{Candidate Query}
\label{sec:query}

Candidate queries reuse the behavior-token feature interface and
compatible embedding tables from Eq.~\ref{eq:shared-tokenization}, but
use a separate query encoder. For candidate $j$ at serving event $t$,
$x_{t,j}^{\mathrm{query}}$ contains the prediction-time candidate and
context features available before the response is observed. Its action
field is set to a dedicated \texttt{unknown-action} symbol, distinct
from all observed action types, so that query construction remains
compatible with historical tokenization without exposing the target
response.

The candidate representation is
\begin{equation}
\label{eq:query-encoding}
q_t^{(u,j)}
=
\operatorname{QueryEnc}
\left(
\operatorname{Tok}_{\phi}
(x_{t,j}^{\mathrm{query}})
\right)
\in\mathbb R^{d_q}, \quad x_{t,j}^{\mathrm{query}} \subseteq x_t^{(u,j)}
\end{equation}
where \(d_q\) is the candidate-query dimension.

For the $k_t$ candidates in the request, these independent query rows
are stacked as
\begin{equation}
\label{eq:candidate-query-matrix}
U_t
=
[
q_t^{(u,1)},\ldots,q_t^{(u,k_t)}
]^\top
\in\mathbb R^{k_t\times d_q}.
\end{equation}
The shared tokenizer aligns historical and candidate feature
representations, while the separate query encoder specializes
serving-time candidate conditioning (see Table~\ref{tab:behavior-query-correspondence} for comparison).

\begin{table}[htb]
\centering
\small
\setlength{\tabcolsep}{5pt}
\renewcommand{\arraystretch}{1.12}
\begin{tabularx}{\linewidth}{
  @{}p{0.22\linewidth}
     X
     X@{}
}
\toprule
\textbf{Field group}
&
\textbf{Historical behavior token}
&
\textbf{Candidate query}
\\
\midrule

Item features
&
Observed interacted item
&
Candidate $j$
\\

Context
&
Historical interaction context
&
Serving-time context
\\

Action
&
Observed action
&
\texttt{unknown-action}
\\

Embeddings / tokenizer
&
Shared
&
Shared
\\

Encoder
&
Behavior encoder
&
Query encoder
\\

\bottomrule
\end{tabularx}
\caption{
Correspondence between historical behavior tokens and candidate queries.
The two paths share the feature/tokenization interface and compatible
embedding tables, while using different action semantics and separate
encoders.
}
\label{tab:behavior-query-correspondence}
\end{table}

\subsection{Fanout and Iterative Query Fusion}
\label{sec:fanout-fusion}

The preceding stages produce a shared behavioral memory
$M_t\in\mathbb R^{n_t\times d_m}$ and candidate queries
$U_t\in\mathbb R^{k_t\times d_q}$.
TransX fuses them through candidate-fanout cross attention. The implementation uses grouped-query attention (GQA), where all candidates share behavioral keys and values, while multiple query heads additionally share each KV head~\citep{ainslie2023gqa}.

\paragraph{Grouped-query cross-attention fusion.}
Let $H_q$ denote the number of query heads and $G$ the number of
key--value heads, where $G$ divides $H_q$. Define
\[
\gamma:
\{1,\ldots,H_q\}
\rightarrow
\{1,\ldots,G\}
\]
to assign each query head to a key--value group. With equal-sized
groups, each key--value head is shared by $H_q/G$ query heads.
The cases $G=1$ and $G=H_q$ recover multi-query attention and
standard multi-head attention, respectively.

\paragraph{Iterative memory read.}
TransX can stack $L_R$ cross-attention reader blocks to perform iterative
candidate-conditioned reads over the same behavioral memory.
Let $U_t^{(0)}=U_t$. At reader block $\ell$, the
GQA operations  are applied with $U_t^{(\ell-1)}$ as the query state and
layer-specific projections, while $M_t$ remains fixed. The
resulting readout is integrated as
\begin{equation}
\begin{aligned}
O_t^{(\ell)}
&=
\operatorname{Concat}
(O_{t,1}^{(\ell)},\ldots,O_{t,H_q}^{(\ell)})
W_O^{(\ell)},\\
\widetilde{U}_t^{(\ell)}
&=
\operatorname{LN}
\left(
U_t^{(\ell-1)}+O_t^{(\ell)}
\right),\\
U_t^{(\ell)}
&=
\operatorname{LN}
\left(
\widetilde{U}_t^{(\ell)}
+
\operatorname{FFN}^{(\ell)}
(\widetilde{U}_t^{(\ell)})
\right),
\end{aligned}
\end{equation}
with final candidate-conditioned representation
$R_t=U_t^{(L_R)}$.
Because each block updates the candidate state before the next read,
later blocks can re-attend to $M_t$ conditioned on evidence
retrieved by earlier blocks. Reader depth therefore increases
candidate--history interaction capacity through iterative memory access
without modifying or replicating the shared behavioral memory.

\paragraph{Candidate-fanout amortization.}
At serving event $t$, history-side computation is performed once and
shared across the $k_t$ candidate queries, while candidate-conditioned
fusion and prediction scale with fanout:
\[
C_t
=
C_{\mathrm{shared}}(l_t,n_t)
+
k_t C_{\mathrm{cand}}(n_t).
\]
Hence the shared history-side cost contributes
$C_{\mathrm{shared}}/k_t$ per candidate, while candidate-specific
interaction remains fully expressive. GQA provides an additional,
orthogonal reduction by sharing projected behavioral KV across query
heads.

\subsection{Parallel Response Prediction}
\label{sec:response-prediction}

The response prediction head combines each candidate-conditioned representation
$r_{t}^{(u,j)}$ with prediction-time user, request, candidate, and user-candidate pairwise
features to produce response probabilities:
\begin{equation}
\label{eq:response-prediction}
\begin{aligned}
h_{t,j}^{(u)}
&=
\left[
r_{t}^{(u,j)}
\,\Vert\,
\phi_{\mathrm{pred}}
\!\left(x_t^{(u,j)}\right)
\right],\\
\hat{p}_{t,j}^{(u)}
&=
f_{\mathrm{out}}
\!\left(
\operatorname{Interact}_{\vartheta}
\!\left(h_{t,j}^{(u)}\right)
\right)
\in\mathbb R^{|\mathsf Y|}.
\end{aligned}
\end{equation}
Here \(\phi_{\rm pred}\) embeds serving-time features, while \(\operatorname{Interact}_\vartheta\) combines them with retrieved behavioral evidence through feature-crossing and candidate-wise prediction modules. In practice, prediction-head capacity can be scaled independently through richer feature interactions without enlarging behavior memory. The formulation above scores candidates in parallel, without candidate-to-candidate attention or autoregressive dependence; Section~\ref{sec:generative-specialization} relaxes this constraint for generative recommendation. For binary response tasks, \(f_{\rm out}\) maps the output logits through a sigmoid.

The prediction head $$\text{PredHead}(r, x):= f_{\mathrm{out}}
\!\left(
\operatorname{Interact}_{\vartheta}
\left[
r
\,\Vert\,
\phi_{\mathrm{pred}}
\!\left(x\right)
\right]
\right)$$ 
is applied row-wise to the $k_t$ candidates:
\begin{equation}
\label{eq:parallel-response-prediction}
\widehat{P}_t
=
\operatorname{PredHead}
\!\left(R_t,X_t\right)
\in\mathbb R^{k_t\times|\mathsf{Y}|},
\end{equation}
where $X_t$ contains the feature records aligned with
$R_t$. 


\section{Core Training}
\label{sec:core-training}

TransX is first trained as a response-prediction model over logged
serving events. The core objective predicts click and conversion responses given the current request and user \(u\)’s behavioral history, e.g.
\[
\left(
H_{\le t}^{(u)}, e_t^{(u,j)}
\right)
\longmapsto
p_{\theta}
\left(
y_t^{(u,j)} = 1
\mid
H_{\le t}^{(u)}, e_t^{(u,j)}
\right).
\]
This stage learns the behavioral encoder, memory
interface, candidate-conditioned reader, and prediction head that are
subsequently reused by the specializations in
Section~\ref{sec:post-training}.

\subsection{Packed Training Unit}
\label{sec:event-training}

Recommendation logs are typically materialized with one record per
exposed candidate, with candidates from the same serving event sharing the
same user, prediction-time contexts, and eligible behavioral history. We
therefore treat the retained behavior trajectory and associated serving-event stream, rather than an individual candidate row or event group, as the basic training unit.

For user $u$, let
\[
\mathcal H^{(u)}
=
(h^{(u)}_1,\ldots,h^{(u)}_{L_u})
\]
denotes user \(u\)’s packed behavioral trajectory, where \(L_u\) denotes the number of retained behavior records in the packed trajectory. Serving event $t$ contains
the exposed candidate set and mature responses
\[
\mathcal E_t^{(u)}
=
\left\{
\left(
e_{t,j}^{(u)},
y_{t,j}^{(u)}
\right)
\right\}_{j=1}^{m_t},
\]
together with behavioral cutoff $l_t$, indicating the prefix available
when the event was scored. A packed training unit is therefore
\[
\mathcal D^{(u)}
=
\left(
\mathcal H^{(u)},
\left\{
l_t,\mathcal E_t^{(u)}
\right\}_{t=1}^{T_u}
\right).
\label{eq:packed-training-example}
\]

All candidates within event $t$ share the same causal behavioral state
\[
M_t^{(u)}
=
\mathcal M
\left(
\operatorname{Enc}
(\mathcal H^{(u)}_{1:l_t})
\right),
\]
but remain distinct prediction targets. Likewise, successive events from
the same user generally observe overlapping prefixes of the same
trajectory. The packed representation therefore preserves the two
amortization structures identified in Section~\ref{sec:framing}:
shared behavioral computation across event cutoffs and shared event
memory across candidate targets.

\subsection{Multi-Prefix Seq-to-Seq Training}
\label{sec:multi-cutoff-training}

\begin{figure*}[htb]
\centering
\begingroup
\resizebox{0.7\linewidth}{!}{%
\begin{tikzpicture}[
  x=1cm, y=1cm,
  font=\footnotesize,
  >=Latex,
  box/.style={
    draw=black!45,
    rounded corners=2pt,
    align=center,
    text width=2.55cm,
    minimum width=2.9cm,
    minimum height=8mm,
    inner sep=2.5pt
  },
  memory/.style={
    box,
    draw=teal!65!black,
    fill=teal!5
  },
  query/.style={
    box,
    draw=violet!60!black,
    fill=violet!5
  },
  prediction/.style={
    box,
    draw=blue!55!black,
    fill=blue!4
  },
  flow/.style={
    ->,
    semithick,
    draw=black!65
  },
  shared/.style={
    semithick,
    dashed,
    draw=teal!65!black
  }
]

\node[
  box,
  fill=black!3,
  minimum height=11mm
] (logs) at (1.5,0) {
  \textbf{Serving-event logs}\\
  $(u,t,j,x,y)$
};

\node[
  box,
  text width=6.2cm,
  minimum width=6.6cm,
  minimum height=11mm,
  fill=black!3
] (packing) at (7.6,0) {
  \textbf{Packed Training Unit $\mathcal D^{(u)}$}\\
  Group by $(u,t)$; route candidate features by event
};

\draw[flow] (logs.east) -- (packing.west);

\node[font=\footnotesize\bfseries] at (1.5,-1.45) {Sequence encoding};
\node[font=\footnotesize\bfseries] at (4.8,-1.45) {Prefix memory};
\node[font=\footnotesize\bfseries] at (8.2,-1.45) {Candidate fanout};
\node[font=\footnotesize\bfseries] at (11.75,-1.45) {Response Prediction};

\node[memory] (history) at (1.5,-2.7) {
  Behavior trajectory\\
  $\mathcal H^{(u)}$
};

\node[
  memory,
  minimum height=13mm
] (encoder) at (1.5,-4.25) {
  \textbf{Causal encoder}\\
  $B^{(u)}
   = \operatorname{Enc}(\mathcal H^{(u)})$
};

\draw[flow] (history.south) -- (encoder.north);

\node[
  align=center,
  text width=2.8cm,
  font=\scriptsize,
  text=teal!55!black
] at (1.5,-5.35) {
  shared across prefixes
};

\draw[flow]
  (packing.south)
  --
  node[right, font=\scriptsize]{event-level features}
  (7.6,-1.2);

\foreach \suffix/\event/\yy in {
  one/{t_1}/-2.7,
  two/{t_2}/-4.65
} {
  \node[
    memory,
    font=\scriptsize,
    text width=2.45cm,
    minimum width=2.75cm
  ] (mem-\suffix) at (4.8,\yy) {
    prefix $l_{\event}$\\
    $M_{\event}^{(u)}$
  };

  \node[
    query,
    font=\scriptsize,
    text width=2.9cm,
    minimum width=3.2cm
  ] (query-\suffix) at (8.2,\yy) {
    \textbf{Event $\event$ candidates}\\[-1pt]
    $e_{\event}^{(u,1)},\ldots,
      e_{\event}^{(u,m_{\event})}$
  };

  \node[
    prediction,
    font=\scriptsize,
    text width=2.9cm,
    minimum width=3.2cm
  ] (pred-\suffix) at (11.75,\yy) {
    predictions\\[-1pt]
    $\widehat p_{\event}^{(u,1)},\ldots,
      \widehat p_{\event}^{(u,m_{\event})}$
  };

  \draw[flow] (mem-\suffix.east) -- (query-\suffix.west);
  \draw[flow] (query-\suffix.east) -- (pred-\suffix.west);
}

\draw[semithick, draw=teal!65!black]
  (encoder.east) -- (3.2,-4.25);

\draw[shared]
  (3.2,-2.7) -- (3.2,-4.65);

\foreach \suffix in {one,two} {
  \draw[shared,->]
    (3.2,0 |- mem-\suffix.west)
    --
    (mem-\suffix.west);
}

\node[
  draw=teal!65!black,
  fill=teal!3,
  rounded corners=2pt,
  font=\scriptsize,
  inner sep=4pt, minimum width=168.84073pt
] (materialization) at (6.65,-6.25) {
  $\displaystyle
  M_t^{(u)}
  =
  \psi_{\mathrm{mem}}\!\left(
  \operatorname{Concat}_{\mathrm{slots}}
  \bigl(
  \mathrm{Global}_t^{(u)},
  \mathrm{Local}_t^{(u)}
  \bigr)
  \right)$
};

\end{tikzpicture}%
}
\endgroup

\caption{
Multi-prefix training in TransX. A shared causal encoding of each behavior
trajectory produces prefix memories that fan out to all candidates at the
corresponding serving event; each memory combines global and local behavioral
state.
}
\label{fig:seq2seq-training}
\end{figure*}
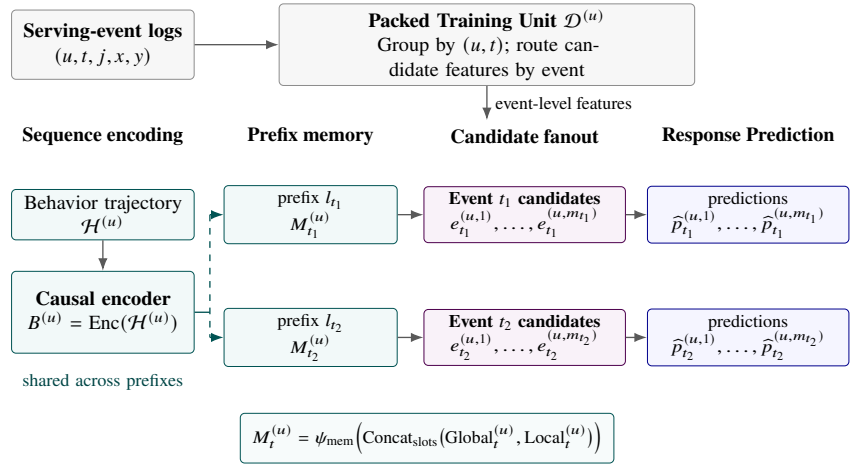

Let the training examples be constructed around a common retained-history
boundary where all supervised events associated with that packed
unit share the same left boundary and differ only in their
right-hand cutoff (see Figure~\ref{fig:seq2seq-training}). Because the behavior encoder is causal, one encoding of the packed
trajectory supplies the states required by all selected event cutoffs:
\[
B^{(u)}
=
\operatorname{Enc}(\mathcal H^{(u)}),
\qquad
B_{\le t}^{(u)}
=
B^{(u)}[1:l_t,:].
\]
Each event constructs its cutoff-valid behavioral memory $M_t^{(u)}$,
and all candidate queries associated with that event are evaluated in
parallel (by fusing with the common memory):
\[
R_t^{(u)}
=
\operatorname{Fuse}
\left(
M_t^{(u)}, U_t^{(u)}
\right).
\]
Thus multiple event losses backpropagate through one shared causal
trajectory, while candidate losses within each event first aggregate
through their common memory state. Section~\ref{sec:training-system}
describes the batched cutoff extraction, prefix-scan construction, and
sequence--fanout execution used to realize this logical computation
efficiently.

Let $\mathcal V_t$ denote the candidate positions with valid,
mature response label. Passing $R_t^{(u)}$ and prediction features $X_t^{(u)}$ through the response prediction head $\operatorname{PredHead}
\!\left(R_t^{(u)},X_t^{(u)}\right)$ produces
\begin{equation}
\widehat{\mathcal{P}}_t^{(u)}
=
\left(
\widehat p_t^{(u,1)},\ldots,
\widehat p_t^{(u,m_t)}
\right).
\end{equation}
For binary response-prediction tasks, the per-user loss is:
\begin{equation}
\label{eq:ctr-seq2seq-loss}
\mathcal{L}^{(u)}
=
\sum_{t=1}^{T_u}
\sum_{j\in\mathcal V_t}
\ell_{\mathrm{BCE}}
\left(
\widehat p_t^{(u,j)},
y_t^{(u,j)}
\right).
\end{equation}
For brevity, Eq.~\ref{eq:ctr-seq2seq-loss} omits normalization by the
number of valid supervised targets and loss re-weighting terms. 
The same training setup is reused by the post-training objectives in Section~\ref{sec:post-training}, with task-specific reductions applied over the shared behavioral trajectory and event structure.

\subsection{Dense--Sparse Hybrid Optimization}
\label{sec:hybrid-optimization}

TransX contains two distinct parameter regimes: densely updated network
parameters and large row-sparse ID embedding tables. We optimize dense
parameters with AdamW~\citep{loshchilov2019decoupled} and sparse embeddings with row-wise Adagrad~\citep{duchi2011adaptive,fbgemm2026rowwise}. For
an embedding table with $n$ rows, row-wise Adagrad maintains one
second-moment accumulator per row rather than one per coordinate,
reducing optimizer state from $O(nd)$ to $O(n)$ while preserving the
within-row gradient direction.

Because behavioral trajectories, event memories, and candidate branches
are shared within one logical training step, all gradient contributions
to the same dense parameter or embedding row are aggregated before the
corresponding nonlinear optimizer update. This includes contributions
from overlapping behavioral prefixes, candidate fanout, and gradient
accumulation. The resulting optimizer semantics are therefore defined at
the logical-step boundary and remain unchanged by the physical
partitioning introduced in Section~\ref{sec:training-system}.
Detailed row-wise update equations and distributed ownership mechanics
are deferred to Appendix~\ref{app:optimizer}.

\section{Model--Infrastructure Co-Design}
\label{sec:general-infra}

The compute-state decomposition is useful only if its reuse structure is
preserved in execution. TransX therefore maps the same hierarchy onto
training and serving. In particular, training executes
\[
\underline{\text{trajectory}}
\;\longrightarrow\;
\underline{\text{memory}}
\;\longrightarrow\;
\underline{\text{candidate fanout}},
\]
without replicating shared computation, while serving maintains
\[
\underline{\text{mutable encoder continuation state}}
\;\rightarrow\;
\underline{\text{immutable behavior memory KV}}
\;\rightarrow\;
\underline{\text{request-local computation}}.
\]
These components operate on different clocks and under different objectives: continuation state advances as new behavior arrives, published behavior memory changes only when refreshed state is committed, and candidate-conditioned computation is triggered by each request. The infrastructure therefore schedules, places, and scales these components independently while preserving the causal information boundary and prediction semantics of the model.

\subsection{Execution Model}
\label{sec:execution-model}

Table~\ref{tab:execution-model} summarizes the execution objects induced by TransX model design, their multiplicities, and their amortization scopes. The architecture determines what state and computation can be shared; the infrastructure determines how they are partitioned, placed, transferred, and cached. We use \(\tau\) to index the represented-behavior cutoff of a serving-state version, distinguishing serving-state versions from the request index \(t\).

\begin{table}[htb]
\centering
\footnotesize
\setlength{\tabcolsep}{4pt}
\renewcommand{\arraystretch}{1.12}
\begin{tabularx}{\linewidth}{
  @{}p{0.28\linewidth}
     p{0.16\linewidth}
     p{0.25\linewidth}
     X@{}
}
\toprule
\textbf{Logical object}
&
\textbf{Multiplicity}
&
\textbf{Amortization scope}
&
\textbf{Primary system implication}
\\
\midrule

Encoded trajectory $B^{(u)}$
&
Per packed trajectory
&
Multiple supervised event cutoffs
&
Sequence throughput and activation memory
\\

Event memory / projected KV
&
Per event cutoff
&
Candidate rows and sharing query heads
&
Fanout efficiency and communication
\\

Candidate-conditioned computation
&
Per candidate
&
No cross-candidate dependence
&
Dense batching and GPU utilization
\\

Encoder continuation state $\mathsf C_{\tau}^{(u)}$
&
Per represented behavior version
&
Successive behavior updates
&
Freshness, restoration, and replay
\\

Published behavior-memory KV $\mathsf V_{\tau}^{(u)}$
&
Per published behavior version
&
Requests and candidate batches
&
Storage, transfer, and online latency
\\

Request-local state
&
Per request
&
Candidate execution batches
&
SLO enforcement and temporary memory
\\

\bottomrule
\end{tabularx}
\caption{
Logical execution objects in TransX. The architecture determines which
computation and state can be reused; the infrastructure determines how
these objects are partitioned, placed, transferred, and cached.
}
\label{tab:execution-model}
\label{tab:training-infra-motivation}
\end{table}


Execution must preserve two contracts. During training, each event reads only causally valid behavioral state, candidates share one event memory, and fanout gradients rejoin before entering the shared trajectory graph.
During serving, continuation state remains maintenance-only, published behavior-memory KV is immutable and version-compatible, and state recovery remains outside the online serving path.

\subsection{Training Infrastructure}
\label{sec:training-system}

Section~\ref{sec:core-training} defines training logically as a
multi-prefix computation over packed behavioral trajectories: multiple
serving events read different causal prefixes of the same trajectory,
and multiple candidate losses within each event read the same
event-specific memory. The training system must realize this objective
without changing its semantics while avoiding the duplicated work that
would arise from materializing each $(u,t,j)$ target independently.

The implementation challenge is therefore to construct the
$\{M^{(u)}_t\}_{t=1}^{T_u}$ memories and their candidate fanout from one
encoded trajectory without materializing the
$T_u$ overlapping prefixes separately. Section~\ref{sec:training-memory-extraction}
addresses this with batched local gathers and shared global prefix
scans; Section~\ref{sec:parallelism-distributed-execution} then maps the
resulting trajectory--event--candidate graph onto distributed
execution.

\subsubsection{Batched Local Gather and Global Prefix Scan}
\label{sec:training-memory-extraction}

Section~\ref{sec:multi-cutoff-training} defines the logical computation
in which multiple serving events read different cutoff-valid memories
from one causal behavioral trajectory. We now describe how those
memories are materialized without reconstructing each eligible prefix
independently.

For a packed trajectory from user $u$, let $B^{(u)}
=
\operatorname{Enc}
\left(
\mathcal H^{(u)}_{1:L_u}
\right)
\in
\mathbb R^{L_u\times d}$
be the causal encoding up to the largest selected cutoff, and let the event cutoffs be $(l_{1}^{(u)},\ldots,l_{T_u}^{(u)})$. Let \(W\) denote the local-memory width and \(R\) the maximum number of global summary states retained in the behavior memory.

\paragraph{Batched local gather.}
For a local-memory width $W$, define the source index for local position
$w$ of event $t$ as
\[
I^{(u)}_{t,w}
=
l_{t}^{(u)}-W+w,
\qquad
w=1,\ldots,W,
\]
with validity mask
\[
A^{(u)}_{t,w}
=
\mathbf 1
\left[
1
\le
I^{(u)}_{t,w}
\le
l_{t}^{(u)}
\right].
\]
A single masked gather over $B^{(u)}$ then constructs the fixed-shape local memory
tensor
\[
L^{(u)}_{t,w,:}
=
A^{(u)}_{t,w}\,
B^{(u)}_{
\operatorname{clip}(I^{(u)}_{t,w},1,L_u),:
},
\]
or, collectively,
\[
L^{(u)}
=
\operatorname{MaskedGather}
\left(
B^{(u)},I^{(u)},A^{(u)}
\right)
\in
\mathbb R^{T_u\times W\times d}.
\label{eq:batched-local-gather}
\]
Invalid positions are masked before memory construction. Thus all local
windows are extracted in one batched operation while overlapping
windows continue to reference the same encoded trajectory states.

This moves only $O(T_u W d)$
states into event-local memory, rather than materializing $O\!\left(
d\sum_{t=1}^{T_u}l_{t}^{(u)}
\right)$
states as would occur if every behavioral prefix were reconstructed
independently.

\paragraph{Global prefix scan.}
Global memory is constructed from sufficient states that compose
associatively along the behavioral trajectory. For any contiguous
segment $\mathcal I$, let $S_{\mathcal I}$ denote its sufficient state.
For adjacent segments $\mathcal I_1$ and $\mathcal I_2$, we require
\[
S_{\mathcal I_1\cup\mathcal I_2}
=
S_{\mathcal I_1}\otimes S_{\mathcal I_2},
\]
where $\otimes$ is associative and has identity $e$. Intuitively, \(S_{\mathcal I}\) is a mergeable summary of one contiguous segment; event-specific global blocks are obtained by merging the segment states that cover their supports. In particular,
writing $S_{\{i\}}$ for the singleton state contributed by encoded
behavior $b_i$, a parallel prefix scan~\citep{blelloch1990prefix}
computes
\[
\bigl(
S_{[1:1]},\ldots,S_{[1:L_u]}
\bigr)
=
\operatorname{PrefixScan}_{\otimes}
\bigl(
S_{\{1\}},\ldots,S_{\{L_u\}}
\bigr).
\]

For event $t$, let
$\{\mathcal I_{t,r}^{(u)}\}_{r=1}^{R_t^{(u)}}$
denote the contiguous supports of its global-memory blocks. The
corresponding global summaries are
\[
g_{t,r}^{(u)}
=
\operatorname{Readout}
\left(
S_{\mathcal I_{t,r}^{(u)}}
\right),
\qquad
r=1,\ldots,R_t^{(u)},
\]
and the global memory bank is
\[
G_t^{(u)}
=
\left[
g_{t,1}^{(u)},
\ldots,
g_{t,R_t^{(u)}}^{(u)}
\right]^\top .
\]
Because the sufficient states compose associatively, the statistics
needed by multiple event cutoffs can be shared rather than recomputed
independently. Appendix~\ref{app:blocked-prefix-scan} gives the exact
blocked realization used to construct these segment states across
overlapping event prefixes.

For example, the learned scalar-attention pooler used in our
global--local memory admits such an associative sufficient state.
Given scores $a_i=a_\theta(b_i)$, 
define, for any contiguous segment $\mathcal I$,
\[
m_{\mathcal I}
=
\max_{i\in\mathcal I} a_i,
\qquad
Z_{\mathcal I}
=
\sum_{i\in\mathcal I}
\exp(a_i-m_{\mathcal I}),
\]
and
\[
V_{\mathcal I}
=
\sum_{i\in\mathcal I}
\exp(a_i-m_{\mathcal I})\,b_i.
\]
Then the sufficient state is $S_{\mathcal I}
=
\left(
m_{\mathcal I},
Z_{\mathcal I},
V_{\mathcal I}
\right),$
with readout
\[
\operatorname{Readout}(S_{\mathcal I})
=
\frac{V_{\mathcal I}}{Z_{\mathcal I}}.
\]
For adjacent segments, these states combine exactly through the same
log-sum-exp rescaling~\citep{dao2022flashattention}, making
$\otimes$ associative and allowing the required block summaries to be
constructed by parallel scan and reduction without changing the exact
pooled result.

\paragraph{Event-memory materialization.}
Let \(n_{t}^{(u)}\) denote the number of reader-visible behavior-memory states for user \(u\) at event \(t\). The behavior memory for event $t$ is assembled from the gathered local
states and global summary bank:
\[
M_{t}^{(u)}
=
\operatorname{Compose}
\left(
G_{t}^{(u)},L_{t}^{(u)}
\right),
\qquad
n_{t}^{(u)}\le W+R,
\]
followed once per event by the memory adapter and shared KV projection,
\[
(K_{t}^{(u)},V_{t}^{(u)})
=
\operatorname{Proj}_{KV}
\left(
\operatorname{Adapter}(M_{t}^{(u)})
\right).
\label{eq:event-kv-materialization}
\]
All candidate rows associated with event $t$ then consume this common
projected state. The resulting execution follows the logical structure
of Section~\ref{sec:multi-cutoff-training}: encode the trajectory once,
extract many cutoff memories, and fan out only after compact event state
has been constructed.

In practice, these optimizations increase supervised-candidate throughput (candidate targets/s) by $1.84\times$ and encoded-behavior throughput (behavior tokens/s) by $1.4\times$, while raising average forward-pass MFU\footnote{
MFU is $F_{\mathrm{model}}/(T_{\mathrm{wall}}N_{\mathrm{GPU}}P_{\mathrm{peak}})$,
measured over the full post-warmup interval using the H200 dense BF16
Tensor Core peak (without sparsity uplift). Model FLOPs count dense
operations at two FLOPs per multiply--accumulate, including forward
and backward passes for training and forward passes only for inference.
Embedding lookups, communication, optimizer updates, and data movement
are excluded from model FLOPs but remain reflected in wall-clock time.
} from \(41\%\) to \(57\%\) in the prefix-execution benchmark relative to event-materialized prefix encoding on our H200 GPU fleet. For our production trajectory-length distribution, repeated-prefix materialization
$\approx 5.2$\,GiB/device versus batched cutoff extraction
$\approx 2.9$\,GiB/device peak activation HBM in our benchmarking.

\begin{figure}[htb]
    \centering
    \includegraphics[width=0.8\linewidth]{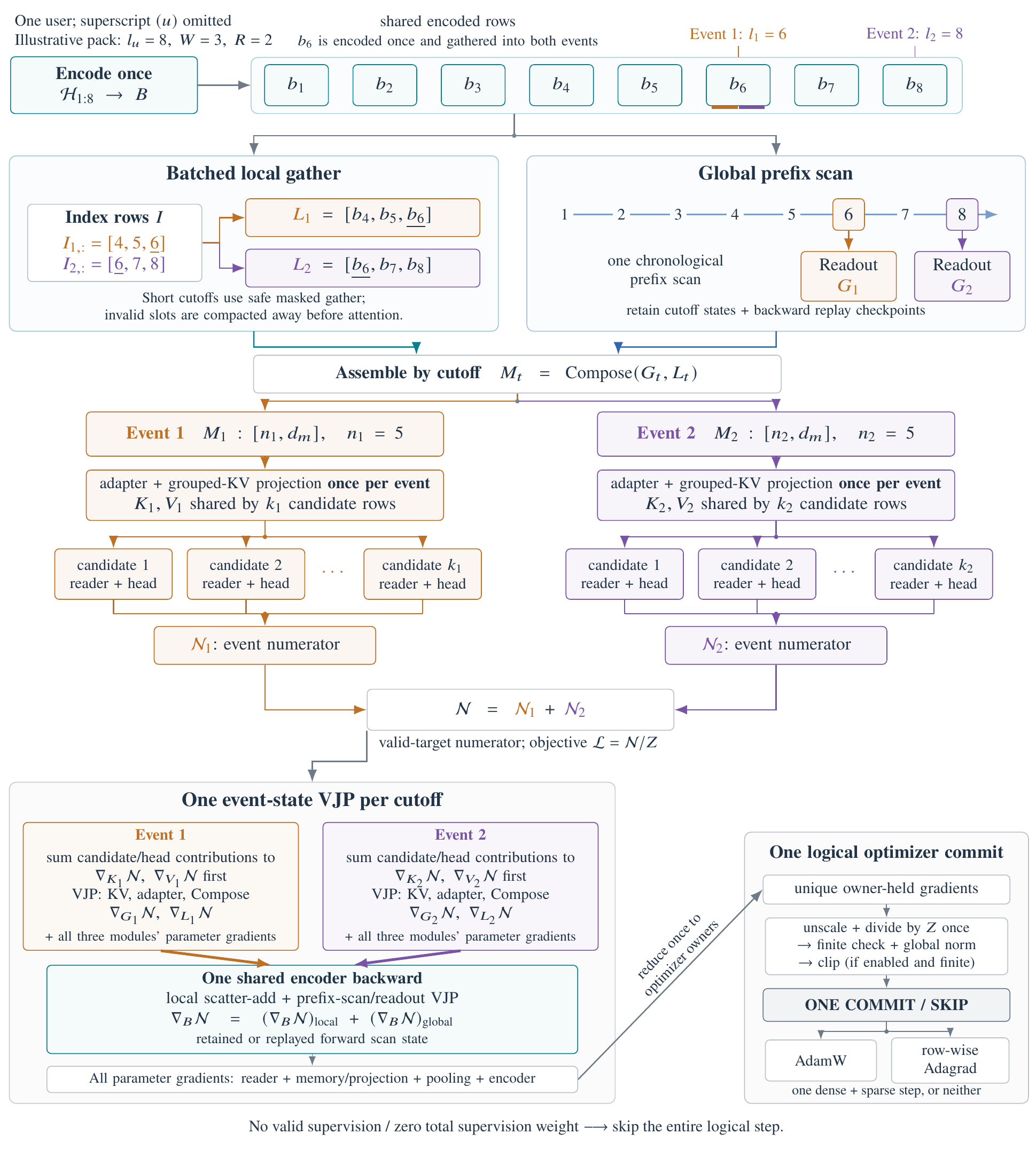}
    \caption{Illustration of the design principle: encode once; extract
    many cutoff memories; fan out candidates; backpropagate jointly.}
    \label{fig:training-infra}
\end{figure}

\paragraph{Backward aggregation.}
For a scalar loss $\mathcal L$, we write
$\nabla_X \mathcal L$ for the reverse-mode gradient with respect to an
intermediate tensor $X$.
Let
\[
\nabla_{K_{t}^{(u)}}\mathcal L^{(r)},\qquad
\nabla_{V_{t}^{(u)}}\mathcal L^{(r)}
\]
denote the KV gradients contributed by candidate/fanout partition $r$.
Contributions from all candidate rows and sharing query heads are summed at
the event boundary:
\[
\nabla_{K_{t}^{(u)}}\mathcal L
=
\sum_r
\nabla_{K_{t}^{(u)}} \mathcal L^{(r)},
\qquad
\nabla_{V_{t}^{(u)}}\mathcal L
=
\sum_r
\nabla_{V_{t}^{(u)}} \mathcal L^{(r)}.
\]

One reverse-mode VJP through the shared KV projection and memory adapter
produces the gradient of the event memory,
\[
\nabla_{M_{t}^{(u)}}\mathcal L
=
\operatorname{VJP}_{\operatorname{Adapter}+\operatorname{Proj}_{KV}}
\!\left(
M_{t}^{(u)};
\nabla_{K_{t}^{(u)}}\mathcal L,
\nabla_{V_{t}^{(u)}}\mathcal L
\right).
\]

Backpropagating through event-memory composition then yields gradients for
the local and global branches,
\[
\left(
\nabla_{G_{t}^{(u)}}\mathcal L,
\nabla_{L_{t}^{(u)}}\mathcal L
\right)
=
\operatorname{VJP}_{\operatorname{Compose}}
\!\left(
G_{t}^{(u)},L_{t}^{(u)};
\nabla_{M_{t}^{(u)}}\mathcal L
\right)
\]

Local-memory gradients are accumulated into the shared encoded trajectory
by the adjoint of the masked gather,
\[
\left.\nabla_{B^{(u)}}\mathcal L\right|_{\mathrm{local}}
=
\operatorname{ScatterAdd}
\!\left(
\{\nabla_{L_{t}^{(u)}}\mathcal L\}_{t=1}^{T_u},
I^{(u)};A^{(u)}
\right).
\]

Likewise, global-memory gradients are propagated through the blocked
prefix-scan and block-readout adjoint,
\[
\left.\nabla_{B^{(u)}}\mathcal L\right|_{\mathrm{global}}
=
\operatorname{VJP}_{\mathrm{BlockScan+Readout}}
\!\left(
B^{(u)};
\{\nabla_{G_{t}^{(u)}}\mathcal L\}_{t=1}^{T_u}
\right).
\]

The shared encoder therefore receives
\[
\nabla_{B^{(u)}}\mathcal L
=
\left.\nabla_{B^{(u)}}\mathcal L\right|_{\mathrm{local}}
+
\left.\nabla_{B^{(u)}}\mathcal L\right|_{\mathrm{global}}.
\]
after all selected event gradients have been accumulated.

\subsubsection{Sequence--Fanout Parallelism}
\label{sec:parallelism-distributed-execution}

Section~\ref{sec:training-memory-extraction} removes redundant
prefix materialization within a packed trajectory. The remaining systems
problem is to distribute this shared computation without turning the
trajectory, event, and candidate multiplicities into replicated encoder
work. In particular, sequence-side parallelism partitions the shared causal trajectory, whereas fanout parallelism partitions candidate-conditioned computation only after compact event state has been constructed.

\paragraph{Sequence--fanout decomposition.}
For sufficiently long trajectories, a sequence team of degree $d_s$
partitions the causal encoder over behavioral positions. After the
cutoff-specific local and global states of event $t$ have been extracted,
one logical event owner assembles the behavior memory and applies the
memory adapter and grouped-KV projection:
\[
B^{(u)}
\xrightarrow[\;d_s\;]{\text{sequence partition}}
M_{t}^{(u)}
\xrightarrow{\operatorname{Adapter}+\operatorname{Proj}_{KV}}
(K_{t}^{(u)},V_{t}^{(u)}).
\]
The event's candidate rows are then partitioned across a fanout team of
degree $d_f$,
\[
(K_t^{(u)},V_t^{(u)}),\,
\{e_t^{(u,j)}\}_{j=1}^{m_t}
\xrightarrow[\;d_f\;]{\text{candidate partition}}
\{\hat p_t^{(u,j)}\}_{j=1}^{m_t}.
\]
Thus increasing $d_f$ does not instantiate another behavioral encoder
or repeat memory construction. The degrees $d_s$ and $d_f$ describe
logical partitioning dimensions and need not correspond to disjoint
hardware pools.

This separation is visible empirically. Increasing event fanout from
64 to 512 candidates changes history-side encoder FLOPs per event by
less than $1\%$, while candidate throughput increases by approximately
$4\times$ before behavior memory or communication saturation. The scaling
behavior therefore reflects additional candidate-side parallelism
rather than replicated history computation.

\paragraph{Event state as the communication boundary.}
The event owner forms the boundary between sequence-side computation and
candidate fanout. Sequence ranks provide only the local states and
global scan outputs required for event $t$; the owner assembles
$M_{t}^{(u)}$ and projects its compact KV once. Candidate workers then
consume this common projected state for their assigned candidate tiles.
Communication across the boundary therefore operates on event-level KV
rather than independently materialized behavioral prefixes.

Backward follows the same ownership structure in reverse. Now we use
$\nabla_{K_{t}^{(u)}}\mathcal L^{(r)}$ and $\nabla_{V_{t}^{(u)}}\mathcal L^{(r)}$ to denote the KV gradients produced
by fanout worker $r$. After reducing them to $\nabla_{K_{t}^{(u)}}\mathcal L
=
\sum_{r=1}^{d_f}
\nabla_{K_{t}^{(u)}}\mathcal L^{(r)},\,
\nabla_{V_{t}^{(u)}}\mathcal L
=
\sum_{r=1}^{d_f}
\nabla_{V_{t}^{(u)}}\mathcal L^{(r)}$ at the event boundary,
the event owner applies the two-stage reverse calculation in the previous section to obtain $\nabla_{M_{t}^{(u)}}\mathcal L$ and then
$(\nabla_{G_{t}^{(u)}}\mathcal L,\, \nabla_{L_{t}^{(u)}}\mathcal L)$.

Afterwards, the local and global gradients re-enter the shared
trajectory through the scatter-add and prefix-scan adjoints of
Section~\ref{sec:training-system}. Hence candidate
multiplicity disappears before gradients return to the shared
trajectory graph.

\paragraph{Fanout-aware execution and communication overlap.}
Within each event, candidate workers batch candidate--head rows that
share the same grouped KV. Native fused GQA avoids explicit
candidate-wise KV replication; where beneficial, an event-level KV tile
remains resident while multiple candidate tiles are processed and their
KV gradients are accumulated locally.

Communication is overlapped across independent event and optimization
work: projected KV for one event may transfer while another event's
behavior memory executes, completed event-KV gradients may return while other
candidate tiles continue, and closed dense-gradient or sparse-row
reductions may proceed during the remainder of backward. In the measured model
training workload, this overlap reduces exposed communication from
approximately $18\%$ to $10\%$ of step wall time.

\paragraph{Objective-preserving distributed execution.}
Physical partitioning must preserve the logical training objective of
Section~\ref{sec:multi-cutoff-training}. Loss numerators are accumulated
over the globally valid supervised targets and the declared
normalization is applied once at the logical update boundary. Likewise,
all same-step gradient contributions to a dense parameter or sparse
embedding row are combined before the corresponding nonlinear optimizer
update. Sequence and fanout partitioning therefore change the physical
execution graph, not the objective or optimizer semantics.

The underlying parameter and embedding sharding use standard distributed
primitives; the TransX-specific execution mechanisms are trajectory
packing, multi-cutoff memory extraction, event-state ownership, and
sequence--fanout scheduling.

\subsection{Disaggregated Serving}
\label{sec:disaggregated-serving}

\begin{figure}[htb]
    \centering
    \includegraphics[width=0.95\linewidth]{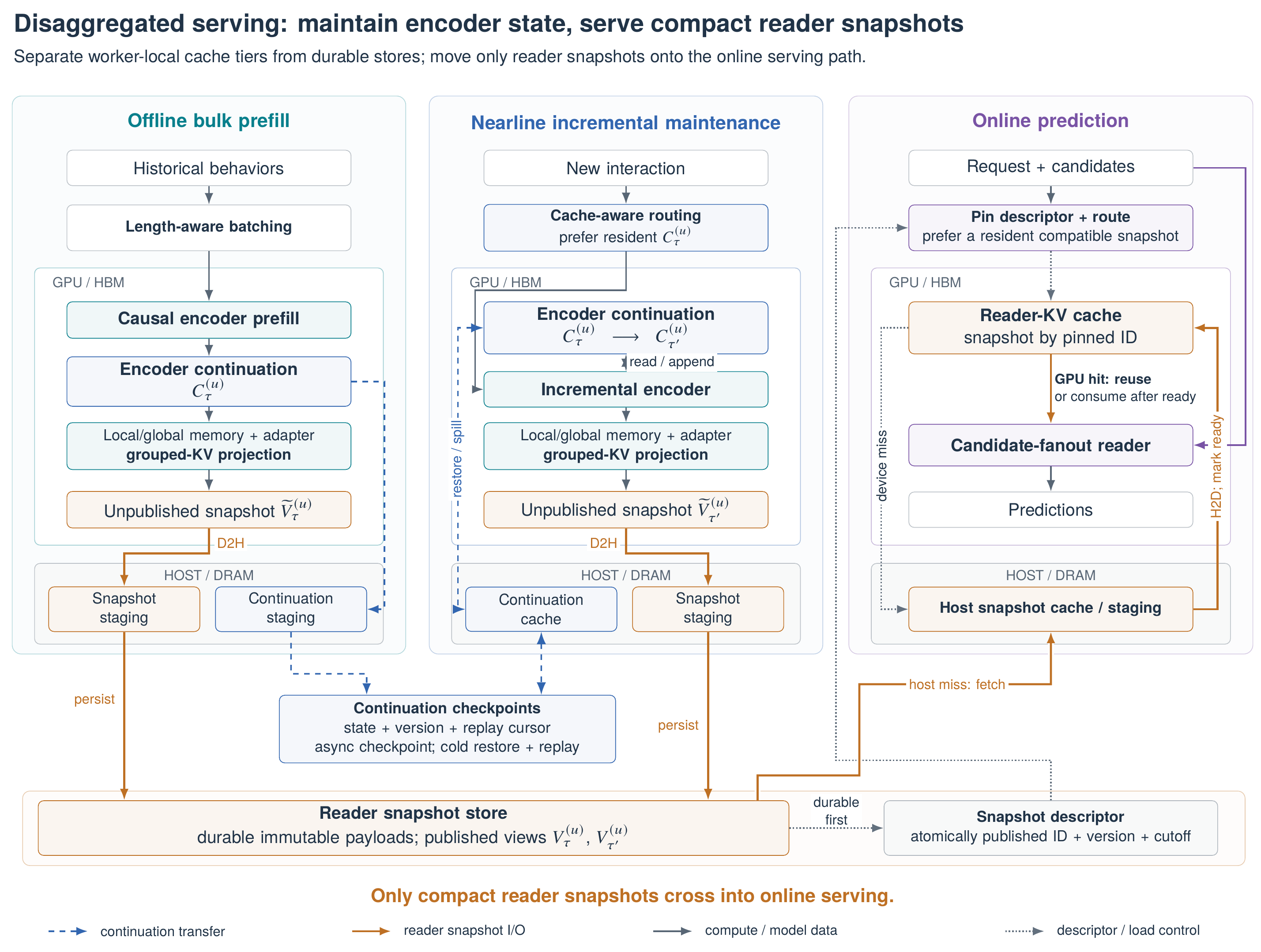}
    \caption{\textbf{Disaggregated TransX serving.} Offline and nearline
    workers maintain encoder continuation state and publish compact immutable
    behavior memory snapshots; online requests pin and consume only compatible behavior memory
    state. Continuation recovery remains off the online serving path, while
    snapshot residency and transfer are managed independently across GPU,
    host, and durable tiers. In practice, hardware placement can be made workload dependent as the three worker stages define logical execution roles rather than fixed device types (CPU/GPU).}
    \label{fig:disaggregated-serving}
\end{figure}

The amortization boundary induces serving workloads that operate on
different clocks. Behavioral state must be constructed when a new model
or representation is deployed, advanced when new user behavior arrives,
and consumed when recommendation requests are issued. These workloads
have fundamentally different execution objectives: bulk construction is
throughput-oriented, incremental maintenance is freshness-constrained,
and online prediction is latency-critical. Co-locating them under one
scheduler would couple workloads with incompatible batching, locality,
and tail-latency requirements.

TransX therefore disaggregates serving into three stages:
\emph{bulk state construction}, \emph{nearline incremental maintenance},
and \emph{online prediction}. This is not a request-synchronous
producer--consumer pipeline. History-dependent state is normally
constructed or refreshed before a prediction request arrives; online
serving consumes already materialized state rather than synchronously
advancing the behavioral encoder. The system therefore performs work at
the rate of the information that invalidates it: model deployment
triggers population-scale construction, behavioral updates trigger
state maintenance, and prediction demand triggers only
request- and candidate-conditioned computation.

The boundary is represented by two persistent state objects. Encoder
\emph{continuation state}
\[
C_{\tau}^{(u)}
\]
is mutable maintenance state containing the information required to
advance the causal encoder, including encoder KV, temporal metadata,
pooling state, recent outputs, and replay position. For serving notation, we use
\[
V_{\tau}^{(u)}
\]
to denote the \emph{published behavior-memory KV snapshot}: a bounded,
immutable projection of the behavioral memory consumed by online prediction. Bulk and nearline workers are the only
components that advance $C$ and publish new $V$; online requests resolve
and pin a compatible $V$ without restoring or advancing $C$.

This separation makes the three workloads independently schedulable and
scalable while preserving a common model interface. Bulk capacity
follows rollout and backfill demand, nearline capacity follows behavior
arrival and freshness lag, and online capacity follows request traffic,
candidate fanout, and the latency SLO. Table~\ref{tab:serving-stage-objectives}
summarizes these execution regimes before the following subsections
describe each path in turn.

\begin{table}[hbt]
\centering
\small
\setlength{\tabcolsep}{4pt}
\renewcommand{\arraystretch}{1.12}
\begin{tabularx}{\linewidth}{
  @{}p{0.20\linewidth}
     p{0.20\linewidth}
     p{0.3\linewidth}
     X@{}
}
\toprule
\textbf{Stage}
& \textbf{Trigger}
& \textbf{Execution objective}
& \textbf{Primary measurements} \\
\midrule

Bulk prefill
&
Deployment / backfill
&
Maximize sequence-processing throughput 
&
tokens/s/GPU, MFU
\\ \hline

Nearline incremental maintenance
&
User behavior updates
&
Maximize state update throughput within the freshness budget
&
Updates/s, p90 update lag, state-transfer payload
\\ \hline

Online prediction
&
Scoring requests
&
Minimize p99 end-to-end latency under production candidate fanout
&
p99 latency, requests/s/GPU, candidates/s/GPU
\\

\bottomrule
\end{tabularx}
\caption{
Disaggregated TransX serving workloads. The three stages operate on
different triggers and optimize different resource objectives even though
they share the same model representation interface.
}
\label{tab:serving-stage-objectives}
\end{table}

\subsubsection{Offline Bulk Prefill}
\label{sec:bulk-prefill}

A new model or representation version initially has no compatible
behavioral state for the serving population. The bulk path constructs
this state from historical trajectories before traffic is migrated to
the new representation:
\[
H_{\le \tau}^{(u)}
\xrightarrow{\operatorname{BulkPrefill}}
\left(
C_{\tau}^{(u)},
V_{\tau}^{(u)}
\right),
\]
where $C_{\tau}^{(u)}$ is encoder continuation state used by subsequent
nearline updates and $V_{\tau}^{(u)}$ is the behavior memory KV snapshot
consumed by online prediction. Bulk prefill is therefore a
population-scale initialization workload rather than a request-serving
workload; its objective is to maximize useful sequence-processing
throughput while completing the backfill within the rollout deadline.

\paragraph{Rolling-history semantics.}
Bulk prefill also serves as the rebase of behavioral state. Each bulk
cycle reconstructs continuation state $C_{\tau}^{(u)}$ and the corresponding
behavior-memory KV snapshot $V_{\tau}^{(u)}$ according to the configured retention
policy, including the lookback window and maximum sequence length.
Nearline maintenance can therefore advance $C_{\tau}^{(u)}$ append-only as
new behavior arrives until the next bulk prefill, avoiding the need to
remove expired records from already-contextualized hidden states in
place. Records that expire between bulk cycles may remain represented
until the next rebase.

\paragraph{Throughput-oriented execution.}
Behavioral trajectories have a highly skewed length distribution, so
fixed user-count batches would either waste computation on padding or
allow a small number of long histories to determine the execution
shape. We instead process workload under a behavior-token budget and form
length-aware variable-length batches. Since causal attention preserves user
boundaries without padded tokens, long trajectories can be processed in
bounded chunks while carrying encoder continuation state across chunks,
and these chunks are co-scheduled with shorter trajectories to reduce
long-tail stragglers. Context parallelism is used selectively when an
individual trajectory remains too large for efficient single-device
execution. Across the production workload, this path processes
approximately $0.7$M valid behavior tokens/s/H200 at $61\%$ encoder MFU
and $34$~GiB peak HBM, a $1.5\times$ throughput improvement over the
naive bulk implementation without chunking.

\paragraph{Bulk prefill state outputs.}
Bulk construction produces two state objects with different lifetimes
and scaling properties. Published behavior-memory KV $V_{\tau}^{(u)}$ is bounded by
the behavior memory interface and is persisted for subsequent online
consumption. Continuation state $C_{\tau}^{(u)}$ contains the encoder state
required to incrementally incorporate future behavior and may grow with
the retained history. The two objects therefore use independent
placement and residency policies: $V$ is optimized for repeated
prediction-time access, whereas $C$ is maintained only by the bulk and
nearline paths. Section~\ref{sec:nearline-maintenance} describes how
$C$ is subsequently advanced and converted into newer immutable behavior memory
snapshots; detailed storage and version-transition mechanics are given
in Appendix~\ref{app:snapshot-lifecycle}.

\subsubsection{Nearline State Refresh}
\label{sec:nearline-maintenance}

The nearline path converts newly arriving behavior into prediction-ready
behavior memory state under a freshness constraint. For user $u$, continuation
state $C_{\tau}^{(u)}$ is advanced with newly available behavior and used to
construct a new immutable behavior memory snapshot:
\[
\left(
C_{\tau}^{(u)},
\Delta H^{(u)}_{(\tau,\tau']}
\right)
\xrightarrow{\operatorname{Advance}}
C_{\tau'}^{(u)}
\xrightarrow{\operatorname{Project}}
\widetilde V_{\tau'}^{(u)}
\xrightarrow{\operatorname{Publish}}
V_{\tau'}^{(u)},
\label{eq:nearline-transition}
\]
where new behavior is contextualized by the encoder and then processed by the
global--local memory construction mechanism. 

\paragraph{Freshness-aware continuation-state routing.}
The cost of an incremental update depends jointly on the arriving
behavior delta and the continuation state needed to advance it.
Routing therefore considers
\[
\left(
|\Delta H|,
|C_{\tau}^{(u)}|,
\text{state locality},
\text{worker queue},
\text{freshness slack}
\right).
\]
Updates are preferentially routed to workers that already hold compatible
continuation state, avoiding unnecessary state movement. When queueing at
the current owner would consume the remaining freshness slack, however,
the update may be reassigned to another worker. Compatible updates are
microbatched across users subject to their freshness deadlines, while
state transitions for each user remain causally ordered. Under matched
update load, this locality- and freshness-aware routing policy reduces
p90 freshness lag by approximately $30\%$ relative to locality-agnostic
FIFO routing.

\paragraph{Immutable publication and recovery.}
A behavior-memory snapshot becomes visible only after its payload and
producing state have been validated. Publication then atomically advances
the user-visible snapshot descriptor, while each request pins one
compatible immutable snapshot for its lifetime. Continuation state is
checkpointed asynchronously together with its source cursor and version
metadata, allowing worker loss to be recovered by restoring a compatible
checkpoint and replaying the retained behavior-log suffix. Recovery
affects subsequent state production but does not block online requests
that can already consume a published compatible snapshot. Across users,
state fetch, incremental encoding, snapshot persistence, and publication
are pipelined to hide transfer and I/O latency while preserving per-user
causal ordering.

More design details are deferred to Appendix~\ref{app:snapshot-lifecycle}.

\subsubsection{Latency-Centric Online Prediction}
\label{sec:online-fusion-prediction}

The online serving path consumes materialized behavioral state and performs only
request- and candidate-conditioned prediction. Its objective is therefore
different from both bulk prefill and nearline maintenance: maximize
request goodput while satisfying the tail-latency SLO under variable
candidate fanout and state locality.

A useful approximation to the latency-critical serving path is
\[
T_{\mathrm{req}}
\approx
T_{\mathrm{queue}}
+
\max
\left(
T_{\mathrm{state}},
T_{\mathrm{feat}}
\right)
+
T_{\mathrm{fusion}}(k_t,n_t)
+
T_{\mathrm{head}},
\label{eq:online-critical-path}
\]
where $T_{\mathrm{state}}$ is snapshot lookup/admission and transfer,
$T_{\mathrm{feat}}$ is request/candidate feature preparation, and
$T_{\mathrm{fusion}}$ and $T_{\mathrm{head}}$ are the
candidate-conditioned model costs. The two preparation paths can
proceed concurrently. Importantly, no encoder restoration or behavioral
advancement appears on this latency-critical serving path.

\paragraph{Request readiness and state locality.}
At admission, each request resolves and pins one compatible
behavior-memory snapshot for all candidate batches.
A device-resident snapshot removes behavior-memory transfer from the
latency-critical path; concurrent misses for the same snapshot are
coalesced, and the resolved state remains pinned for the request lifetime.
Request-side sparse features follow the same readiness discipline:
GPU- and host-resident rows are served locally, while remote misses are
deduplicated and batched by owning shard. Behavior-memory KV preparation
and independent feature lookup proceed concurrently, and each candidate
batch is admitted only after all required state is ready. Detailed
sparse-embedding lookup mechanics are deferred to
Appendix~\ref{app:sparse-embedding-serving}.

\paragraph{Fanout-efficient candidate execution.}
Once request state is ready, candidate fanout is the dominant
model-side source of latency variation. Our production requests contain
approximately $0.5$K--$1.2$K candidate queries, with a bounded number
of behavior memory states. TransX therefore evaluates candidate-to-memory
attention as a rectangular operation,
\[
Q_t^{\text{cand}} {(K_t^{\text{mem}})}^\top
\in
\mathbb R^{k_t\times n_t},
\qquad
n_t \ll k_t,
\]
using preprojected behavior memory KV shared across the complete candidate
fanout, where \(Q_t^{\rm cand}\in\mathbb R^{k_t\times d_h}\) denotes the projected candidate-query matrix and \(K_t^{\rm mem}\in\mathbb R^{n_t\times d_h}\) the projected keys from the pinned behavior-memory snapshot. The online model cost consequently has the form
\[
C_{\mathrm{online},t}
=
C_{\mathrm{req}}(n_t)
+
k_t C_{\mathrm{cand}}(n_t),
\label{eq:online-fanout-cost}
\]
with no dependence on the full behavioral-history length once the
snapshot has been materialized. Behavior-memory KV is staged once and reused across candidate tiles and query heads sharing the same KV state. Candidate execution uses a small set of regular fanout shapes to bound temporary memory and per-kernel execution time while maintaining high GPU utilization.

\paragraph{Latency-budgeted admission and degradation.}
Queueing and state movement can dominate tail latency even when model
execution is efficient. Let $B_{\mathrm{rem}}$ denote the residual
request budget at a scheduling decision. A snapshot load or candidate
batch is admitted only when its estimated remaining work satisfies
\[
\widehat T_{\mathrm{load}}
+
\widehat T_{\mathrm{queue}}
+
\widehat T_{\mathrm{score}}
\le
B_{\mathrm{rem}}.
\label{eq:latency-admission}
\]
If the freshest compatible snapshot cannot be made ready within the
remaining budget, the request consumes an older admissible snapshot or
falls back to the designated non-sequential model rather than blocking
on continuation-state recovery. Snapshot refresh, cache refill, and
continuation reconstruction proceed asynchronously outside the
latency-critical serving path.

The resulting online contract is therefore simple: consume one
compatible immutable snapshot, make request state ready within the
latency budget, and scale prediction work with candidate fanout rather
than behavioral-history reconstruction.

\subsection{KV Storage and Execution}
\label{sec:kv-cache}

Online TransX execution operates on projected cross-attention KV rather
than on the model-side behavioral memory itself. Once a
memory $M_{\tau}^{(u)}$ has been constructed, each cross attention layer projects it
into the KV state to be consumed by candidate-conditioned cross attention:
\[
M_{\tau}^{(u)}
\;\longrightarrow\;
V_{\tau}^{(u)}
\;\longrightarrow\;
\operatorname{Pack}(V_{\tau}^{(u)}),
\label{eq:kv-representation-boundary}
\]
where $V_{\tau}^{(u)}$ denotes the logical behavior memory KV and
$\operatorname{Pack}(V_{\tau}^{(u)})$ its physical storage representation.
Bulk and nearline execution publish this KV state, and online workers
retrieve, cache, and execute directly against it. 

This boundary separates three concerns. \emph{KV fidelity} determines
how accurately the projected state is represented;
\emph{KV residency} determines whether the state resides in on-disk
storage, host memory, or GPU memory; and
\emph{KV execution format} determines whether compressed KV is decoded
on demand or expanded once and retained for repeated attention reads.
These choices affect storage, transfer, HBM footprint, and kernel cost, but they do not change the request-level semantics of the published KV.

\subsubsection{KV Quantization}
\label{sec:kv-quantization}

Although the number of behavior-memory states is bounded, storing
high-precision projected KV for a large active population can still
create substantial storage and transfer cost. We therefore quantize
behavior memory KV before publication.

The deployed codec uses a calibrated orthogonal rotation followed by
energy-aware precision allocation under a fixed storage budget. 
We observe pronounced channel anisotropy in unrotated KV projections,
similar to activation outliers reported in large Transformers
\citep{dettmers2022gpt3,xiao2023smoothquant}. We therefore rotate the
normalized KV into an empirical energy-aligned basis and allocate precision
according to per-coordinate energy, yielding a representation
better suited to low-bit quantization.
Codec parameters are estimated offline from the model-training behavior memory KV and versioned with the corresponding model and projection. Bulk and nearline workers apply the codec before publishing KV, so no fitting is required on the online serving path. Appendix \ref{app:kv-codec} gives the rotation, bit-allocation, and reconstruction equations.

In the measured configuration, a 4-bit-average direction payload reduces median behavior memory KV from approximately $28$~KiB/user in BF16 to $8.5$~KiB/user, a
$3.2\times$ reduction, with absolute test-AUC change within
$2\times10^{-3}$. Thus the physical KV footprint can be reduced
substantially without changing the logical memory resolution exposed by
the model.

\paragraph{Packed versus expanded KV execution.}
Compressed KV can be executed in two forms. Under \emph{packed
residency}, the compressed representation remains resident and active
KV tiles are decoded inside the attention path:
\[
\operatorname{Pack}(V)
\xrightarrow{\text{decode-and-attend}}
O.
\]
Under \emph{expanded residency}, KV is decoded once when admitted to the
GPU and the reconstructed representation is retained in HBM:
\[
\operatorname{Pack}(V)
\xrightarrow{\text{decode once}}
V^{\mathrm{bf16}}
\xrightarrow{\text{repeated attention}}
O.
\]
Packed execution reduces HBM footprint but pays decode cost repeatedly;
expanded execution spends additional HBM to amortize decoding across
subsequent candidate tiles or requests.

For 4-bit-average behavior memory KV configuration, packed residency uses approximately $3.1\times$
less behavior memory-state HBM than retaining expanded high-precision KV, with
approximately $6\%$ higher behavior memory-kernel latency. In the measured
workload, expanded residency becomes preferable after approximately
$4$--$6$ reads of the same KV snapshot. The packed-versus-expanded
choice is therefore itself an amortization decision: decode repeatedly,
or pay the expansion cost once and amortize it over future reads.

\subsubsection{Tiered KV Residency}
\label{sec:kv-residency}

Published KV is stored in compressed form in a shared on-disk storage tier and
cached opportunistically in host and GPU memory:
\[
\text{on-disk KV}
\;\longrightarrow\;
\text{host-resident KV}
\;\longrightarrow\;
\text{GPU-resident KV}.
\]
Residency is determined by expected KV reuse relative to resident bytes,
transfer cost, decode cost, and available memory. Cold KV remains packed
in the on-disk store; warmer state may be retained in host memory; and
active KV can be admitted to GPU in either packed or expanded form.

Importantly, residency and execution format are independent decisions.
A packed KV snapshot may move across disk, host, and GPU tiers
without changing its logical version, while promotion from packed to
expanded GPU residency changes only its local execution representation.
The same published KV therefore admits the following operating points:

\begin{center}
\small
\begin{tabular}{lll}
\toprule
\textbf{Placement} & \textbf{Format} & \textbf{Preferred regime} \\
\midrule
Backing store & Packed
& Cold population; minimize storage footprint \\
Host memory & Packed
& Warm state; avoid repeated on-disk store reads \\
GPU memory & Packed
& HBM-constrained or low expected reuse \\
GPU memory & Expanded
& Repeated reads with sufficient HBM \\
\bottomrule
\end{tabular}
\end{center}

KV admission can therefore adapt to access frequency without changing
the model interface. A newly fetched snapshot may initially remain
packed on GPU and later be promoted to expanded residency if observed or
predicted reuse justifies the additional HBM. Conversely, inactive,
unpinned KV can be demoted while preserving its compressed on-disk
representation.

All requests consume only KV compatible with their pinned model and
representation version. As the persistent behavioral state consumed by online prediction, the published behavior-memory KV snapshot \(V_{\tau}^{(u)}\) forms the natural boundary for storage, transfer, and device-execution optimization. Detailed compatibility metadata, immutable
publication, on-disk-store layout, request-scoped pinning, and
dual-version rollout are described in
Appendix~\ref{app:snapshot-lifecycle}.

\section{Scaling and Compute Allocation}
\label{sec:scaling}


Prior model scaling studies primarily characterized how predictive quality scales with model
capacity, data, and compute in language modeling and recommendation~\citep{hoffmann2022training,
zhang2024wukong,zhai2024actions}.
Rather than treating model size or total FLOPs as a single scaling
coordinate, we study how model's predictive gains are distributed across
behavioral encoding, behavior-memory resolution, and
candidate-conditioned fusion, and \emph{how this allocation changes with an increasing supervision budget}.
TransX makes these computational roles explicit so scaling can be studied along
three model coordinates:
\[
C_E
\;=\;
\text{behavior-encoder compute per encoded trajectory},
\]
\[
B_M
\;=\;
\text{behavior-memory budget exposed to online prediction},
\]
\[
C_F
\;=\;
\text{candidate-conditioned fusion compute per candidate}.
\]
The behavior memory budget $B_M$ measures the logical size of the state
exposed through the behavioral-memory interface; under fixed memory
width and precision it is proportional to the number of behavior-memory slots exposed to online prediction and to the uncompressed snapshot size. Unless otherwise stated, the memory-resolution sweeps in this section use
BF16 projected KV without quantization.

We study these model scaling axes under two distinct forms of data support.
Behavioral support $h$ characterizes the amount of historical evidence
available to the model, whereas supervision budget $\eta$ characterizes
the amount of mature response feedback available for training the
prediction model.


\subsection{Experimental Design}
\label{sec:scaling-design}

For data regime $d=(h,\eta)$,
we write
\[
Q(C_E,B_M,C_F\mid d)
\]
for held-out predictive quality. Unless otherwise stated, $Q$ denotes
the primary AUC metric used in the scaling experiments. This formulation leads to three questions. 
\begin{itemize}
    \item \textbf{First}, can additional behavioral-encoding compute improve prediction without increasing the behavior-memory budget exposed to prediction?
    \item \textbf{Second}, how does increasing supervision change the relative return to behavioral encoding versus candidate-conditioned fusion?
    \item \textbf{Third}, as model capacity grows, must the persistent behavior-memory interface become progressively finer?
\end{itemize}

\paragraph{Data support.}
We separate behavioral evidence from prediction supervision. Let
$l_t^{(u)}$ denote the number of eligible behavioral records available
at prediction event $t$. Using quantiles measured on a fixed reference
population, we construct three disjoint behavioral-support cohorts:
\[
\mathcal D_{\mathrm{sparse}}
=
\{(u,t): l_t^{(u)} \in[\ell_{5}, \ell_{33})\},\, \mathcal D_{\mathrm{medium}}
=
\{(u,t):  l_t^{(u)} \in [\ell_{33}, \ell_{66})\}, \, \mathcal D_{\mathrm{rich}}
=
\{(u,t): l_t^{(u)} \in [\ell_{66}, \ell_{95})\},
\]
where $\ell_q$ denotes the $q$-th percentile of history length in our
reference population. We use
\[
h\in\{\mathrm{sparse},\mathrm{medium},\mathrm{rich}\}
\]
to index these regimes. Their purpose is to test whether the value of model and memory resources changes systematically across populations with different amounts of behavioral support.

Within each history cohort, we independently vary the amount of
training supervision. Our scaling study uses a one-pass training protocol, so increasing \(\eta\) also increases the number of optimizer updates. Appendix~\ref{app:scaling-manifest} reports the corresponding training budgets. Let $N_{\mathrm{ref}}$ denote the reference number
of mature supervised candidate targets and define
\[
\eta = \frac{N}{N_{\mathrm{ref}}},
\qquad
\eta\in
\{\eta_{\mathrm{low}},
  \eta_{\mathrm{mid}},
  \eta_{\mathrm{high}}\}
\]
where \(N\) denote the number of supervised targets included in the corresponding regime. To construct a fixed $(h,\eta)$ regime, whole serving-event groups are sampled until approximately $\eta N_{\mathrm{ref}}$ valid candidate targets are
included. All architectural configurations compared within that regime
use the same sampled training data and the same held-out evaluation
population. Thus, $(h_1,\eta)\rightarrow(h_2,\eta)$
compares populations with different behavioral support at matched
supervision, while $(h,\eta_1)\rightarrow(h,\eta_2)$
changes supervision while holding the behavioral-support population
fixed.

\paragraph{Model scaling coordinates.}
The encoder coordinate $C_E$ is varied through a controlled family of
behavior-encoder configurations ordered by encoder MFLOPs per encoded
trajectory. The memory coordinate $B_M$ is varied through progressively
finer global--local memory realizations while holding the memory
construction family fixed. The fusion coordinate $C_F$ is varied
through candidate-conditioned cross-attention width and depth, reported
in MFLOPs per candidate.

For memory configuration $m_j$, let $n_t^{(j)}$ denote its number of memory slots exposed to online prediction. We additionally report the \emph{normalized memory resolution}
\[
\rho_j(d)
=
\mathbb E
\left[
\left.
\frac{n_t^{(j)}}{l_t^{(u)}}
\right| d
\right],
\label{eq:memory-density}
\]
which measures the fraction of encoded history exposed to online prediction.
Because $\rho$ can decrease mechanically as histories become longer, it
is never interpreted alone as a systems cost. The next section reports $\rho$ together with absolute slot
count and behavior memory state size.

For model configuration
\[
z=(C_E,B_M,C_F),
\]
we measure quality relative to a fixed reference configuration $z_0$
within the same data regime:
\[
\Delta Q(z\mid d)
=
Q(z\mid d)-Q(z_0\mid d).
\label{eq:relative-quality}
\]

\paragraph{Comparing marginal returns.}
Encoder and fusion FLOPs occur at different reuse multiplicities:
encoder computation is shared by all supervised events and candidate
targets associated with an encoded trajectory, whereas fusion is
executed per candidate. Therefore, raw component MFLOPs are not directly
exchangeable. To achieve matched-allocation analysis, let
\[
\nu(d)
=
\frac{
\text{number of valid supervised candidate targets}
}{
\text{number of encoded trajectories}
}
\]
denote the measured encoder-reuse factor under regime $d$. We define
the amortized encoder compute per supervised target as
\[
\bar C_E(d)
=
\frac{C_E}{\nu(d)}.
\label{eq:amortized-encoder-compute}
\]
This allows encoder and fusion computation to be expressed on a common
training-equivalent per-target axis $C_{\mathrm{eq}}
=
\bar C_E + C_F$
with all other model components held fixed\footnote{Note that $C_{\rm eq}$ is a training-equivalent allocation coordinate under the
measured reuse density $\nu$; it is not a serving-cost objective, where
encoder and fusion work occur at different execution rates.}. 

For adjacent measured configurations $c_1$ and $c_2$ along one scaling
axis, let $C(c_i), i\in\{1,2\}$ denote the corresponding compute
coordinate: $\bar C_E$ for encoder comparisons and $C_F$ for fusion
comparisons. We summarize local empirical return as
\begin{equation}
g(c_1,c_2)
=
\frac{
Q(c_2)-Q(c_1)
}{
\log C(c_2)-\log C(c_1)
}.
\label{eq:local-scaling-return}
\end{equation}
We use this quantity only over overlapping measured compute ranges. It
is a finite-difference summary of the experimental sweep, not a
continuous derivative or a universal scaling exponent. The complete scaling manifest is reported in Appendix~\ref{app:scaling-manifest}.

\subsection{Scaling Returns and Compute Allocation}
\label{sec:scaling-analysis}

\begin{figure}
    \centering
    \includegraphics[width=0.95\linewidth]{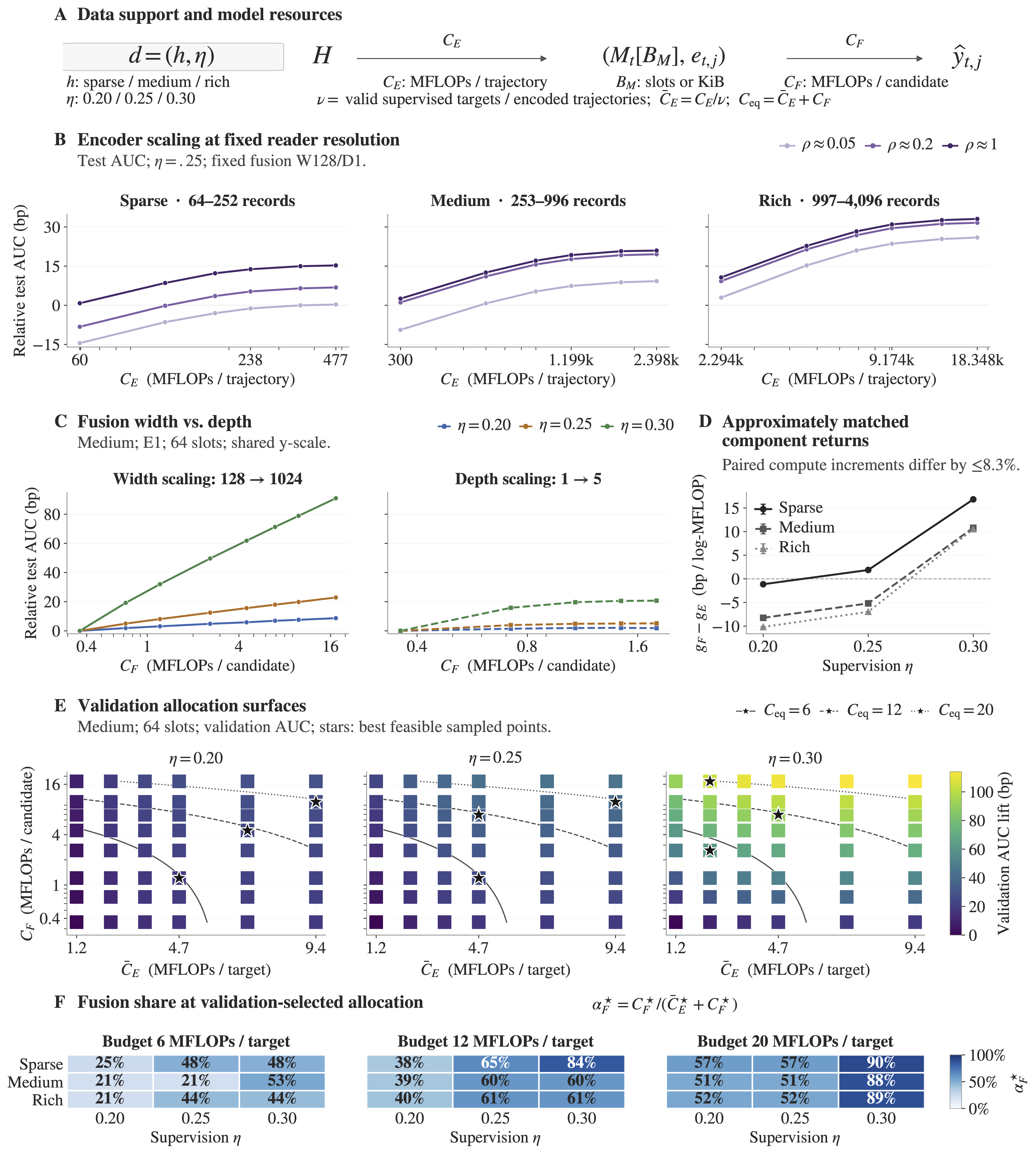}
    \caption{
    \textbf{Compute allocation under changing data support.}
    \textbf{(A)} Experimental decomposition of data support
    $d=(h,\eta)$ and model resources: behavioral-encoder compute $C_E$,
    behavior memory budget $B_M$, and candidate-conditioned fusion compute
    $C_F$. Encoder compute is amortized by the measured number of valid
    supervised targets per encoded trajectory,
    $\bar C_E=C_E/\nu$, with equivalent compute
    $C_{\mathrm{eq}}=\bar C_E+C_F$.
    \textbf{(B)} Increasing encoder capacity improves test AUC at fixed
    behavior-memory resolution across sparse-, medium-, and rich-history
    cohorts.
    \textbf{(C)} Fusion width provides more sustained gains than additional
    cross attention depth, with larger returns under increased supervision.
    \textbf{(D)} Approximately matched finite-difference comparisons show
    the relative marginal return shifting from encoder toward fusion compute
    as supervision increases; paired compute increments differ by at most
    $8.3\%$.
    \textbf{(E)} Validation AUC over the sampled encoder--fusion allocation
    grid for the medium-history cohort; contours denote equivalent-compute
    budgets and stars mark the best feasible measured configuration on each
    budget.
    \textbf{(F)} Fusion share
    $\alpha_F^\star=C_F^\star/(\bar C_E^\star+C_F^\star)$ at the
    validation-selected allocation across history cohorts, supervision
    budgets, and equivalent-compute budgets.
    Together, the results show that recommendation scaling is an
    allocation problem: increasing supervision changes where additional
    compute is most productively placed. 
    }
    \label{fig:scaling-compute-allocation}
\end{figure}

\begin{figure}[htb]
    \centering
    \includegraphics[width=0.95\linewidth]{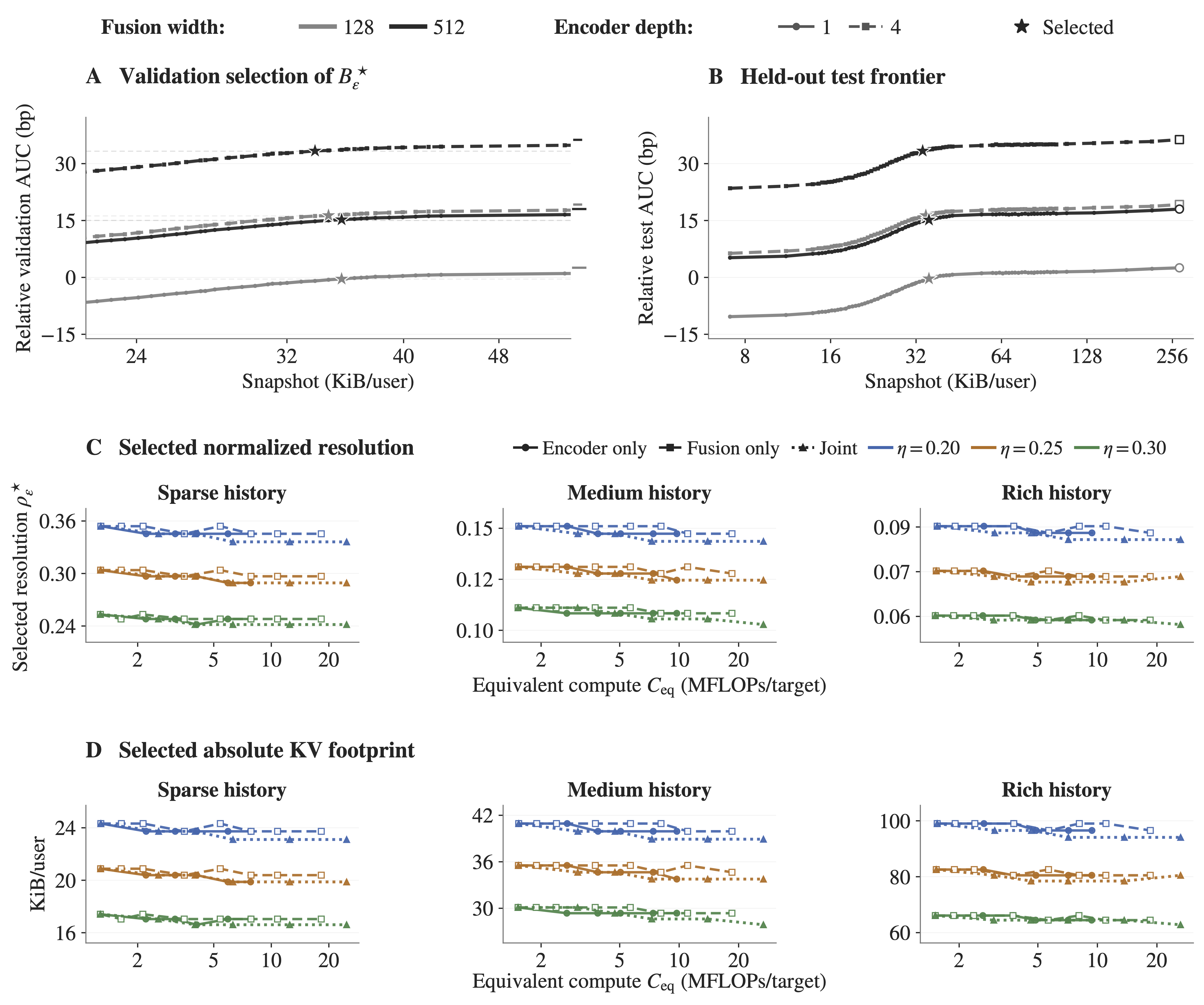}
    \caption{
    \textbf{Behavior memory resolution under model scaling.}
    \textbf{(A)} Validation selection of the minimum
    $\epsilon$-adequate behavior memory budget $B_\epsilon^\star$: for each
    encoder--fusion configuration, the star marks the smallest evaluated
    snapshot remaining within $\epsilon$ of its model-specific
    highest-resolution reference.
    \textbf{(B)} Held-out test quality–memory-budget frontier for the same
    configurations, with the validation-selected operating points carried
    forward to test evaluation.
    \textbf{(C)} Selected normalized behavior-memory resolution (density)
    $\rho_\epsilon^\star$ as equivalent model compute increases across
    sparse-, medium-, and rich-history cohorts and across supervision
    budgets.
    \textbf{(D)} Corresponding absolute prediction-facing physical KV footprint in
    KiB/user.
    Across encoder-only, fusion-only, and joint scaling paths, increasing
    model capability does not systematically require finer behavior memory
    resolution or proportionally larger KV state. Larger models instead
    shift the empirical quality-memory-budget frontier upward, separating
    model scaling from growth of the persistent behavioral interface.
    }
    \label{fig:scaling-memory-frontier}
\end{figure}

\paragraph{Behavioral Encoder Scaling at Fixed Memory Resolution.}
Figure~\ref{fig:scaling-compute-allocation}B isolates encoder scaling
while holding behavior memory resolution and candidate-conditioned fusion
fixed. Across the evaluated history-support regimes and memory
resolutions, increasing encoder compute improves held-out quality, with
the largest gains occurring in the lower-to-middle portion of the sweep
and substantially smaller returns beyond approximately four layers.
Because these gains occur at fixed $B_M$ and $\rho$, they cannot be attributed simply to exposing more behavioral state to online prediction.

The result shows that additional behavioral computation can improve the
predictive utility of a fixed-resolution interface. A stronger encoder
may preserve more response-relevant information through memory
construction, organize the retained information into a representation
that is easier for the downstream predictor to exploit, or both. The observed
saturation is also operationally important: encoder capacity affects
training, bulk construction, incremental maintenance, and continuation
state, so behavioral compute should be treated as one resource in a
broader allocation problem rather than as the sole scaling axis.

\paragraph{Supervision Changes the Return to Candidate-Conditioned Compute.}
Figure~\ref{fig:scaling-compute-allocation}C shows that the return to
candidate-conditioned fusion increases with response supervision.
Fusion width provides sustained gains over the evaluated compute range,
whereas additional cross attention depth saturates earlier. Note this does not imply that repeated memory reads are redundant, only that width is the more
effective allocation over the measured configurations. Because training
is end to end, increased supervision also benefits the encoder and
memory construction. The results above should be interpreted as the \emph{relative} return to where additional capacity is placed.

Figure~\ref{fig:scaling-compute-allocation}D--F makes this shift explicit.
For each supervision regime $\eta$, let $g_E(h,\eta)$ and $g_F(h,\eta)$
denote the local empirical returns from
Eq.~\ref{eq:local-scaling-return} along the amortized encoder-compute
axis $\bar C_E$ and the candidate-conditioned fusion-compute axis $C_F$,
respectively, evaluated over approximately matched compute increments
with the remaining scaling coordinates held fixed. We summarize their
relative return by
\[
\Delta_{\mathrm{alloc}}(h,\eta)
=
g_F(h,\eta)-g_E(h,\eta).
\]
As supervision increases, \(\Delta_{\rm alloc}(h,\eta)\) increases within each history-support regime, indicating that the relative marginal return shifts toward candidate-conditioned fusion over the measured range. Consistently,
the allocation-grid analysis shows that, within the evaluated
encoder--fusion grid, the validation-selected configuration under a
fixed equivalent-compute budget moves toward a larger fusion share:
\[
\alpha_F^\star
=
\frac{C_F^\star}
{\bar C_E^\star+C_F^\star}.
\]
Note this does not imply a universal preference for fusion over encoding: the selected allocation depends on the architecture family, behavioral support, memory interface, supervision regime, and compute range. Rather, marginal returns are anisotropic and differ across scaling axes and change with the supervision regime: increasing supervision shifts the most productive use of additional compute toward candidate-conditioned fusion over the measured range.

\paragraph{Scaling the Model Does Not Systematically Require Finer Behavior Memory Resolution.}
We next ask whether larger models require proportionally finer
online-visible behavioral state. For each encoder--fusion configuration,
let $B_\epsilon^\star$ denote the smallest evaluated memory budget whose
quality remains within $\epsilon$ of that configuration's
highest-resolution reference. We use $\epsilon=3$ AUC basis points in
the primary analysis. Figure~\ref{fig:scaling-memory-frontier} shows that
increasing encoder, fusion, or joint model capacity does not
systematically increase this requirement. Instead, the selected
normalized density $\rho_\epsilon^\star$ remains approximately stable
or decreases across the evaluated scaling paths and supervision budgets.

Because $\rho=\mathbb E[n/l]$ can fall mechanically as histories grow,
we additionally report the corresponding absolute KV footprint under a
fixed cross attention width and precision. The absolute results show the same
qualitative behavior: improvements in model capacity do not require
proportional growth in prediction-facing state. Equivalently, larger
models shift the empirical quality-memory-budget frontier upward---higher
quality can be obtained at a comparable behavior-memory budget. Thus,
within the evaluated regimes, scaling model capability and scaling the
persistent behavioral interface are distinct decisions rather than
quantities that must grow together.

\subsection{Interpretation and Implications}
\label{sec:scaling-interpretation}

The scaling results separate two roles: constructing a useful behavioral
representation and learning the candidate-conditioned mapping that
consumes it. For the binary response-prediction setting studied in our scaling analysis, a population log-loss decomposition makes this distinction explicit. For behavioral memory $M=f_\phi(H)$ and predictor
$q_\psi(Y\mid M,X)$ where \(X\) is the prediction-time candidate/request context, we prove in Appendix~\ref{app:scaling-decomposition}:
\begin{equation*}
    \mathcal L(\phi,\psi)= \underbrace{\mathsf H(Y\mid H,X)}_{\text{irreducible uncertainty}}
    +
    \underbrace{I(Y;H\mid M,X)}_{\text{information lost through memory}}
    +
    \underbrace{
    \mathbb E_{M,X}
    D_{\mathrm{KL}}
    \!\left(
    p_M(\cdot\mid M,X)
    \,\Vert\,
    q_\psi(\cdot\mid M,X)
    \right)
    }_{\text{readout mismatch}},
\end{equation*}
where $\mathsf H$ denotes the Shannon entropy. The decomposition is exact for population log loss; for binary response
tasks, BCE is its empirical counterpart, whereas AUC is a complementary
rank-order metric. The value of the decomposition here is to separate the
\emph{information bottleneck} imposed by behavior memory from the mismatch of the
downstream readout. Under this lens, the memory budget $B_M$ constrains the
admissible memory representation and therefore the achievable
information-loss term $I(Y;H\mid M,X)$, while fusion capacity $C_F$
constrains the candidate-conditioned readout and therefore the achievable
KL mismatch.

This suggests two distinct ways in which additional capacity can improve
prediction:
\begin{itemize}
    \item \textbf{Increasing encoder capacity} can improve both the
    predictive information preserved in the behavioral memory and how
    effectively that information is presented for downstream readout.

    \item \textbf{At a fixed memory map, increasing candidate-conditioned
    capacity} acts on the mapping from $(M,X)$ to the response, reducing
    readout mismatch without changing the information available through
    $M$. The observed shift toward fusion under stronger supervision is consistent with a larger supervision budget supporting more effective learning of the candidate-conditioned mapping once the memory is sufficiently informative.
\end{itemize}

This decomposition also helps explain \textbf{why increasing model
capability need not require finer behavior-memory resolution}. The
$\epsilon$-adequate operating point is defined relative to the
highest-resolution memory at the same encoder, fusion, and data regime.
Thus, what matters is not whether model quality improves in absolute
terms, but whether the quality gap between compact and high-resolution
memory grows. If stronger encoding and prediction improve both
representations comparably---or improve the compact representation
relatively more---then the smallest $\epsilon$-adequate memory budget
need not increase with model compute. Consequently,
$B_\epsilon^\star$ and $\rho_\epsilon^\star$ may remain stable or even
decrease as model capability grows. Empirically,
this is what we observe: encoder scaling improves quality at fixed
behavior-memory resolution, increasing supervision shifts the relative
return of compute toward candidate-conditioned fusion, and
quality is preserved at approximately stable or lower behavior-memory resolution.

These results motivate viewing recommendation scaling as
\textbf{conditional compute allocation over an amortized behavioral
interface}. Behavioral encoding, memory budget, and
candidate-conditioned interaction are distinct resources with different
costs and returns. The preferred allocation therefore depends on the
joint data--model operating point rather than on a fixed allocation ratio between encoder and fusion capacity. Operationally, this is particularly favorable for TransX: model capability can increase without requiring the persistent KV state
transferred, cached, and consumed by online prediction to grow
proportionally. 

\section{Post-Training for Capability Specialization}
\label{sec:post-training}

Core TransX training learns a flexible behavioral interface that supports
candidate-conditioned response prediction. We use \emph{post-training} to denote specialization of this interface after core response-model training for various downstream tasks.

\paragraph{Why post-training?}
We use post-training rather than joint multi-task training as the specialization strategy because the core response task and downstream capabilities can impose different output geometries, supervision structures, and serving interfaces~\citep{tang2020progressive,geng2022recommendation}. Core click/conversion training provides dense, high-volume supervision for learning a general behavioral representation, whereas retrieval, preference ranking, and generative recommendation respectively require contrastive matching, event-level preference comparisons, and autoregressive sequence objectives. Separating these stages allows the expensive behavioral
backbone to be learned once and subsequently held fixed while each capability adapts only the reader and output components required by its prediction geometry.  We do not assume that
sequential specialization is universally superior to joint training. Rather, joint learning, partial backbone adaptation, and cross-task distillation are alternative operating points whose relative value depends on task compatibility and available training and serving budgets.

For post-training, let
\[
M_t^{(u)}
=
f_\phi
\left(
H_{\le t}^{(u)}
\right)
\]
denote the behavioral memory produced by the pre-trained tokenizer, causal
encoder, and memory construction from the core response prediction task. A capability $a$ defines its own query
set $Q_t^{(a)}$, $\mathrm{TaskReader}^{(a)}$, and output head $\mathrm{TaskHead}^{(a)}$:
\[
\hat{P}_t^{(a)}
=
\mathrm{TaskHead}^{(a)}
\left(
\mathrm{TaskReader}^{(a)}
\left(
M_t^{(u)},Q_t^{(a)}
\right)
\right).
\label{eq:posttraining-interface}
\]
In the post-training recipe, the behavioral encoder and memory
construction $\phi$ are frozen and task-specific parameters $\psi_a$ are optimized:
\[
\psi_a^\star
=
\arg\min_{\psi_a}
\sum_{u,t}
\mathcal L^{(a)}
\left(
\hat{P}_t^{(a)},Y_t^{(a)};
M_t^{(u)}
\right),
\label{eq:posttraining-objective}
\]
where \(Y_t^{(a)}\) denotes the task-specific target and $\mathcal L^{(a)}$ is the loss objective.
Thus post-training changes how the behavioral representation is
\emph{read and supervised} rather than requiring a separate behavioral
backbone for every capability.

The model-side memory $M_t$ remains shared while the projected reader KV is $\psi_a$-specific:
\[
V_t^{(a)}
=
\operatorname{Proj}_{KV}^{(a)}(M_t).
\label{eq:task-specific-kv}
\]
Consequently, behavioral encoding and memory materialization remain
shared, while reader projections, output geometry, objective, and the
corresponding KV version may differ across specializations.

Table~\ref{tab:posttraining-geometries} summarizes the resulting
readout geometries. The distinction is structural: retrieval collapses
the behavioral state into one contextual matching representation;
preference ranking preserves candidate-parallel scoring but changes the
event-level supervision; generative retrieval replaces explicit
candidate scoring with autoregressive Semantic-ID decoding; and
generative slate ranking conditions sequential selection on
candidate representations computed from the shared TransX backbone.

\begin{table}[t]
\centering
\small
\begin{tabular}{llll}
\toprule
\textbf{Capability} &
\textbf{Readout geometry} &
\textbf{Output} &
\textbf{Post-training objective} \\
\midrule

Core response
&
Candidate query $\rightarrow M_t$
&
Per-candidate response
&
Pointwise click/CVR \\

Retrieval
&
Context query $\rightarrow M_t$
&
One retrieval embedding
&
Contrastive / InfoNCE \\

Preference ranking
&
Candidate queries $\rightarrow M_t$
&
Candidate logits
&
Event-wise Bradley--Terry \\

Generative retrieval
&
Decoder state $\rightarrow M_t$
&
Semantic-ID sequence
&
AR + preference optimization \\

Generative ranking
&
Slate state $\rightarrow \{z_{t,j}\}$
&
Ordered candidate slate
&
AR + preference optimization \\
\bottomrule
\end{tabular}
\caption{Capability specialization over the TransX behavioral
representation. The behavioral backbone remains common while the query
multiplicity, output geometry, and supervision change across tasks.
Generative variants are evaluated offline in this work.}
\label{tab:posttraining-geometries}
\end{table}

\subsection{Retrieval Specialization}
\label{sec:retrieval-specialization}

Retrieval requires one request-level representation that can be matched
against an independently indexed item corpus. Unlike the core response
reader, whose query is candidate specific, the retrieval reader uses
serving context $\xi_t^{(u)}$ to query the behavioral memory once: 
\[
h_t^{(u)}
=
\operatorname{norm}
\left[
\operatorname{CrossAttn}_{\mathrm{ret}}
\left(
M_t^{(u)},\xi_t^{(u)}
\right)
\right]
\in\mathbb R^{d_r}
\label{eq:retrieval-readout}
\]
where $d_r$ is the retrieval embedding dimension.
The item tower independently produces
\[
v_j
=
\operatorname{norm}
\left[
\operatorname{ItemEnc}(e_j)
\right],
\qquad
s_{t,j}^{(u)}
=
\frac{
\langle h_t^{(u)},v_j\rangle
}{
\tau_{\mathrm{ret}}
},
\label{eq:retrieval-score}
\]
where $e_j$ denotes the item attributes and $\tau_{\mathrm{ret}}$ is the temperature.
Because the item representation \(v_j\) is independent of the current user and request, it can be precomputed and indexed, as in standard two-tower retrieval models.

Post-training uses positive items together with hard non-positive
candidates from the same serving event and negatives from other
user--event examples in the logical batch. The resulting contrastive
objective changes the readout geometry from candidate-wise response
prediction to request-level matching while leaving behavioral encoding
and memory construction unchanged. Exact negative-set construction and
the InfoNCE objective~\citep{oord2018cpc} are given in
Appendix~\ref{app:retrieval-posttraining}.

\subsection{Preference-Ranking Specialization}
\label{sec:ranking-specialization}

Preference ranking retains the candidate-conditioned TransX scorer but
changes the reduction applied to candidate outputs. For event $t$,
valid candidates are scored in parallel:
\[
a_{t,j}^{(u)}
=
\operatorname{RankLogit}
\left(
M_t^{(u)},
x_{t,j}^{(u)}
\right),
\qquad
j\in\mathcal V_t^{(u)} .
\label{eq:ranking-logit}
\]
Here \(\operatorname{RankLogit}\) denotes the same TransX-Base scoring path—candidate-query encoding, candidate-conditioned cross-attention, concatenation with \(\phi_{\rm pred}(x_{t,j}^{(u)})\), and the shared interaction network—followed by a scalar ranking output head.

A declared preference relation over mature response outcomes induces
within-event ordered pairs
\[
\mathcal P_t^{(u)}
=
\left\{
(i,j):
y_{t,i}^{(u)}
\succ_y
y_{t,j}^{(u)}
\right\}.
\]
For a sampled subset
$\widetilde{\mathcal P}_t^{(u)}$, the specialization minimizes a
Bradley--Terry loss~\citep{turner2012bradleyterry} of the form
\[
\mathcal L_{\mathrm{rank}}^{(u,t)}
=
-
\frac{1}{
|\widetilde{\mathcal P}_t^{(u)}|
}
\sum_{(i,j)\in\widetilde{\mathcal P}_t^{(u)}}
\log
\sigma
\left(
\frac{
a_{t,i}^{(u)}-a_{t,j}^{(u)}
}{
\tau_{\mathrm{rank}}
}
\right),
\label{eq:ranking-main-loss}
\]
where \(\tau_{\rm rank}\) is the ranking temperature.

The important architectural property is that no slate-conditioned model
is required for this specialization: candidate scoring remains parallel,
and only the event-level supervision changes from independent response
losses to preference comparisons. The packed event representation from
core training therefore becomes the natural post-training unit. Pair
sampling, weighting, normalization, and distributed pair reduction are
specified in Appendix~\ref{app:ranking-posttraining}.

\subsection{Generative Recommendation Extensions}
\label{sec:generative-specialization}

We additionally study two generative extensions of the shared behavioral
backbone. These variants remain under offline iteration and are included
to examine how the TransX interface supports sequential output
geometries.

\paragraph{Generative retrieval.}
TransX-GenRet replaces embedding-space retrieval with autoregressive
Semantic-ID generation~\citep{rajput2023tiger}. For item $j$, let
\[
G_j=\operatorname{SID}(j)=(g_{j,1},\ldots,g_{j,K})
\]
denote its Semantic-ID sequence.
At decoding step $k$, the decoder state induced by the serving context
and generated prefix $G_{<k}$ queries the same behavioral memory:
\[
h_{t,k}^{(u)}
=
\operatorname{CrossAttn}_{\mathrm{GR}}
\left(
M_t^{(u)},
q_{t,k}^{(u)}
\right),
\qquad
q_{t,k}^{(u)}
=
\operatorname{DecodingQuery}
\left(
\xi_t^{(u)},G_{<k}
\right).
\]
The resulting next-token distribution is
\[
\pi_\theta
\left(
g_k
\mid
M_t^{(u)},\xi_t^{(u)},G_{<k}
\right).
\]
Because $M_t^{(u)}$ is fixed throughout the request, each generative
reader layer projects its behavioral KV once and shares that state across
decode steps and hypotheses. Post-training uses autoregressive
supervision together with preference optimization~\citep{rafailov2023dpo} over preferred and
rejected Semantic-ID sequences. Semantic-ID construction, negative
selection, constrained decoding, and the full preference objective are
given in Appendix~\ref{app:genret-posttraining}.

\paragraph{Generative slate ranking.}
TransX-GenRank instead preserves candidate-conditioned TransX
representations and concentrates sequential capacity on dependencies
among slate selections~\citep{jiang2019listcvae,pang2025higr}. For each candidate $i$, the parallel TransX scoring path first produces
a candidate-conditioned representation $z_{t,i}$ and a base relevance
logit $b_{t,i}$. These quantities are computed once and remain fixed
throughout slate decoding.

For an ordered slate $S=(j_1,\ldots,j_m)$, let $h_{t,k}$ denote the
autoregressive slate state summarizing the serving context and the
previously selected prefix $j_{<k}$. At position $k$, each remaining
candidate $i\notin\{j_1,\ldots,j_{k-1}\}$ is assigned the score
\[
s_{t,k,i}
=
b_{t,i}
+
\frac{(W_q h_{t,k})^\top(W_c z_{t,i})}{\sqrt{d_{\text{slate}}}},
\]
with previously selected candidates masked from the next-item distribution.
The resulting slate policy factorizes as
\[
\pi_\theta(S\mid \xi_t)
=
\prod_{k=1}^{m}
\pi_\theta(j_k\mid \xi_t,j_{<k}),
\]
where $\xi_t$ denotes the fixed request context and candidate
representations available to the slate decoder. Behavioral memory and candidate representations are therefore computed
once, while autoregressive computation is restricted to the evolving
slate and slate-conditioned scoring residual.

Preference pairs are constructed from candidate slates scored by a
frozen preference teacher, and the generative policy is optimized
against a frozen pre-preference reference policy. The exact slate-pair
construction and preference objective are given in
Appendix~\ref{app:genrank-posttraining}.

Across these specializations, the common design principle is not a
single universal output head but a common behavioral interface with
task-specific readout geometry. Post-training changes how behavioral
state is queried, compared, or decoded while preserving the expensive
history-side representation learned by the core model.
\section{Real-World Deployment}
\label{sec:evaluation}

We evaluate TransX as an industrial recommendation system rather than as
an isolated model benchmark. The evaluation therefore considers predictive
quality together with the training and serving properties that motivate the
architecture: reuse of behavioral computation across serving events,
candidate-fanout efficiency, online latency and freshness, and operational
robustness under heterogeneous production traffic. Unless otherwise stated,
we compare against the baseline production recommendation stack under
matched data and product constraints.

\subsection{Deployment Setting and Evaluation Protocol}
\label{sec:deployment-setting}

\paragraph{Production setting.}
We deploy TransX in a large-scale social recommendation application with long, heterogeneous user histories and highly variable candidate fanout, ranging from approximately $0.5$K to $1.2$K candidates per request. The system serves on the order of $20$K requests per second and processes hundreds of millions of user interactions per day under strict latency and infrastructure-cost constraints. Online model scoring runs primarily on an H200 GPU serving fleet and is subject to a p99 end-to-end scoring-latency SLO of less than $50$ ms.

Historical behavior is constructed over a 180-day lookback window, allowing the behavior encoder to exploit substantially longer context than the short-history features typically exposed to conventional ranking models. As described in Section~\ref{sec:disaggregated-serving}, a fallback DLRM~\citep{naumov2019deep} is co-deployed to serve requests for which no valid KV cache can be retrieved.

The deployed system follows the same disaggregated serving design. Bulk and nearline workers maintain behavioral state and publish compact behavior memory KV snapshots, while online workers perform only candidate-conditioned readout and prediction. This separation makes both model quality and infrastructure efficiency first-class deployment criteria.

Models are retrained and freshly deployed daily, with an end-to-end refresh SLO requiring training, bulk prefill, and production rollout to complete within $20$ hours. Nearline state maintenance is subject to a p90 freshness-lag SLO of $2$ seconds, aligned with the system's typical pagination interval.

\paragraph{Production reference.}
For each TransX ramp, the control is the incumbent model currently serving the same recommendation product surface, so the reported gains measure the production transition rather than a comparison against an offline or simplified baseline. For the click/conversion prediction task on the primary landing and home-page surfaces, the control is the production HSTU-style~\citep{zhai2024actions} sequential recommender, trained over a 180-day behavioral lookback window (capped at $512$ tokens) using our production configuration. This control performs history encoding and candidate scoring entirely on the online serving path. For retrieval, the control is a sequential embedding-based retrieval (EBR) model whose user tower encodes up to 180 days of behavioral history (capped at $1024$ tokens) with Transformer blocks. The user and item towers map into a shared 128-dimensional embedding space, and the resulting representations are used by an ANN service to retrieve the top-200 candidates, which are merged with candidates from other retrieval sources before downstream ranking. For the slate-ranking task used on notification and email surfaces, the control is a production DLRM-based learning-to-rank model (LtR DLRM) \citep{borisyuk2024lirank} in which behavioral history is represented through engineered aggregate and sequence-derived features and scored entirely online.

Accordingly, the emphasis of this technical report is the production transition from established serving systems to TransX---including quality, efficiency, and operational behavior---rather than leaderboard-style benchmarking against generic reference models.

\paragraph{Offline protocol.}
We construct the offline dataset from up to one year of historical user behavior and 180 days of event-level feedback, choosing the terminal date sufficiently far in the past for delayed feedback to have fully matured. For each user trajectory, the behavioral sequence contains up to the preceding $180$ days of activity, capped at $8192$ tokens. For each packed training unit, supervised events within that unit share its left boundary and differ only in their causal right-hand cutoffs. Training, validation, and test sets are split chronologically, with the validation and test sets each spanning one week. We exclude users with insufficient behavioral history and downsample users whose sequence lengths exceed the $99$th percentile.

The same data construction is used across click and conversion prediction, retrieval, and preference-ranking tasks. For core click/conversion training, we retain at most five sampled negatives
per positive for click and one sampled negative per positive for conversion
within each serving event, downsampling excess negatives when necessary.
Events containing no positive example for a task are discarded. Validation and testing both use the natural label distribution. Retrieval and preference-ranking tasks reuse the same underlying labeled interactions. For preference ranking, a product-defined utility scoring function is further used to construct ordered candidate lists for training. The task-specific modeling objectives and offline evaluation protocols are described in the following sections.

\paragraph{Online experiment protocol.}
For \texttt{TransX-Base} and \texttt{TransX-Retrieval} ramps, the topline business metrics are Online \texttt{CTR}, \texttt{Conversion} (7-day attribution window), \texttt{DAU} (daily active users), and \texttt{Home Feed Update Views}. For \texttt{TransX-Ranking}, the business metrics are the downstream \texttt{conversion} and \texttt{DAU}. Unless otherwise stated, our online experiments randomize at the user level and run for at least 7 days. Treatment and control receive approximately 20\% and 20\% of eligible traffic, corresponding to approximately 2 billion requests. The reported results use CUPED~\citep{deng2013cuped} for variance reduction where applicable. Primary online estimands are aggregated at the user level to match the randomization unit, with p99 end-to-end system latency monitored as a guardrail. Orthogonal randomization is used to enable overlapping experimentation traffic.

\subsection{Offline Model Quality}
\label{sec:offline-quality}

\begin{table}[hbt]
\centering
\footnotesize
\setlength{\tabcolsep}{5pt}
\renewcommand{\arraystretch}{1.15}
\begin{tabularx}{\linewidth}{@{}p{0.20\linewidth}p{0.45\linewidth}X@{}}
\toprule
\textbf{Model} & \textbf{Description} & \textbf{Task-specific Component} \\
\midrule

\texttt{TransX-Base}
&
Base TransX model trained jointly on binary click and conversion objectives
as described in Section~\ref{sec:core-training}.
&
Deep\&Wide-based prediction head~\citep{cheng2016wide} with DCN feature interaction~\citep{wang2021dcn}
\\ \hline

\texttt{TransX-Retrieval}
&
\texttt{TransX-Base} post-trained for discriminative retrieval using the
contrastive objective described in
Section~\ref{sec:retrieval-specialization}.
&
Context-to-memory retrieval readout and 128-dimensional retrieval embedding
head 
\\ \hline

\texttt{TransX-Ranking}
&
\texttt{TransX-Base} post-trained for preference ranking under
product-defined preference criteria using the Bradley--Terry objective
described in Section~\ref{sec:ranking-specialization}.
&
Same candidate-conditioned prediction architecture as
\texttt{TransX-Base}
\\ \hline

\texttt{TransX-GenRet}
&
\texttt{TransX-Base} post-trained for generative retrieval through
autoregressive semantic-ID generation using the generative objective
described in Section~\ref{sec:generative-specialization}.
&
Autoregressive semantic-ID decoding head
\\ \hline

\texttt{TransX-GenRank}
&
\texttt{TransX-Base} post-trained for autoregressive slate recommendation
using the generative preference objective described in
Section~\ref{sec:generative-specialization}.
&
Autoregressive next-item / slate decoding head
\\

\bottomrule
\end{tabularx}
\caption{
TransX model variants evaluated for offline model quality. All variants
reuse the same pretrained behavior tokenization, encoder, and
memory-construction backbone; task-specific differences arise after the
shared behavioral representation.
}
\label{tab:transx-model-variants}
\end{table}

We evaluate TransX and its post-training variants as described in Table~\ref{tab:transx-model-variants} before examining online deployment. All
variants use the common behavioral tokenizer, causal encoder, and
memory-construction backbone learned by \texttt{TransX-Base}; task-specific
specializations use the readout geometries and objectives defined in
Section~\ref{sec:post-training}. The common backbone configuration is
selected using the scaling analysis in Section~\ref{sec:scaling}, which has $\sim 0.1$B dense parameters and $\sim 0.3$B sparse parameters. Under our production sequence-length and request-fanout distributions, TransX executes approximately \(5\times\) the model FLOPs per request of the HSTU-style baseline. The
complete base model configurations and post-training setups are reported in
Appendix~\ref{app:model-configurations}.

\begin{figure}[htb]
    \centering
    \includegraphics[width=0.9\linewidth]{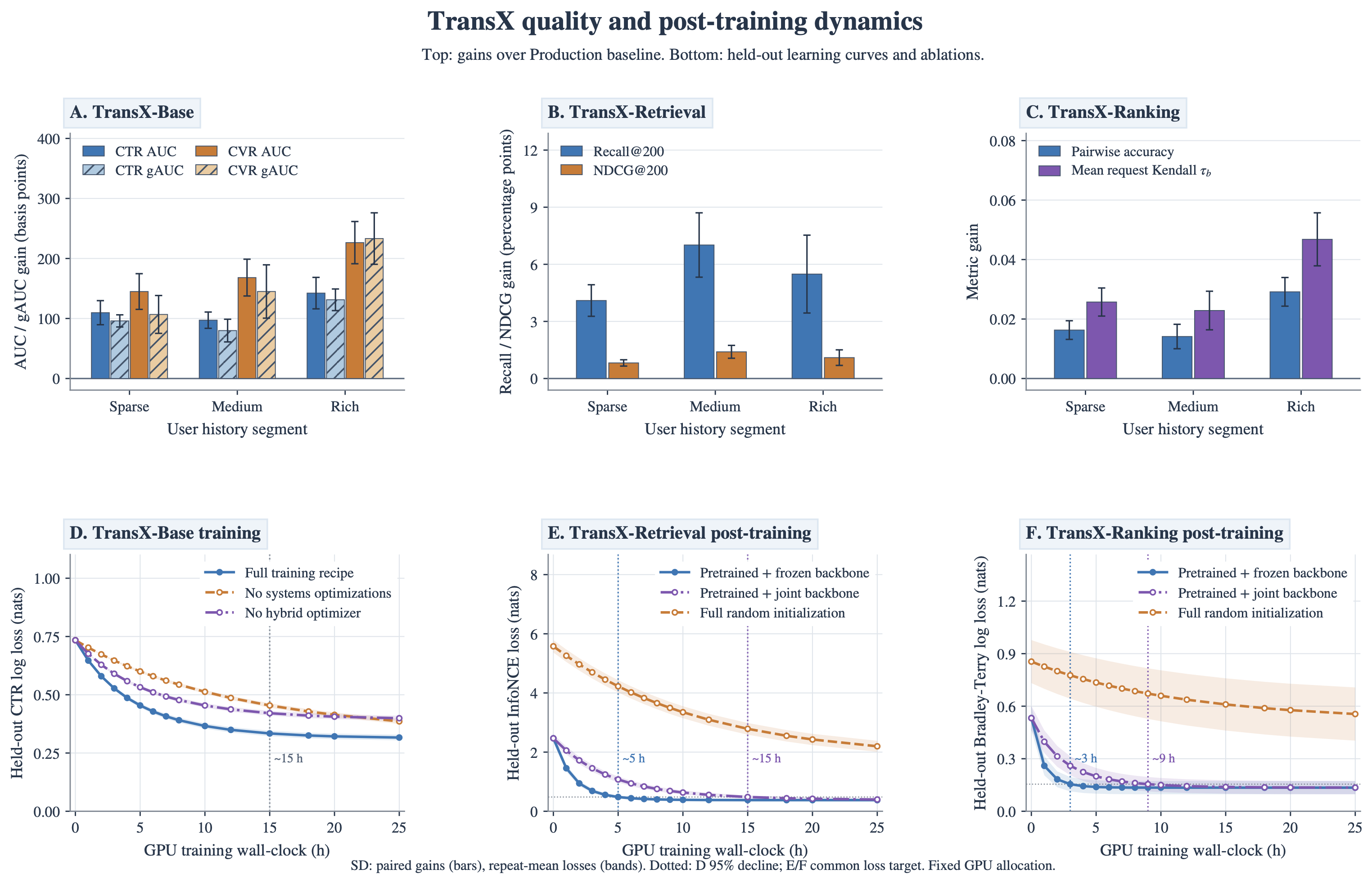}
    \caption{
    Offline quality and training dynamics for the three TransX variants.
    \textbf{(A--C)} Performance comparisons with the corresponding production
    baselines across sparse-, medium-, and rich-history users for
    \texttt{TransX-Base}, \texttt{TransX-Retrieval}, and
    \texttt{TransX-Ranking}. gAUC is the unweighted average of per-user AUC.
    \textbf{(D--F)} Held-out learning curves and training ablations.
    The optimized base-training recipe converges faster, while retrieval and
    ranking specialization benefit substantially from pretrained TransX
    representations compared with task-specific training from random
    initialization. 
    }
    \label{fig:transx-training}
\end{figure}

\begin{figure}[hbt]
    \centering
    \includegraphics[width=0.9\linewidth]{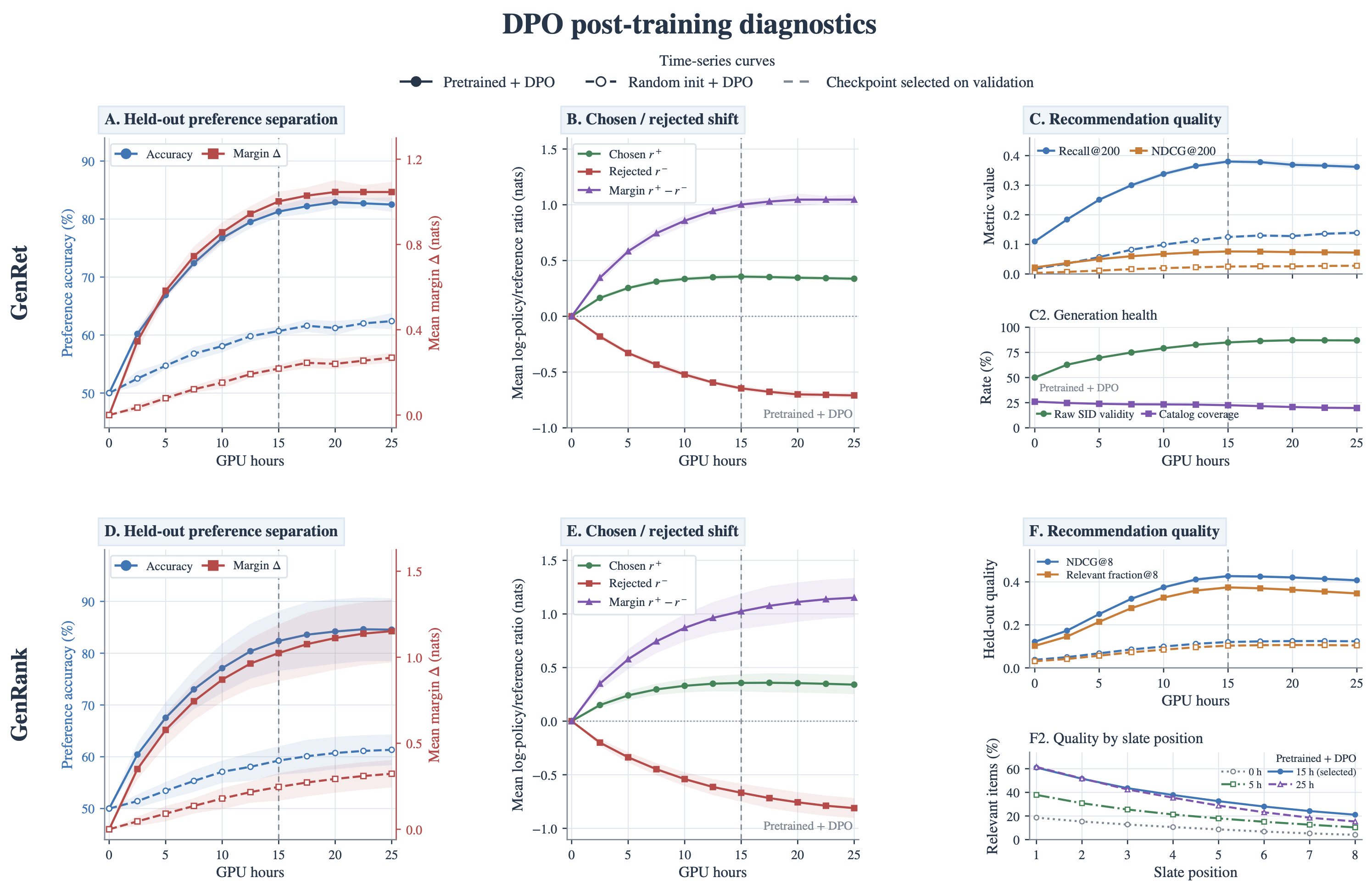}
    \caption{
    DPO post-training diagnostics for generative retrieval (top) and ranking
    (bottom).
    \textbf{(A,D)} Teacher-derived preference accuracy and chosen--rejected margin improve
    through training.
    \textbf{(B,E)} DPO progressively separates chosen and rejected responses.
    \textbf{(C,F)} Recommendation quality improves while generation validity,
    catalog coverage, and slate-level quality remain well behaved. Recall and NDCG are computed based on sampled pool.
    Dashed vertical lines indicate checkpoints selected on validation.
    }
    \label{fig:DPO-results}
\end{figure}

Figure~\ref{fig:transx-training} summarizes the primary offline results.
The upper panels compare each production-ready TransX variant with its
corresponding production baseline across sparse-, medium-, and
rich-history users. The lower panels evaluate optimization efficiency
and representation transfer. For retrieval and preference ranking, we
compare post-training initialized from \texttt{TransX-Base} with otherwise
matched training from random initialization, so that the benefit of the
shared behavioral representation can be separated from the
task-specific objective itself.

\paragraph{Core response prediction.}
\texttt{TransX-Base} improves offline response-prediction quality over the
production HSTU baseline across all evaluated history cohorts
(Figure~\ref{fig:transx-training}A). The gains are positive across the evaluated history cohorts and are largest for rich-history users, consistent with the online cohort results reported in the next
section. The optimized multi-prefix training
pipeline also reaches its held-out operating point substantially faster
than the straightforward training baseline
(Figure~\ref{fig:transx-training}D), complementing the systems
measurements in Section~\ref{sec:training-system}.

\paragraph{Retrieval specialization.}
\texttt{TransX-Retrieval} improves Recall@$K$ and NDCG@$K$ over the production
neural retrieval baseline across the evaluated behavioral cohorts, with
the smallest advantage in the sparse-history regime where little behavioral
state is available. Retrieval diversity, measured by ILD@$K$ (Intra-List Average Distance), does not
decrease. More importantly, initialization from \texttt{TransX-Base} yields
better validation quality throughout the matched GPU-hour trajectory
than training the same retrieval architecture from random
initialization (Figure~\ref{fig:transx-training}E). The result provides
direct evidence that the behavioral representation learned by the core
response task transfers to a different matching geometry rather than
being useful only to the original prediction head.

\paragraph{Preference-ranking specialization.}
\texttt{TransX-Ranking} shows the same transfer pattern. Starting from the
\texttt{TransX-Base} representation produces lower pairwise validation loss
throughout matched post-training and improves both pairwise accuracy and
event-level Kendall's $\tau$ relative to the production
LtR DLRM baseline
(Figure~\ref{fig:transx-training}C,F). Because the ranking architecture
retains candidate-conditioned TransX scoring while changing the
supervision to within-event preferences, this experiment isolates the
ability of post-training to redirect a learned behavioral
representation toward a product-defined ordering objective.

\paragraph{Generative recommendation.}
\texttt{TransX-GenRet} and \texttt{TransX-GenRank} also benefit from initialization with
the shared TransX behavioral representation and show successful
preference optimization offline (Figure~\ref{fig:DPO-results}). Their end-to-end quality, however, is more sensitive to output-space construction, preference supervision,
and autoregressive decoding: GenRet depends strongly on Semantic-ID
design and valid-ID handling, while GenRank is more exposed to
sequential error propagation and slate-preference coverage. We
therefore treat these variants as exploratory extensions rather than
production-validated results; detailed diagnostics are deferred to
Appendix~\ref{app:model-configurations}.

\subsection{Online Deployment Results}
\label{sec:online-results}




We evaluated \texttt{TransX-Base}, \texttt{TransX-Retrieval}, and
\texttt{TransX-Ranking} in three separate production A/B experiments,
each against the incumbent model on its respective product surface
(e.g., distinct landing pages or email and notification channels
within the same recommendation product). At the time of writing, \texttt{TransX-GenRet} and \texttt{TransX-GenRank} remain under offline iteration while we address the performance tradeoffs identified in the previous section. We therefore restrict the online evaluation to the three production-ready variants.

\begin{table}[hbt]
\centering
\small
\setlength{\tabcolsep}{6pt}
\renewcommand{\arraystretch}{1.12}
\begin{tabular}{@{}lccc@{}}
\toprule
& \multicolumn{3}{c}{\textbf{TransX Online A/B Results}} \\
\cmidrule(lr){2-4}
\textbf{Treatment}
& \texttt{TransX-Base}
& \texttt{TransX-Retrieval}
& \texttt{TransX-Ranking} \\
\textit{Control}
& HSTU
& EBR
& LtR DLRM \\
\midrule
CTR
& $+6.02\%$
& $+2.96\%$
& --- \\
Conversion
& $+4.45\%$
& $+0.05\%$
& $+1.33\%$ \\
DAU
& $+0.26\%$
& $+0.04\%$
& $+0.10\%$ \\
Feed views
& $+0.08\%$
& $+0.06\%$
& --- \\
\bottomrule
\end{tabular}
\caption{
Online A/B results for \texttt{TransX} variants, reported as relative
lifts over their respective control production models. Columns
represent separate experiments, not head-to-head comparisons.
All reported lifts have $p<10^{-4}$ except feed views ($p<10^{-2}$).
An em dash denotes an unreported metric.
}
\label{tab:online-product-impact}
\end{table}

For the retrieval experiment, we additionally report \emph{Candidate CTR},
defined as the number of clicks on exposed candidates attributed to the
evaluated retrieval source divided by the number of such candidate exposures;
\emph{Top-20 unique contribution} is the mean, over eligible requests, of the
number of top-20 exposed items contributed exclusively by that retrieval
source. For ranking, \emph{RankingUtility@$K$} is the user-level average of per-request top-\(K\) exposed-item utility:
$
\mathrm{RankingUtility@}K
=
\mathbb{E}_{u}\!\left[\frac{1}{|\mathcal T_u|}\sum_{t\in \mathcal T_u}
\frac{1}{K}\sum_{j=1}^{K} \text{utility}\big(\text{cand}^{(u,[j])}_t\big)
\right],
$
where $\text{cand}^{(u,[j])}_t$ denotes the candidate at ranked position $j$, $\text{utility}(\cdot)$ is the fixed production-defined utility function, and $\mathcal T_u$ is the eligible request set.

\begin{table*}[htb]
\centering
\small
\setlength{\tabcolsep}{6pt}
\renewcommand{\arraystretch}{1.15}

\begin{tabular*}{\textwidth}{
    @{\extracolsep{\fill}}lllccc@{}
}
\toprule
& & & \multicolumn{3}{c}{\textbf{User history distribution}} \\
\cmidrule(lr){4-6}
\textbf{TransX variant}
& \textbf{Control}
& \textbf{Metric}
& \textbf{Rich}
& \textbf{Medium}
& \textbf{Sparse} \\
\midrule

\multirow{2}{*}{\texttt{TransX-Base}}
& \multirow{2}{*}{HSTU}
& CTR
& $+6.2\%$
& $+6.0\%$
& $+5.5\%$ \\
& & Conversion
& $+4.8\%$
& $+4.3\%$
& $+4.1\%$ \\
\midrule

\multirow{2}{*}{\texttt{TransX-Retrieval}}
& \multirow{2}{*}{EBR}
& Top-20 unique contribution
& $+3.38\%$
& $+8.96\%$
& $+4.55\%$ \\
& & Candidate CTR
& $+6.8\%$
& $+6.0\%$
& $+6.2\%$ \\
\midrule

\multirow{2}{*}{\texttt{TransX-Ranking}}
& \multirow{2}{*}{LtR DLRM}
& Conversion
& $+1.47\%$
& $+1.25\%$
& $+1.12\%$ \\
& & RankingUtility@$K$
& $+5.66\%$
& $+4.01\%$
& $+4.32\%$ \\
\bottomrule
\end{tabular*}

\caption{
History-segmented results for the online A/B experiments summarized
in Table~\ref{tab:online-product-impact}.
Values are relative percentage lifts over the respective control
production models, stratified by user history length.
Each block reports variant-specific metrics from a separate
experiment; the blocks are not head-to-head comparisons.
}
\label{tab:online-history-segments}
\end{table*}

\paragraph{TransX-Base.}
Against the production HSTU baseline, TransX-Base improves CTR by
$6.02\%$ and conversion by $4.45\%$, with smaller positive gains in DAU
and feed views (Table~\ref{tab:online-product-impact}). The history-segmented results
in Table~\ref{tab:online-history-segments} remain positive across sparse-,
medium-, and rich-history users (with the largest gains among rich-history users), showing that the production improvement is not confined to rich-history users as we scale up both the model and the behavior trajectory.

\paragraph{TransX-Retrieval.}
Relative to the production neural retrieval baseline, TransX-Retrieval
improves CTR by $2.96\%$ while maintaining positive conversion, DAU, and
feed-view lifts (Table~\ref{tab:online-product-impact}). Candidate CTR and Top-20
unique contribution are also positive across all history cohorts
(Table~\ref{tab:online-history-segments}), providing production evidence that
post-training can specialize the shared TransX behavioral representation
for retrieval.

\paragraph{TransX-Ranking.}
Against the production LtR DLRM baseline, TransX-Ranking
improves Conversion by $1.33\%$ and DAU by $0.10\%$
(Table~\ref{tab:online-product-impact}). Conversion gains are highest for
rich-history users, while the product-defined
RankingUtility@$K$ metric improves across all history cohorts
(Table~\ref{tab:online-history-segments}). These results show that preference
post-training can redirect a response-prediction backbone toward
product-specific ranking objectives without changing the underlying
behavioral representation.


\paragraph{Sustained production behavior.}
A five-week holdout evaluation of \texttt{TransX-Base} showed that
the product gains persisted over time and across all evaluated
user segments. The sustained effect is consistent with the daily retraining and redeployment cadence. In particular, the training and prefill optimizations in Section \ref{sec:general-infra} enabled this daily update cadence, making continual model refresh practical at production scale.

\subsection{System Efficiency and Robustness}
\label{sec:deployment-system}

The system evaluation asks whether the amortization structure introduced
in Sections~\ref{sec:training-system}--\ref{sec:kv-cache} improves the
end-to-end operating point rather than merely relocating computation.
We evaluate training efficiency, online serving behavior,
freshness, and the additional state-maintenance resources introduced by
disaggregation. 

Certain deployment details and absolute cost measurements are omitted due to disclosure policies. The systems results should be interpreted within their stated measurement scopes and do not constitute a complete lifecycle-cost analysis.

\paragraph{Training efficiency.}
We compare the optimized TransX-Base training pipeline with a
straightforward FSDP~\citep{zhao2023fsdp} baseline on 32 H200 GPUs, holding the model,
training data, history-length distribution, target-sampling policy,
effective batch size, optimizer, and precision fixed. Multi-cutoff
trajectory packing, batched memory extraction, shared prefix scans, and
sequence--fanout execution increase supervised candidate-target
throughput by $84\%$ and full-pipeline training MFU from approximately $11\%$ to
$18\%$. End-to-end peak GPU memory remains approximately unchanged. Using stabilization
of click and conversion AUC and gAUC within $0.01\%$ as the convergence criterion,
wall-clock time to convergence decreases from approximately 25 to 13 hours.
%

\paragraph{Online serving efficiency.}
Production TransX-Base achieves approximately 250 requests/s/H200 with
p99 e2e scoring latency of 49~ms. The GPU-resident behavior-memory KV hit rate
is approximately $25\%$, while combined GPU/host residency exceeds
$99\%$; behavior memory KV transfer contributes less than $5\%$ of request
feature-transfer volume. Under the production traffic distribution,
online serving operates at approximately $55\%$ model FLOPs
utilization.

\begin{figure}[htb]
    \centering
    \includegraphics[width=\linewidth]{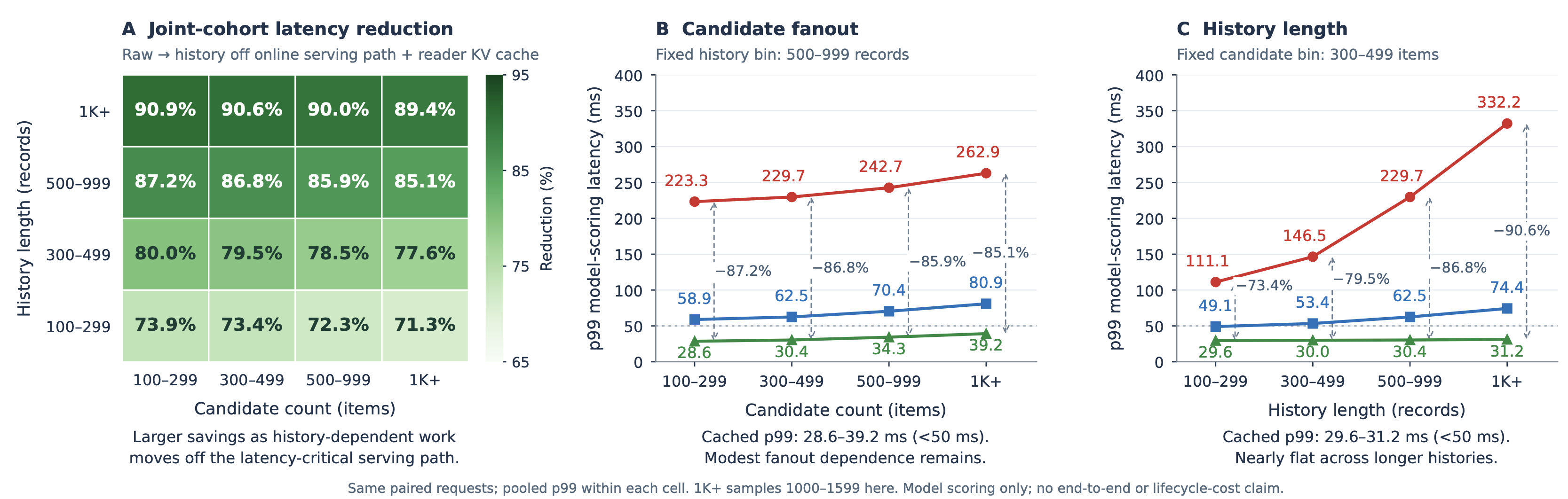}
    \caption{
    Serving-efficiency analysis. In panels B--C, red denotes colocated serving, blue denotes disaggregated serving without KV caching, and green denotes disaggregated serving with KV caching.
    \textbf{(A)} Reduction in p99 model-scoring latency after
    moving history-dependent computation off the latency-critical serving path and reusing
    materialized behavior memory KV state.
    \textbf{(B)} At fixed history length, cached p99 latency grows only
    modestly with candidate fanout.
    \textbf{(C)} At fixed candidate fanout, cached p99 latency remains
    nearly flat as history length increases.
    Together, the results show that online latency is largely decoupled
    from behavioral-history length and is driven primarily by
    candidate-conditioned scoring.}
    \label{fig:serving-efficiency}
\end{figure}

Figure~\ref{fig:serving-efficiency} isolates the effect of moving
history-dependent computation off the latency-critical serving path. Relative to a raw
online serving path that performs history processing synchronously, cached
behavior memory KV serving reduces p99 model-scoring latency by $71$--$91\%$
across the measured history-length and fanout cohorts. At fixed history
length, latency increases primarily with candidate fanout; at fixed
candidate fanout, cached p99 latency remains nearly flat as history
length grows from hundreds to more than one thousand records. This
matches the online cost decomposition of Section~\ref{sec:disaggregated-serving}: once behavior memory KV is materialized, prediction cost is governed primarily by candidate-conditioned work rather than full behavioral-history length.

Finally, we benchmark an HSTU-style baseline configured to approximately match the deployed TransX in per-request model FLOPs. The matched-FLOP HSTU-style baseline uses the same tokenizer and prediction head as TransX, with the sequential backbone scaled to match average per-request model FLOPs under the same production sequence-length/fanout distribution, precision, H200 hardware, and latency SLO. The two models achieve comparable offline quality (click AUC within $1\times 10^{-3}$), while TransX delivers approximately \(4\times\) higher serving throughput in requests per H200.

\paragraph{Freshness--quality trade-off.}
Disaggregation removes encoder advancement from the online serving
path but introduces a freshness constraint on materialized state. In the
TransX-Base production ramp, p50 and p90 nearline update lags are
approximately $1.2$~s and $1.8$~s, respectively, and more than $92\%$
of eligible traffic is served from the latest compatible snapshot.
Approximately $3\%$ of traffic uses the fallback model, primarily for
cold-start users without sufficient behavioral history.

\begin{table}[htb]
\centering
\footnotesize
\setlength{\tabcolsep}{5pt}
\renewcommand{\arraystretch}{1.12}

\begin{tabular}{lcccc}
\toprule
& \multicolumn{4}{c}{\textbf{Increase in p90 behavior memory KV freshness lag}} \\
\cmidrule(lr){2-5}
\textbf{Metric}
& \textbf{1.8\,s $\rightarrow$ 1.2\,min}
& \textbf{1.2\,min $\rightarrow$ 18\,min}
& \textbf{18\,min $\rightarrow$ 76\,min}
& \textbf{76\,min $\rightarrow$ $\sim$22\,h} \\
\midrule
Overall CTR
& $-0.21\%$
& $-0.12\%$
& $-0.03\%$
& $<0.01\%$ \\

In-session CTR
& $-0.37\%$
& $-0.18\%$
& $-0.05\%$
& $<0.01\%$ \\
\bottomrule
\end{tabular}

\caption{
Online ablation of behavior memory KV freshness. Each column reports the
incremental CTR change when increasing the observed p90 freshness lag
from one regime to the next. Performance degrades gradually with
staleness, with a larger effect on in-session CTR at short freshness
timescales.
}
\label{tab:freshness-ablation}
\end{table}

To quantify the predictive cost of staleness, we vary the allowable
behavior memory KV freshness lag online and report the resulting observed p90
lag in Table~\ref{tab:freshness-ablation}. CTR degrades gradually as
state becomes older, with the largest sensitivity occurring when lag
moves from seconds to minutes and a larger effect on in-session CTR.
This establishes freshness as an explicit quality--systems trade-off:
nearline maintenance need not be request synchronous, but its service
level materially affects short-horizon recommendation quality.

\paragraph{End-to-end resource accounting.}
Disaggregation reduces online GPU consumption from $5.8$ to $1.1$
GPU-hours per million requests ($-81.2\%$) without increasing storage
or transfer on the online serving path. In colocated serving, scoring stores and
transfers the full pre-tokenization behavioral history and repeatedly
executes history-dependent computation; TransX instead materializes a
compact behavior-memory KV snapshot for online reuse. As a result,
persistent storage and median network traffic remain approximately flat,
while the remaining history-state construction and refresh cost is moved
to bulk and nearline paths outside the online serving path. The gain
therefore comes from both avoiding repeated history computation and
replacing the online history representation with compact reusable state. After converting online GPU, nearline CPU, daily bulk compute, storage, and network resources to a common normalized production-cost basis at matched traffic and SLO, total serving cost is reduced to 22\% (colocated \(=1.00\), TransX \(0.22\)).

\begin{table}[t]
\centering
\footnotesize
\begin{tabular}{lccc}
\toprule
\textbf{Resource} & \textbf{Colocated} & \textbf{TransX} &
\textbf{Relative} \\
\midrule
Online GPU-hours / M requests
    & 5.8 & 1.1 & $-81.2\%$ \\
Nearline CPU-hours / M behavior updates
    & -- & 7.5 & -- \\
Offline bulk GPU-hours / deployment
    & -- & 3.2 & -- \\
Continuation state $C$: p50/p95/p99 bytes/user
    & -- & 496/1824/4125 KiB/user & -- \\
Behavior-memory KV $V$: p50/p95/p99 bytes/user
    & -- & 8.5/12.3/15.1 KiB/user & -- \\
Host memory per worker
    & 60 GiB & 60 GiB & $0.0\%$ \\
Persistent storage
    & 0.312 PB & 0.308 PB & $-1.3\%$ \\
Network bytes (p50) / request
    & 4.12 MiB & 3.96 MiB & $-3.9\%$ \\
Nearline bytes (p50) / behavior update
    & -- & 1.21 MiB & -- \\
\bottomrule
\end{tabular}
\caption{
\textbf{End-to-end production resource accounting for colocated and
disaggregated TransX serving.}
TransX reduces online GPU consumption by $81.2\%$ while storage and
network traffic on the online serving path remain approximately flat because compact
behavior-memory KV replaces the full pre-tokenization behavioral history consumed by
colocated scoring. Additional state-production cost is isolated to the
bulk and nearline paths. Nearline state maintenance runs entirely on CPUs in the measured production deployment; accordingly, the nearline CPU-hours account for the full compute consumption of the nearline stage. Finally, note that online, nearline, and bulk rows use different natural denominators and are not additive.
}
\label{tab:system-resource-accounting}
\end{table}

\paragraph{Operational robustness and independent scaling.}
TransX treats stale or unavailable behavioral state as a degradation
condition rather than allowing recovery to enter the online serving
path. Requests consume a compatible immutable snapshot when available,
otherwise degrade to an admissible older snapshot or the designated
non-sequential fallback, while continuation-state reconstruction proceeds
asynchronously. Controlled failure-injection experiments in
Appendix~\ref{app:system-robustness} show that injected failures primarily
reduce freshness or increase fallback traffic rather than causing
request failure or synchronous history reconstruction.

Disaggregation also allows bulk construction, nearline maintenance, and
online prediction to scale independently according to their own
bottlenecks. Bulk capacity follows rollout backlog, nearline capacity
follows update load and freshness lag, and online capacity follows
request traffic, candidate fanout, and latency pressure. Appendix
\ref{app:system-robustness} reports the corresponding horizontal-scaling
efficiency and failure-recovery measurements.

\section{Related Work}
\label{sec:related-work}

\begin{table}[htbp]
\centering
\small
\setlength{\tabcolsep}{4pt}
\renewcommand{\arraystretch}{1.10}
\begin{tabularx}{\linewidth}{
  @{}>{\raggedright\arraybackslash}p{0.13\linewidth}
  >{\raggedright\arraybackslash}X
  >{\raggedright\arraybackslash}X
  >{\raggedright\arraybackslash}X@{}
}
\toprule
\textbf{Work}
&
\textbf{Shared/materialized state}
&
\textbf{Documented reuse/refresh scope}
&
\textbf{Prediction-specific workload}
\\
\midrule
MIMN/UIC
&
Fixed-size user interest memory
&
Behavior-triggered; cross-request
&
Online scoring
\\
HSTU
&
Causal sequential state/cache
&
Candidate microbatches; cross-request reuse
&
Transduction/scoring
\\
SequenceO1
&
Compressed user sketch
&
Targets; examples; requests
&
Recent branch and target-conditioned reasoning
\\
OneTrans-V2
&
Causal user context
&
Training prefixes across requests; cascade stages
&
Stage-specific candidate computation
\\
LLaTTE
&
Compressed user representation
&
Asynchronous, selected behavior-event refresh
&
Online target-aware ranker
\\
SOLARIS
&
User--item interaction representation
&
Speculative pair reuse; TTL refresh
&
Off-request pair inference; downstream ranking
\\
\textbf{TransX}
&
Continuation state and bounded behavior-memory KV
&
Shared trajectory prefixes; behavior-triggered refresh;
cross-request snapshot reuse
&
Candidate-conditioned fusion
\\
\bottomrule
\end{tabularx}
\caption{
State, reuse, and prediction-specific work in the nearest related systems.
Scopes include training and serving where specified, not performance
rankings or exclusive capabilities.
}
\label{tab:related-scaling}
\end{table}

\paragraph{Long-history modeling and reusable behavioral state.}
Industrial recommenders amortize behavioral computation through persistent
user state, selective history access, and candidate-independent
preprocessing. SASRec and BERT4Rec established self-attentive and
bidirectional history modeling \citep{kang2018sasrec,sun2019bert4rec};
industrial systems extend this line toward longer histories, high
candidate fanout, and serving constraints. MIMN/UIC maintains
behavior-triggered user memory \citep{pi2019mimn}; SIM and TWIN select
target-relevant history, with TWIN V2 adding offline clustering
\citep{pi2020sim,chang2023twin,si2024twinv2}; TransAct combines realtime
and batch representations \citep{xia2023transact}. HSTU is particularly
close in causal sequential transduction and M-FALCON's candidate-side
amortization \citep{zhai2024actions}; LONGER reduces sequence cost through
token merging and cached attention \citep{chai2025longer}; and SequenceO1
separates reusable compressed sketches from target-conditioned reasoning
\citep{guan2026sequenceo1}. TransX builds on this direction by exposing
behavior-memory resolution as an independent scaling coordinate and by
separating mutable continuation state from published behavior-memory KV.
This runtime state also differs from LMN's shared trainable memory
parameters \citep{lu2025lmn}.

\paragraph{Scaling architectures and compute allocation.}
Wukong and RankMixer scale feature interaction
\citep{zhang2024wukong,zhu2025rankmixer}; OneTrans unifies sequence and
feature interaction within one Transformer \citep{zhang2025onetrans},
whereas TransX preserves the boundary between candidate-independent
behavioral construction and candidate-conditioned fusion. This operator
view is a comparison lens, not a claim that these architectures are
strict special cases of TransX. Their internal operators map to different
choices for, or collapses of, encoding, memory, and candidate-conditioned
interaction, without implying plug-and-play compatibility.
Sequential scaling studies examine model and data size
\citep{zhang2024scaling}; MixFormer and UniScale also study component
balance \citep{huang2026mixformer,yu2026uniscale}. LLaTTE is especially
related in allocating asymmetric compute between an asynchronous user
model and an online ranker through a fixed representation interface
\citep{xiong2026llatte}. TransX broadens this allocation question to
behavioral-encoder compute $C_E$, behavior-memory budget $B_M$, and
candidate-conditioned fusion compute $C_F$, asking how their relative
returns change with behavioral support and supervision. The focus is a controlled
allocation surface rather than a universal one-dimensional scaling law:
it exposes both where capacity is placed and how much behavioral state
is available to prediction, while the execution design controls when
each computational role is executed.

\paragraph{Training reuse and capability specialization.}
RecD deduplicates features end to end \citep{zhao2023recd}; MTGR and
Douyin batching share user-side computation across grouped examples
\citep{han2025mtgr,guan2026long}. OneTrans-V2 and TransX both exploit causal
prefix sharing \citep{cao2026onetransv2}; OneTrans-V2 uses it within a
jointly optimized retrieval-to-ranking cascade, whereas TransX couples
prefix reuse to bounded event memory, candidate-fanout execution, and
subsequent capability specialization. MTFM and LUM likewise share
behavioral representations across scenarios or tasks
\citep{song2026mtfm,yan2026lum}. Generative recommenders use semantic-ID
retrieval or slate generation \citep{rajput2023tiger,jiang2019listcvae};
TransX treats these as capability-specific output geometries over the
shared behavioral interface. In TransX's reference recipe, core
response-prediction training learns the behavioral backbone; post-training
freezes the encoder and memory construction and optimizes capability-specific
readers and heads. Partial adaptation and joint training remain
alternatives; KV projections can remain capability-specific.

\paragraph{Amortized serving and execution clocks.}
MIMN/UIC updates state on behavior events; HSTU and OneTrans reuse causal
computation across candidates or requests. LLaTTE refreshes user
representations on selected high-value behavior events, whereas SOLARIS
speculatively caches user--item interaction features
\citep{xiong2026llatte,liu2026solaris}. PagedAttention/vLLM and DistServe
address KV management and prefill/decode disaggregation
\citep{kwon2023pagedattention,zhong2024distserve}. These are adjacent
systems analogies: LLM KV is typically request/session scoped, whereas
TransX maintains persistent per-user state driven by an independent
behavioral stream.
TransX distinguishes mutable continuation state $C$ from immutable,
version-compatible published behavior-memory KV $V$, aligning bulk
construction, behavior-triggered refresh, and online prediction with
different execution clocks. The distinguishing feature is that the same
compute–state decomposition governs model
scaling, training reuse, serving-state lifetimes, and capability
specialization.
Table~\ref{tab:related-scaling} compares the resulting state and reuse
boundaries.

\section{Conclusion, Limitations, and Future Directions}
\label{sec:conclusion}

\paragraph{Conclusion.}
TransX is organized around two coupled questions: where should additional model
capacity be allocated to improve recommendation quality, and must
increasing that capacity proportionally increase online serving cost?
The model's compute--state decomposition separates sequence encoding, behavior
memory, and candidate-conditioned fusion into distinct scaling axes,
while the infrastructure preserves the same boundaries in training and
serving. Behavior computation is shared across overlapping prefixes,
recomputed as behavioral state changes, and materialized behind a
bounded behavior-memory interface, while request-time computation is
concentrated on candidate-conditioned prediction.

Within the evaluated regimes, marginal returns differ substantially
across these scaling axes. Encoder scaling improves prediction at fixed
behavior-memory resolution but saturates beyond moderate depth at fixed
fusion capacity; increasing supervision shifts the
allocation frontier toward candidate-conditioned fusion; and the memory resolution
needed to preserve quality remains roughly stable as overall model
capacity grows. Together with disaggregated serving, these results
suggest that efficient recommendation scaling is a joint problem of
compute allocation and amortized/disaggregated execution: additional capacity should
be allocated according to its marginal return rather than increased
proportionally across all computational roles, while reusable
computation should remain off the request-critical path.

The same behavior-memory interface also supports capability
specialization. Starting from a core response-prediction backbone, we
post-train capability-specific memory readers and output heads for
retrieval and preference ranking, with generative recommendation
explored offline. Production deployment of the core, retrieval, and
ranking variants further validates the framework across recommendation
quality, training efficiency, and serving economics.

\paragraph{Limitations and future directions.}
The observed scaling patterns are empirical and conditional on the
studied data, supervision regimes, architecture family, and deployment
envelope; they do not constitute a universal scaling law or uniquely
identify the mechanisms behind the patterns. Materialized behavioral encoder continuation and memory state maintenance
also introduce freshness, storage, compatibility, and recovery costs.
Finally, for task specialization via post-training, a shared behavior-memory interface and frozen backbone need not be optimal for every task. Future work should therefore study: (1) adaptive allocation across encoder capacity, behavior-memory resolution, and candidate-conditioned fusion; and (2) when the behavioral memory interface and model backbone should remain shared, be partially adapted, or become task-specific.

\begin{quote}
    \emph{More broadly, the central scaling question for industrial recommendation is not only how much computation to add, but where the capacity should be allocated, when each computation should be executed, and over what scope the resulting states can be reused.}
\end{quote}


\clearpage
\phantomsection
\addcontentsline{toc}{section}{References}
\bibliographystyle{plainnat}
\bibliography{references}

\clearpage
\appendix
\numberwithin{equation}{section}

\section{Training Addendum}
\label{sec:app-examples}

\begin{remarkbox}{Example: From record-wise logs to one packed training unit}

Consider a user with the following serving logs:
\[
\begin{array}{c|c|c|c|c}
\text{event} & \text{candidate} & \text{features} &
\text{label} & \text{behavior cutoff} \\
\hline
t=1 & c_1^{(1)} & x_1^{(1)} & 1 & l_1=4 \\
t=1 & c_1^{(2)} & x_1^{(2)} & 0 & l_1=4 \\
t=1 & c_1^{(3)} & x_1^{(3)} & 0 & l_1=4 \\
\hline
t=2 & c_2^{(1)} & x_2^{(1)} & 0 & l_2=7 \\
t=2 & c_2^{(2)} & x_2^{(2)} & 1 & l_2=7 \\
\hline
t=3 & c_3^{(1)} & x_3^{(1)} & 0 & l_3=9 \\
t=3 & c_3^{(2)} & x_3^{(2)} & 0 & l_3=9
\end{array}
\]

Together with the behavior trajectory
\[
\mathcal H^{(u)}
=
(h_1,h_2,\ldots,h_9),
\]
these rows are packed as
\[
\mathcal D^{(u)}
=
\left(
\mathcal H^{(u)},
\left\{
(\mathcal E_1,l_1),
(\mathcal E_2,l_2),
(\mathcal E_3,l_3)
\right\}
\right),
\]
where, for example,
\[
\mathcal E_1
=
\big\{(e_1^{(1)},1),\, (e_1^{(2)},0),\, (e_1^{(3)},0) \big\},
\qquad
e_1^{(i)}=\big(c_1^{(i)}, x_1^{(i)} \big), \, i=1,2,3,
\]
and we can use behavior $(h_1,h_2,h_3,h_4)$ to construct the behavior memory for event $t=1$ (since its behavior cutoff is $l_1=4$).
In general, all candidates in event $t$ share
\[
M_t
=
\mathcal M
\left(
B[1{:}l_t,:]
\right)
\longrightarrow
\left\{
e_t^{(u,j)}
\right\}_{j=1}^{m_t}.
\]
Event grouping therefore also provides the natural unit for pairwise
and listwise ranking objectives introduced in
Section~\ref{sec:post-training}.
\end{remarkbox}

\subsection{Hybrid Optimizer Detail}
\label{app:optimizer}

TransX combines densely updated network parameters with large,
row-sparse ID embedding tables. We optimize dense parameters with AdamW
and row-sparse embeddings with row-wise Adagrad. 

\begin{table}[hbt]
\centering
\small
\setlength{\tabcolsep}{4pt}
\renewcommand{\arraystretch}{1.12}
\begin{tabularx}{\linewidth}{
  @{}p{0.2\linewidth}
     p{0.18\linewidth}
     X@{}
}
\toprule
\textbf{Parameter regime}
& \textbf{Optimizer}
& \textbf{Update semantics} \\
\midrule

Dense network parameters
&
AdamW
&
Encoder, fusion, prediction, and other densely updated parameters.
First- and second-moment adaptation with decoupled weight decay. \\

\addlinespace

Row-sparse ID embeddings
&
Row-wise Adagrad
&
One squared-gradient accumulator per row, shared across its
coordinates. Only rows receiving gradient contributions participate
in updates; no momentum or embedding weight decay is applied. \\

\bottomrule
\end{tabularx}
\caption{Optimizer assignment by parameter gradient structure.}
\label{tab:dense-sparse-optimization}
\end{table}

\paragraph{Row-wise adaptation.}
For embedding row $e_i\in\mathbb R^{d_i}$, let $g_i$
be its complete logical-step gradient after loss normalization,
mixed-precision unscaling, and any configured clipping. Row-wise
Adagrad maintains a scalar accumulator $a_i$ and applies
\begin{equation}
\begin{aligned}
a_i^{+}
&=
a_i+\frac{1}{d_i}\|g_i\|_2^2,\\
e_i^{+}
&=
e_i
-
\eta_E
\frac{g_i}{\sqrt{a_i^{+}}+\epsilon_E},
\end{aligned}
\label{eq:row-wise-adagrad}
\end{equation}
where \(\eta_E\) and \(\epsilon_E\) are the embedding learning rate and numerical stabilizer, respectively.
The shared denominator preserves the within-row gradient direction while
reducing optimizer state from $nd$ scalars for coordinate-wise Adagrad
to $n$ scalars for an $n\times d$ embedding table. This is a distinct
optimizer rather than a compressed representation of coordinate-wise
Adagrad.

All gradient contributions to a parameter or embedding row are aggregated at the
logical update boundary before its optimizer step, including
contributions from shared histories, event cutoffs, candidate fanout,
and gradient-accumulation microbatches.

\subsection{Blocked Prefix-Scan Realization}
\label{app:blocked-prefix-scan}

Section~\ref{sec:behavior-memory} writes the global branch abstractly as
block-wise pooling over the history preceding the recent local window.
This section gives the corresponding execution used to construct the
same cutoff-specific local--global memories jointly for multiple events
on one encoded trajectory.

\paragraph{Exact local--global memory.}
For user $u$ and event $t$, let $l^{(u)}_t$ denote the eligible behavioral
cutoff and $W$ the local-memory width. Define
\[
    b^{(u)}_t
    =
    \max(0,l_{t}^{(u)}-W),
\]
so that the local branch contains
\[
    L^{(u)}_t
    =
    B^{(u)}[b^{(u)}_t+1:l^{(u)}_t],
    \qquad
    W^{(u)}_t=\min(W,l^{(u)}_t).
\]
The preceding support $[1,b^{(u)}_t]$ is partitioned according to the
configured global-memory policy into at most $R$ contiguous,
non-overlapping causal blocks
\[
    \mathcal P^{(u)}_t
    =
    \{
        \mathcal I_{t,r}^{(u)}
    \}_{r=1}^{R^{(u)}_t},
    \qquad
    R^{(u)}_t\le R,
\]
with
\[
    \mathcal I_{t,r}^{(u)}\subseteq[1,b^{(u)}_t],
    \qquad
    \mathcal I_{t,r}^{(u)}\cap\mathcal I_{t,r'}^{(u)}=\varnothing
    \quad (r\neq r').
\]
Each block is reduced independently,
\[
    g_{t,r}^{(u)}
    =
    \operatorname{Pool}
    \left(
        \{b_i^{(u)}:
          i\in\mathcal I_{t,r}^{(u)}\}
    \right),
\]
giving the global summary bank
\[
    G^{(u)}_t
    =
    [g_{t,1}^{(u)},\ldots,g_{t,R^{(u)}_t}^{(u)}]^\top .
\]
The event memory is therefore
\[
    M^{(u)}_t
    =
    \psi_{\mathrm{mem}}
    \left(
        [G^{(u)}_t;
         L^{(u)}_t]
    \right),
    \qquad
    n^{(u)}_t=R^{(u)}_t+W^{(u)}_t\le R+W.
\]
Thus the local and global branches have disjoint support: recent states
are retained explicitly, while only preceding behavior contributes to
the global summaries. The particular block widths and maximum numbers
of local and global states used in the evaluated configurations are
reported in Appendix~\ref{app:model-configurations}.

\paragraph{Associative block pooling.}
The global pooler is implemented through an associative sufficient
state, allowing block summaries to be constructed without repeatedly
materializing overlapping behavioral prefixes. For the learned
attention pooler used in our implementation, let
\[
    a_i=a_\theta(b_i)
\]
denote the scalar pooling score associated with encoded state $b_i$.
For any contiguous segment $\mathcal I$, define
\[
    m_{\mathcal I}
    =
    \max_{i\in\mathcal I} a_i,
\]
\[
    Z_{\mathcal I}
    =
    \sum_{i\in\mathcal I}
    \exp(a_i-m_{\mathcal I}),
\]
and
\[
    V_{\mathcal I}
    =
    \sum_{i\in\mathcal I}
    \exp(a_i-m_{\mathcal I})\,b_i .
\]
The sufficient state
\[
    S_{\mathcal I}
    =
    (m_{\mathcal I},Z_{\mathcal I},V_{\mathcal I})
\]
produces the pooled representation
\[
    \operatorname{Pool}(\mathcal I)
    =
    \frac{V_{\mathcal I}}{Z_{\mathcal I}}.
\]

For adjacent segments $\mathcal I_1$ and $\mathcal I_2$, their states
can be merged exactly. Let
\[
    m
    =
    \max(m_{\mathcal I_1},m_{\mathcal I_2}).
\]
Then
\[
\begin{aligned}
    Z
    &=
    e^{m_{\mathcal I_1}-m}Z_{\mathcal I_1}
    +
    e^{m_{\mathcal I_2}-m}Z_{\mathcal I_2},\\
    V
    &=
    e^{m_{\mathcal I_1}-m}V_{\mathcal I_1}
    +
    e^{m_{\mathcal I_2}-m}V_{\mathcal I_2}.
\end{aligned}
\]
We write
\[
    S_{\mathcal I_1\cup\mathcal I_2}
    =
    S_{\mathcal I_1}\otimes S_{\mathcal I_2}.
\]
The merge operator $\otimes$ is associative, so segment statistics can
be constructed by a parallel scan or reduction while producing the
same pooled block representation as direct attention pooling.

\paragraph{Blocked prefix scan across event cutoffs.}
A packed trajectory generally contains multiple supervised events whose
eligible prefixes overlap. Rather than evaluate
$\operatorname{Pool}$ independently for every block of every event, we
first compute the token-level pooling statistics once along
$B^{(u)}$ and organize them according to the block boundaries required
by the selected event cutoffs.

Define \(S_u\) as the number of atomic segments. Conceptually, let
\[
    \mathcal A^{(u)}
    =
    \{\mathcal A^{(u)}_s\}_{s=1}^{S_u}
\]
denote the finest contiguous segmentation induced by the union of the
required global-block boundaries. Each target block
$\mathcal I_{t,r}^{(u)}$ is then a contiguous union of elements of
$\mathcal A^{(u)}$. The execution first computes
\[
    S_{\mathcal A^{(u)}_s}
    =
    \operatorname{BlockScan}
    \left(
        \{b_i^{(u)}:
          i\in\mathcal A^{(u)}_s\}
    \right)
\]
once for each atomic segment. A requested global block is reconstructed
by associative merging,
\[
    S_{t,r}^{(u)}
    =
    \bigotimes_{
        \mathcal A_s^{(u)}
        \subseteq
        \mathcal I_{t,r}^{(u)}}
    S_{\mathcal A_s^{(u)}},
\]
followed by
\[
    g_{t,r}^{(u)}
    =
    \operatorname{Readout}(S_{t,r}^{(u)})
    =
    \frac{V_{t,r}^{(u)}}{Z_{t,r}^{(u)}}.
\]
Collecting these readouts yields
\[
    G_{t}^{(u)}
    =
    [g_{t,1}^{(u)},\ldots,g_{t,R_{t}^{(u)}}^{(u)}]^\top .
\]

The recent local states are constructed independently through the
batched masked gather. Event memory is then
materialized exactly as
\[
    M_{t}^{(u)}
    =
    \psi_{\mathrm{mem}}
    \left(
        [G_{t}^{(u)};
         L_{t}^{(u)}]
    \right),
\]
followed once per event by the memory adapter and projected KV
construction. Candidate fanout occurs only after this shared event state
has been formed.

This organization computes token-side pooling statistics once and
shares them across overlapping event prefixes. Ignoring projection
costs, the history-side work scales with the encoded trajectory length
plus the number of exposed memory slots,
\[
O(\ell_u d)
+
O\!\left(
d\sum_{t=1}^{T_u}(R_{t}^{(u)}+W_{t}^{(u)})
\right),
\]
rather than with the sum of independently materialized prefix lengths
\[
    O\!\left(
        d\sum_{t=1}^{T_u}l_{t}^{(u)}
    \right).
\]
The exact constant factors depend on the configured block partition and
scan implementation.

\paragraph{Backward aggregation.}
For a map $y=f(x)$ and upstream gradient
$\nabla_y\mathcal L$, we write
\[
\operatorname{VJP}_f
\left(
x;\nabla_y\mathcal L
\right)
=
J_f(x)^\top \nabla_y\mathcal L.
\]
After candidate losses have been reduced at the event boundary, the
backward path through the shared KV projection and memory adapter described
in Section~\ref{sec:training-system} produces an event-memory
gradient $\nabla_{M_{t}^{(u)}}\mathcal L$ for each selected cutoff.

Because
\[
M_{t}^{(u)}
=
\psi_{\mathrm{mem}}
\!\left(
[G_{t}^{(u)};L_{t}^{(u)}]
\right),
\]
one reverse-mode application of the memory-construction map gives
\[
\left(
\nabla_{G_{t}^{(u)}}\mathcal L,
\nabla_{L_{t}^{(u)}}\mathcal L
\right)
=
\operatorname{VJP}_{\psi_{\mathrm{mem}}}
\!\left(
G_{t}^{(u)},L_{t}^{(u)};
\nabla_{M_{t}^{(u)}}\mathcal L
\right).
\]

The local-memory gradients are accumulated into the shared encoded
trajectory by the adjoint of the batched gather:
\[
\left.\nabla_{B^{(u)}}\mathcal L\right|_{\mathrm{local}}
=
\operatorname{ScatterAdd}
\!\left(
\{\nabla_{L_{t}^{(u)}}\mathcal L\}_{t=1}^{T_u},
I^{(u)};A^{(u)}
\right),
\]

For the global branch, each
\[
G_{t}^{(u)}
=
[g_{t,1}^{(u)},\ldots,g_{t,R_{t}^{(u)}}^{(u)}]^\top
\]
contributes gradients
$\nabla_{G_{t}^{(u)}}\mathcal L$ to the corresponding block readouts.
These gradients are propagated through the block readouts, associative
merge trees, and shared atomic-segment scan:
\[
\left.\nabla_{B^{(u)}}\mathcal L\right|_{\mathrm{global}}
=
\operatorname{VJP}_{\mathrm{BlockScan+Readout}}
\!\left(
B^{(u)};
\{\nabla_{G_{t}^{(u)}}\mathcal L\}_{t=1}^{T_u}
\right).
\]
The shared encoded trajectory then receives
\[
\nabla_{B^{(u)}}\mathcal L
=
\left.\nabla_{B^{(u)}}\mathcal L\right|_{\mathrm{local}}
+
\left.\nabla_{B^{(u)}}\mathcal L\right|_{\mathrm{global}}.
\]
Because gradients from all event cutoffs are summed at shared trajectory
positions, this backward realization is equivalent to independently
constructing and differentiating every cutoff-specific event memory, up to
the numerical precision of the implemented kernels. Gradients with respect
to trainable pooling and memory-construction parameters are accumulated by
the same reverse-mode operations and are omitted above for brevity.

For an individual global block,
\[
\nabla_{S_{t,r}^{(u)}}\mathcal L
=
\operatorname{VJP}_{\mathrm{Readout}}
\!\left(
S_{t,r}^{(u)};
\nabla_{g_{t,r}^{(u)}}\mathcal L
\right).
\]
and gradients from all target blocks containing a shared atomic segment are
summed before propagating through the underlying block scan.

\paragraph{Execution-equivalence contract.}
For every event cutoff $\ell_{t}^{(u)}$, the blocked implementation produces
the same forward memory
\[
\widetilde M_{t}^{(u)}
=
\mathcal M\!\left(B^{(u)}_{1:\ell_{t}^{(u)}}\right)
=
M_{t}^{(u)},
\]
and, for any event-memory gradients
$\{\nabla_{M_{t}^{(u)}}\mathcal L\}_{t=1}^{T_u}$, its reverse-mode execution produces the
same accumulated trajectory gradient
\[
\nabla_{B^{(u)}}\mathcal L
=
\sum_{t=1}^{T_u}
J_{\mathcal M_t}\!\left(B^{(u)}\right)^\top
\nabla_{\widetilde M_{t}^{(u)}}\mathcal L.
\]
as independently materializing and differentiating each cutoff-specific
memory.
Thus the blocked prefix scan changes only the physical realization of
overlapping forward and backward computation; it does not change the
local/global supports, pooling rule, reader-visible state, or resulting
gradient.

\paragraph{Global-memory overflow.}
When the pre-local history would require more than $R$ global
blocks, adjacent blocks are compacted pairwise using the associative
pooling operator, doubling the effective block width until the complete
pre-local history fits in at most $R$ summaries. Thus no retained
behavior is discarded solely because the global-memory bank is full.

For example, with the scaling analysis configuration
$L_{\text{max}}=4096$, $W=32$, $R=32$, and base width $64$, the
$4064$ pre-local states exceed the nominal $32\times64=2048$ support.
One compaction level increases the effective width to $128$, yielding
31 full 128-state blocks and one 96-state partial block. These produce
32 global summaries, which together with the 32 local states give the
64-state reader memory.

During incremental maintenance, newly eligible global states extend the
newest block. When another block would exceed the $R$-slot budget,
adjacent blocks are merged to the next compaction level using the same
associative sufficient statistics. Bulk rebase reconstructs the identical
partition from the retained trajectory.
\section{Additional Model--Infrastructure Details}
\label{app:infra-details}

We provide the additional model-infrastructure co-design details in this part of the appendix.

\subsection{Nearline Infrastructure Detail}
\label{app:snapshot-lifecycle}

Section~\ref{sec:nearline-maintenance} treats nearline maintenance as a
state transition from mutable encoder continuation state to immutable
behavior memory state. This appendix specifies the corresponding publication,
recovery, storage, and model-transition mechanics.

\begin{figure}[hbt]
    \centering
    \includegraphics[width=0.9\linewidth]{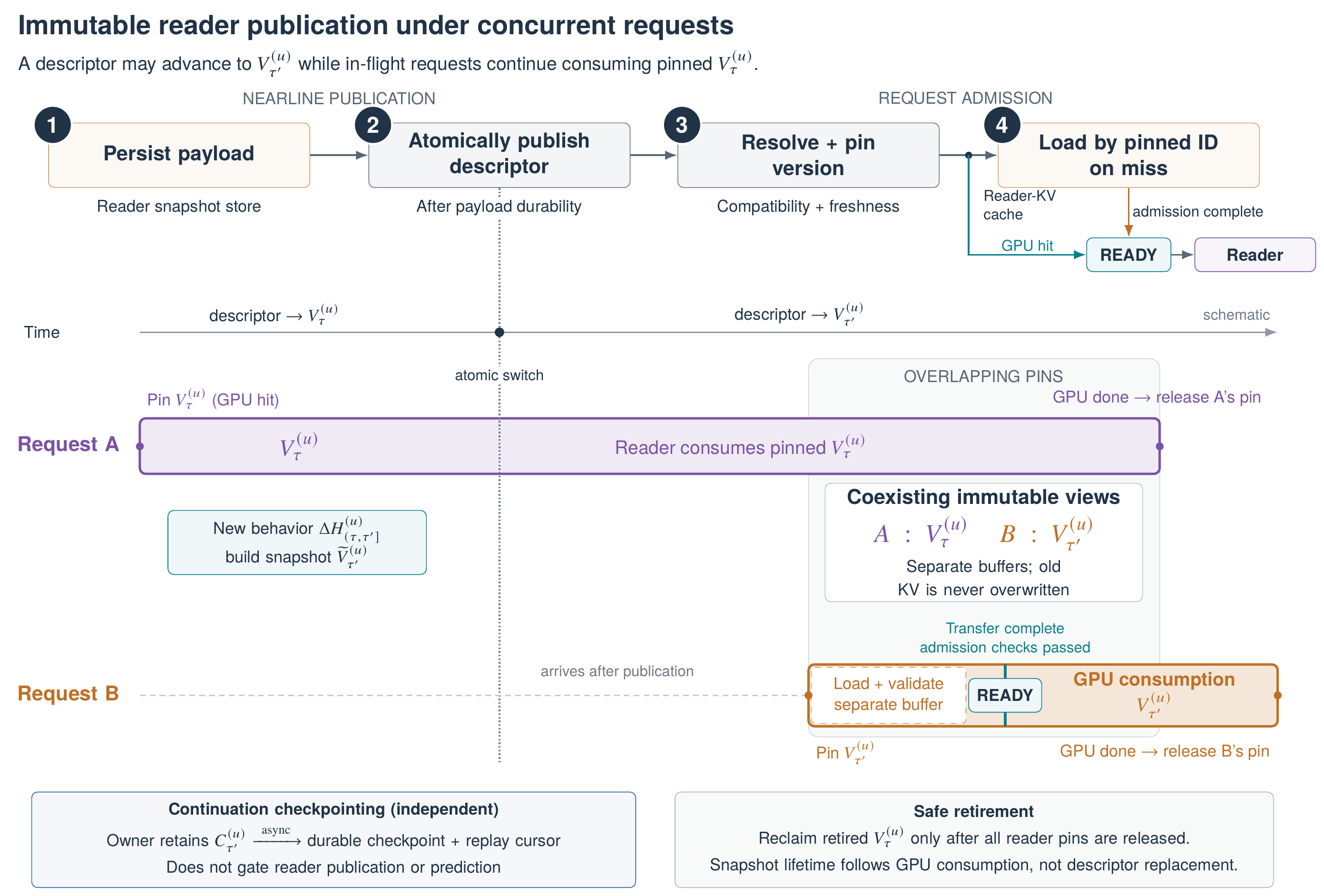}
    \caption{\textbf{Immutable behavior memory publication and request-scoped snapshot
    lifetimes.} Descriptor publication can advance to a new snapshot while
    in-flight requests continue consuming their pinned immutable view. Cache
    hits and miss-driven loads converge at a ready boundary before behavior memory
    execution, and retired snapshots are reclaimed only after all behavior memory KV pins
    are released.}
    \label{fig:reader-publication}
\end{figure}

\paragraph{Publication contract.}
For user $u$, let $C_{\tau}^{(u)}$ denote continuation state representing
behavior through cutoff $\tau$. After processing newly available
behavior, the nearline worker constructs a prospective behavior-memory snapshot
$\widetilde V_{\tau'}^{(u)}$ from the advanced state $C_{\tau'}^{(u)}$.
The snapshot becomes externally visible only after its payload has been
persisted and its producing metadata have been validated.

Each publication records the represented-history cutoff and source
cursor together with the representation dependencies required for
compatibility. At minimum these include
\[
(v_{\mathrm{model}},
 v_{\mathrm{memory}},
 v_{\mathrm{projection}},
 v_{\mathrm{codec}},
 v_{\mathrm{layout}}),
\]
as well as the continuation-state version and ownership generation used
to produce the snapshot. Publication validates the source cursor,
base-state version, and ownership generation before atomically advancing
the user-visible descriptor:
\[
\widetilde V_{\tau'}^{(u)}
\xrightarrow[\text{validate + persist}]
{\operatorname{Publish}}
V_{\tau'}^{(u)}.
\]
A failed or stale publication attempt does not replace the currently
visible compatible snapshot.

This ordering separates prediction-visible state from mutable state
maintenance. The continuation state may continue advancing, fail, or be
reconstructed without modifying a snapshot that has already been
published.

\paragraph{Request-scoped pinning and reclamation.}
At request admission, the online serving path resolves one compatible
behavior-memory snapshot descriptor and pins that snapshot for the lifetime of
the request. All candidate batches associated with the request (denoted by $r$) consume
the same pinned version:
\[
r
\longrightarrow
V_{\tau(r)}^{(u)}.
\]
If nearline publication advances the descriptor while the request is in
flight, the request continues using its previously pinned immutable
snapshot rather than switching versions mid-execution.

A superseded snapshot may therefore remain resident while requests still
hold references to it. Retired snapshots are reclaimed only after all
request-scoped pins have been released. Cache-hit and miss-driven
loading paths converge at the same ready-state boundary, so the
request-visible semantics do not depend on whether the snapshot was
already resident or fetched after admission.

\paragraph{Continuation checkpointing and replay recovery.}
Encoder continuation state is checkpointed asynchronously together with
its source cursor, represented-history cutoff, and compatible
representation versions. Let $c_0$ denote the source cursor recorded by
the latest compatible checkpoint and $c$ the current committed source
position. Recovery reconstructs current continuation state by restoring
the checkpoint and replaying the retained behavior-log suffix:
\[
C^{(u)}(c_0)
+
\Delta H^{(u)}(c_0,c]
\;\longrightarrow\;
C^{(u)}(c).
\]
Source retention is configured so that every acknowledged behavior
update is either represented in a durable checkpoint or remains
available for replay.

Recovery reconstructs mutable continuation state $C$; it does not
invalidate an already published compatible behavior memory snapshot $V$.
Consequently, online requests can continue to consume the latest
admissible published state while checkpoint restore and log replay
proceed asynchronously. Once continuation state catches up, ordinary
nearline advancement resumes and a fresher behavior memory snapshot can be
published.

\paragraph{Pipelined nearline execution.}
Nearline workers overlap continuation-state access, incremental
encoding, memory construction, snapshot persistence, and publication
across independent users. State required by a future update may be
fetched while another update executes, while completed behavior memory snapshots
may be persisted and published as subsequent encoding work proceeds.
This pipeline hides a portion of state-transfer and storage latency
without relaxing causal ordering of committed state transitions within a
user.

\subsection{Large Sparse Embedding Serving}
\label{app:sparse-embedding-serving}

The sparse embedding tables need not reside entirely in accelerator
memory. TransX exposes one versioned embedding namespace but uses
different placement and lookup paths for bulk prefill, nearline prefill,
and online prediction, as summarized in
Table~\ref{tab:embedding-serving}. Small tables may be replicated in HBM
within the serving-memory budget. Larger tables use sharded backing
storage, with GPU and host caches retaining selected rows. Inference
placement is chosen independently of the training sharding plan because
serving has no gradient traffic and substantially different locality and
latency requirements.

Each operation pins a compatible model and embedding version before
lookup; deduplication and cache reuse never cross incompatible versions,
and reused continuation or behavior memory state must be compatible with the
corresponding producing parameters. 

\begin{table}[hbt]
\centering
\small
\setlength{\tabcolsep}{4pt}
\renewcommand{\arraystretch}{1.12}
\begin{tabularx}{\linewidth}{
  @{}p{0.13\linewidth}
     p{0.24\linewidth}
     X
     p{0.22\linewidth}@{}
}
\toprule
\textbf{Path}
& \textbf{Primary placement}
& \textbf{Lookup execution}
& \textbf{Scheduling objective} \\
\midrule

Bulk prefill
&
Locally staged snapshot; GPU shards or host-memory backing
&
Extract upcoming IDs, deduplicate by $(\text{table},\text{ID})$,
batch reads by shard, prefetch a bounded future window while the current
window is encoded, and restore original occurrences and ordering
&
Maximize completed prefill work per unit cost
\\

Nearline prefill
&
Worker-local embedding cache near the continuation-state owner
&
Start lookup when a new behavior event arrives; route the fetched
embedding rows toward the continuation owner and combine compatible
updates into bounded microbatches
&
Maximize sustainable update throughput within the freshness budget
\\

Online prediction
&
Per-worker GPU hot cache, optionally backed by node-local host memory
and remote embedding shards
&
Resolve GPU hits immediately; deduplicate misses, group them by
table/shard, fetch them in batches, and admit candidate work only after
its required embedding rows are ready
&
Maximize request goodput within latency and memory limits
\\

\bottomrule
\end{tabularx}
\caption{Workload-specific lookup paths for large sparse embedding tables.
All three paths consume the same versioned embedding namespace, but use
different placement, batching, and admission policies.}
\label{tab:embedding-serving}
\end{table}

Candidate batches are dispatched to cross attention kernels only after all
required embedding rows are resolved. Pending lookups still consume
queue, buffer, and transfer resources, which are bounded separately.
A behavior-memory hit removes historical re-encoding from the online serving
path, but request- and candidate-side embedding lookups may remain.

A cache miss for a known row triggers retrieval, not substitution with
an unknown-ID vector. Legitimate unseen IDs follow the model's declared
OOV policy, while unavailable or incompatible rows follow an explicit
failure policy. In particular, prefill does not publish behavioral
state constructed from temporary cache-miss defaults.

\subsection{KV Quantization and Execution Details}
\label{app:kv-codec}

Section~\ref{sec:kv-quantization} summarizes the storage--quality
trade-off of compressed behavior memory KV. This appendix specifies the codec and
the two GPU execution representations in greater detail.

\paragraph{Quantization unit.}
Quantization is applied after behavior memory projection. For one key or
value vector
\[
x\in\mathbb R^{d_h},
\]
we separate magnitude and direction,
\[
a=\|x\|_2,
\qquad
s=
\begin{cases}
x/a, & a>0,\\
0,   & a=0.
\end{cases}
\]
Zero vectors use an explicit zero representation. The magnitude $a$ is
stored separately, while the normalized direction $s$ is compressed.

For each cross attention layer, KV group, and tensor type (K or V), calibration data from the training
distribution define the empirical second moment
\[
\widehat\Sigma
=
\mathbb E[ss^\top]
=
U\Lambda U^\top.
\]
We use the orthogonal rotation
\[
R=U^\top,
\qquad
z=Rs,
\label{eq:kv-spectral-rotation}
\]
so that quantization is performed in coordinates aligned with the
empirical energy spectrum. All $d_h$ coordinates are retained:
compression changes precision rather than truncating dimensions.

\paragraph{Precision allocation.}
Let $\lambda_i$ denote the calibrated energy of rotated coordinate $i$
and let $b_i$ denote its assigned number of bits. Under a direction
payload budget $B_{\text{bits}}$, we use the high-resolution scalar-quantization
distortion proxy~\citep{goyal2001transform}
\begin{equation}
D(b)
=
\sum_{i=1}^{d_h}
\bar\lambda_i\,2^{-2b_i},
\qquad
\bar\lambda_i=\max(\lambda_i,\lambda_{\text{min}}),
\label{eq:kv-distortion-proxy}
\end{equation}
subject to
\[
\sum_i b_i=B_{\text{bits}},
\qquad
b_{\min}\le b_i\le b_{\max}.
\]
The continuous allocation has the water-filling form
\[
b_i^\star
=
\operatorname{clip}_{[b_{\min},b_{\max}]}
\left[
\frac{1}{2}
\log_2
\left(
\frac{\bar\lambda_i}{\mu}
\right)
\right],
\qquad
\sum_i b_i^\star=B_{\text{bits}},
\label{eq:kv-waterfilling}
\]
where $\mu$ is chosen to satisfy the bit budget.

The continuous solution is converted to supported integer bit widths
without exceeding the payload budget. Scalar codebooks are then fitted
to the rotated coordinates using the calibration population. Vector
magnitudes, code indices, shape metadata, and codec identifiers are
accounted for separately from the direction-payload budget.

Given decoded rotated coordinates $\widehat z$, reconstruction is
\[
\widehat x
=
aR^\top\widehat z.
\label{eq:kv-reconstruction}
\]

\paragraph{Rotation-aware attention.}
The orthogonal change of basis itself introduces no approximation. For
key rotation $R_K$,
\[
q^\top k
=
(R_K q)^\top(R_K k),
\label{eq:key-rotation}
\]
so the key rotation may be absorbed into the corresponding query
projection when operation order permits.

Likewise, if values are represented in a rotated basis, the weighted
attention output can be mapped through the inverse value rotation before
the cross attention heads are combined. Fixed rotations can therefore be folded
into adjacent learned linear operators where compatible with the kernel
layout; otherwise they remain explicit execution operations.
Approximation arises from quantization, not from the orthogonal change
of basis.

\paragraph{Quality calibration.}
The distortion proxy in Eq.~\eqref{eq:kv-distortion-proxy} is used to
allocate precision but is not itself a guarantee on recommendation
quality. Candidate queries can amplify different components of the same
KV reconstruction error. We therefore select operating points using
held-out recommendation quality under representative cross attention queries,
rather than relying only on vector reconstruction error.

\paragraph{Packed KV execution.}
A published snapshot may remain in its compressed representation when
resident on GPU. For attention tile $\mathcal T$, the cross attention executes
schematically as
\[
\operatorname{Pack}(K_{\mathcal T},V_{\mathcal T})
\xrightarrow{\operatorname{Decode}}
(\widehat K_{\mathcal T},\widehat V_{\mathcal T})
\xrightarrow{\operatorname{Attend}}
O_{\mathcal T}.
\]
Decoding is fused as closely as practical with the attention operation
so that the complete expanded snapshot need not be materialized in HBM.
This mode minimizes persistent HBM footprint at the cost of repeated
decode work when the same KV is read multiple times.

\paragraph{Expanded KV execution.}
Alternatively, a compressed snapshot can be expanded when admitted to
GPU:
\[
\operatorname{Pack}(V_{\tau}^{(u)})
\xrightarrow{\operatorname{Decode}}
V^{(u),\,\mathrm{bf16}}_{\tau},
\]
after which candidate tiles and subsequent requests attend directly to
the retained high-precision representation:
\[
V^{(u),\,\mathrm{bf16}}_{\tau}
\xrightarrow{\operatorname{Attend}}
O.
\]
Expansion therefore pays decode and HBM-write cost once and amortizes
that cost over subsequent reads, at the expense of a larger resident
working set.

The two paths implement the same logical published KV snapshot; they
differ only in local execution representation. In the measured 4-bit-average-payload
configuration, packed residency reduces cross-attention-KV HBM by approximately
$3.1\times$ relative to expanded high-precision residency, with
approximately $6\%$ additional cross-attention-kernel latency. Expanded residency
becomes preferable after approximately $4$--$6$ reads per snapshot
under the measured workload.

\paragraph{Execution-format selection.}
The packed--expanded choice can be viewed as an amortization decision.
Let $C_{\mathrm{dec}}$ denote the incremental decode cost per packed
read, $C_{\mathrm{expand}}$ the one-time expansion cost, $B_{\mathrm{exp}}$
the additional HBM required by expanded residency, and $R$ the expected
number of reads before refresh or eviction. Expansion is favored when
the avoided repeated decode work is sufficient to justify the
one-time decode and additional residency:
\begin{equation}
N_{\text{read}}\,C_{\mathrm{dec}}
\gtrsim
C_{\mathrm{expand}}
+
\lambda_{\mathrm{HBM}} B_{\mathrm{exp}},
\label{eq:kv-expansion-condition}
\end{equation}
where $\lambda_{\mathrm{HBM}}$ denotes the effective opportunity cost of
GPU memory, and $N_{\text{read}}$ is the expected reuse count of the same behavior-memory KV snapshot before it is refreshed or evicted. Therefore, Eq.~\eqref{eq:kv-expansion-condition} is an operational
interpretation rather than the exact admission policy; the production
system uses measured access frequency, memory pressure, and kernel
profiles to choose the execution representation.

\paragraph{KV byte accounting.}
For $L_R$ cross-attention layers, $N$ memory slots, $G$ KV heads, head dimension
$d_h$, and $b$ bytes per scalar, the uncompressed projected-KV footprint is
\[
B_{\rm KV}
=
L_R\,N\,2Gd_h\,b,
\]
where the factor $2$ accounts for keys and values. In our deployed cross attention configuration, with \(H_q=32\) query heads and 8-way query-to-KV sharing (\(G=4\)), \(d_h=32\), \(L_R=2\), and median \(N=28\), BF16 therefore requires
\[
2\times N\times2\times G\times d_h\times2
= 28~\mathrm{KiB/user}.
\]
The reported $8.5$~KiB compressed footprint includes the quantized
coordinate payload together with per-vector magnitudes and per-snapshot
codec/layout metadata.

\subsection{Behavior-Memory KV Snapshot Storage and Version Rollout}
\label{app:kv-cache-rollout}

\paragraph{Shared-store layout.}
The backing store holds compressed projected KV as a derived payload of
the versioned snapshot $\mathsf V_{\tau}^{(u)}$. Each record contains packed
code indices, vector magnitudes, tensor shapes, and codec/layout
identifiers. Rotation matrices, bit allocations, and codebooks are shared
across users through immutable codec identifiers. Retaining the underlying
memory $M_{\tau}^{(u)}$ adds storage but permits reprojection or
requantization without re-encoding history when that memory remains
compatible. Encoder continuation checkpoints are managed separately.

For small snapshots, per-object I/O overhead can dominate payload transfer.
The backing implementation therefore packs records into indexed, immutable
storage segments and coalesces reads within request deadlines. Updates
append new records; online-prediction-visible payloads are never modified in place.

\paragraph{Snapshot refresh.}
Fresh KV is encoded on the prefill path from compatible materialized memory,
not repeatedly transcoded from older lossy payloads. Expansion changes
execution layout rather than quantization fidelity: higher-fidelity KV
must be regenerated from retained memory or another higher-fidelity source.

\paragraph{Model rollout and cache version transition.}
A full retrain can change the behavior memory representation, making the new model
incompatible with the existing population-wide KV cache. Because rebuilding
snapshots for all users can take substantially longer than loading new model
weights, rollout proceeds as a dual-version transition rather than an
in-place replacement. For notational convenience, we use the user index as a subscript in this section, rather than the superscript notation used above.

Let version $v$ denote the currently serving model/cache pair and $v+1$
the new rollout. Bulk prefill constructs the new cache namespace in the
background while requests continue to use
\[
\left(
\theta^{(v)},
\mathsf V_u^{(v)}
\right).
\]
For each user, bulk prefill records the source cursor $c_u^{\mathrm{bulk}}$
through which behavior has been incorporated into the new continuation
state. Behavior arriving after this cursor remains available in the source
log and is either consumed concurrently by the $v+1$ nearline path or
replayed after bulk construction completes:
\[
\mathsf C_u^{(v+1)}(c_u^{\mathrm{bulk}})
+
\Delta H_u
\left(
c_u^{\mathrm{bulk}},c_u^{\mathrm{current}}
\right]
\longrightarrow
\mathsf C_u^{(v+1)}(c_u^{\mathrm{current}})
\longrightarrow
\mathsf V_u^{(v+1)}.
\]
Thus, no behavior acknowledged after the bulk-prefill cutoff is lost between
backfill and cutover.

Traffic migrates to version $v+1$ only after the user's compatible snapshot
is available and its represented cutoff satisfies the rollout freshness
policy:
\[
\left(
\theta^{(v)},
\mathsf V_{u}^{(v)}
\right)
\quad\longrightarrow\quad
\left(
\theta^{(v+1)},
\mathsf V_{u}^{(v+1)}
\right).
\]
Users not yet ready remain on version $v$, allowing serving to continue
throughout cache reconstruction. During the rollback window, the old
version either continues receiving updates or retains sufficient replay
coverage to catch up before traffic returns. It is reclaimed only after
traffic has drained and the rollback window has expired. This temporarily
duplicates model/cache state and requires additional maintenance or
catch-up work, but preserves serving continuity and an explicit rollback
path.

\begin{table}[hbt]
\centering
\small
\setlength{\tabcolsep}{5pt}
\renewcommand{\arraystretch}{1.15}

\begin{tabularx}{\linewidth}{
    @{}
    >{\raggedright\arraybackslash}p{0.17\linewidth}
    *{3}{>{\raggedright\arraybackslash}X}
    @{}
}
\toprule
\textbf{Dimension}
& \textbf{Offline bulk prefill}
& \textbf{Nearline incremental maintenance}
& \textbf{Online prediction} \\
\midrule

\textbf{Workload trigger}
& Model deployment, backfill, or representation migration
& New user-behavior events
& Recommendation requests \\ \hline
\addlinespace[5pt]

\textbf{Execution objective}
& Maximize history-processing throughput per unit cost while
  meeting the rollout deadline
& Maximize sustainable state-update throughput within the
  freshness SLO
& Maximize request goodput while satisfying the request-level
  latency SLO \\ \hline
\addlinespace[5pt]

\textbf{Primary constraints}
& Long-history compute, accelerator utilization, and backfill volume
& Continuation-state locality, update queueing, and
  state-transfer bandwidth
& Candidate fanout, queueing, behavior memory locality, and
  memory pressure \\ \hline
\addlinespace[5pt]

\textbf{Capacity-control signals}
& Backfill backlog and required rollout-completion time
& Update-queue depth and p90 freshness lag
& Offered traffic, p99 latency, request queue depth, and
  host/GPU memory pressure \\

\bottomrule
\end{tabularx}

\caption{
Stage-specific workload objectives and capacity-control signals
under disaggregated TransX serving. Each stage can scale
independently according to its own bottleneck rather than
requiring uniform scaling of the complete serving stack.
}
\label{tab:stage-specific-capacity-scaling}
\end{table}

\paragraph{Stored KV representation.}
The backing store contains compressed projected KV rather than
high-precision model-side memory. Each immutable record contains the
packed KV payload together with the metadata required to interpret it,
including tensor shapes, vector magnitudes where required by the codec,
and codec/layout identifiers. Rotation matrices, bit allocations, and
codebooks are shared across users through immutable codec versions
rather than duplicated in every snapshot.

\begin{table}[htb]
\centering
\small
\begin{tabular}{llll}
\toprule
\textbf{Residency} &
\textbf{Format} &
\textbf{Preferred regime} &
\textbf{Primary trade-off} \\
\midrule

Backing tier &
Packed &
Cold population &
Minimum storage; fetch required \\

Host cache &
Packed &
Warm state &
Avoid backing reads; device transfer remains \\

GPU &
Packed &
Low reuse or HBM pressure &
Small footprint; repeated decode \\

GPU &
Expanded &
Repeated reads &
Larger HBM; decode amortized \\
\bottomrule
\end{tabular}
\caption{Residency and execution representations for published behavior memory KV.
Placement and execution format are separate decisions: the same logical
snapshot may move across tiers or change its local execution
representation without changing its request-level semantics.}
\label{tab:kv-residency-execution}
\end{table}

For small snapshots, per-object I/O overhead can dominate the payload
itself. KV records are therefore grouped into indexed immutable storage
segments, allowing compatible reads to be coalesced. New snapshots are
appended rather than modifying previously published behavior memory KV in
place.

\paragraph{Logical version versus physical representation.}
We distinguish the semantic compatibility of a behavior memory KV snapshot from
its current residency or execution format. At minimum, compatibility
depends on
\[
\left(
v_{\mathrm{model}},
v_{\mathrm{memory}},
v_{\mathrm{projection}},
v_{\mathrm{codec}},
v_{\mathrm{layout}}
\right).
\label{eq:kv-version-tuple}
\]
A request may consume KV only when the producing dependencies are
compatible with the pinned model.

By contrast, movement among backing storage, host cache, packed GPU
residency, and expanded GPU residency does not by itself create a new
semantic snapshot version. These transitions modify placement or local
execution representation, not the behavioral state represented by the
snapshot.

Compatibility is dependency based rather than requiring blanket
invalidation. For example, a prediction-head-only change need not
invalidate existing behavior memory KV if none of the parameters or transforms
that produced that KV have changed.

\paragraph{Regeneration and transcoding.}
Fresh behavior memory KV is generated from compatible model-side memory when that
memory is retained; it is not repeatedly transcoded from an older lossy
KV payload. Changing only execution layout may transform the physical
representation of an existing snapshot, but increasing representation
fidelity requires regeneration from retained memory or another
higher-fidelity source.

\subsection{System Robustness and Horizontal Scaling Analysis}
\label{app:system-robustness}

\begin{table*}[htb]
\centering
\footnotesize
\setlength{\tabcolsep}{4pt}
\renewcommand{\arraystretch}{1.10}
\begin{tabularx}{\textwidth}{
@{}p{0.15\textwidth}
   p{0.15\textwidth}
   p{0.15\textwidth}
   p{0.15\textwidth}
   X@{}}
\toprule
\textbf{Failure / stress condition}
&
\textbf{Injected condition}
&
\textbf{Serving impact}
&
\textbf{Recovery measurement}
&
\textbf{Performance}
\\
\midrule

Continuation-owner loss
&
Terminate owners serving $1\%$ of active users
&
Latest-snapshot share, stale-view / fallback rate, p90 freshness lag
&
Checkpoint restore time, replayed events/user, time to freshness recovery
&
Latest-compatible share falls from $97.8\%$ to $93.6\%$;
$4.9\%$ stale-compatible, $1.5\%$ fallback, $<0.01\%$ failed.
Median restore $1.7$\,s; p90 freshness peaks at $6.2$\,s and returns below
the $2$\,s SLO within $21$\,s.
\\ \hline

KV cache eviction spike
&
Evict $50\%$ of GPU-resident behavior memory snapshots
&
GPU and mixed-residency hit rates, backing-store reads, p99 scoring latency
&
Time to restore steady-state residency
&
GPU hit rate drops from $25\%$ to $12\%$ and mixed-residency hit rate from
$99.3\%$ to $96.8\%$; p99 scoring latency increases from $49$ to
$58$\,ms. Mixed-residency hit rate recovers above $99\%$ within $47$\,s.
\\ \hline

Snapshot-store unavailability
&
Block backing-store reads for $60$\,s
&
Resident / stale-compatible / fallback shares and request failure rate
&
Time to normal snapshot availability
&
$96.9\%$ of requests continue from resident compatible snapshots,
$2.2\%$ use an older admissible snapshot, and $0.9\%$ use the
non-sequential fallback; request failures remain below $0.01\%$.
Normal fetch availability recovers within $8$\,s after store restoration.
\\ \hline

Publication conflicts
&
Inject stale or duplicate publication attempts at a $1\%$ rate
&
Retry rate, stale-view fraction, publication delay
&
Retry success and additional publication latency
&
$99.4\%$ of conflicts resolve on the first retry and $>99.99\%$ within
two retries; p99 additional publication delay is $38$\,ms, with no
measurable change in online request failure rate.
\\ \hline

Nearline publication stall
&
Delay publication by $10$\,s for $5\%$ of active users
&
Latest-snapshot share, stale-compatible share, fallback rate
&
Time for latest-compatible coverage to recover
&
During the stall, $92.7\%$ of affected requests use the latest compatible
snapshot, $5.8\%$ use an older admissible snapshot, and $1.5\%$ fall back.
Latest-compatible coverage returns above $99\%$ within $24$\,s after
publication resumes.
\\

\bottomrule
\end{tabularx}
\caption{
\textbf{System robustness measurements} for controlled
failure injection in TransX serving.}
\label{tab:transx-robustness}
\end{table*}

\paragraph{Horizontal scaling observations.}
Across all three stages, useful goodput scales close to linearly at moderate
replication and then exhibits expected diminishing efficiency as
stage-specific bottlenecks become dominant. Bulk prefill reaches approximately
$96\%$ scaling efficiency at $2\times$ capacity and $90\%$ at $4\times$,
with the remaining loss primarily associated with long-history stragglers and
batch imbalance. Nearline maintenance shows a similar trend, but its scaling
degrades somewhat faster as hot-user ownership and continuation-state movement
begin to dominate over dense compute. Online prediction scales most efficiently
at $2\times$ replication and remains above $90\%$ efficiency at $4\times$,
while the residual loss is driven by candidate-fanout skew, cache locality, and
smaller effective batches. Importantly, raw GPU utilization remains high even
when scaling efficiency declines, confirming that device occupancy alone is not
a sufficient measure of system efficiency; stage-specific useful goodput and
SLO compliance provide the more relevant scaling signals.
\section{Scaling Analysis}
\label{app:scaling}

We provide the proof for the log-loss decomposition and the scaling experiment manifests in this part of the appendix.

\subsection{Log-Loss Decomposition for Compressed Behavioral State}
\label{app:scaling-decomposition}

\paragraph{Setting and assumptions.}
Let \(Y\in\mathcal Y\), where \(\mathcal Y\) is finite (e.g. \(Y\in\{0,1\}\) for binary response target). Fix a held-out population \(P(H,X,Y)\) for which the relevant conditional distributions exist and the population log loss is finite.
Let
\begin{equation}
\label{eq:app-memory-readout}
M=f_{\phi}(H),
\qquad
q_{\psi}(y\mid M,X)
\end{equation}
denote a deterministic behavioral memory and a probabilistic
fusion-plus-head readout, including any learned processing of $X$.
Write
\begin{equation}
\label{eq:app-true-conditionals}
\begin{aligned}
p_H(y\mid H,X)&=P(Y=y\mid H,X),\\
p_M(y\mid M,X)&=P(Y=y\mid M,X),
\end{aligned}
\end{equation}
and define population log loss by
\begin{equation}
\label{eq:app-population-logloss}
\mathcal{L}(\phi,\psi)
=
\mathbb{E}_{P}
\left[-\log q_{\psi}(Y\mid f_{\phi}(H),X)\right].
\end{equation}
Note that the conditional distribution $p_M$ is induced by $P$ and the encoder rather than a fitted readout.

\begin{assumption}[Fixed evaluation population, fitted model, and finite risk]
\label{ass:scaling-regularity}
The distribution $P$ and the observed inputs $(H,X)$ are fixed for each
application of the identity. The parameters $(\phi,\psi)$ are held
fixed when evaluating population risk; they may have been learned jointly or sequentially.
Equivalently, $(\phi,\psi)$ may be viewed as conditioned on the training data and training randomness, with evaluation then performed on a fresh draw from \(P\). The readout assigns
positive probability wherever $p_M(y\mid M,X)>0$, almost surely,
and $\mathcal{L}(\phi,\psi)<\infty$.
\end{assumption}

The theorem below adopts the above setting and assumption. We use $\mathcal{H}$ to denote the Shannon conditional entropy operator, distinguishing it from the behavioral-history variable \(H\).

\begin{theorem}[Predictive-information and readout decomposition]
\label{thm:scaling-decomposition}
Population log loss decomposes as
\begin{equation}
\label{eq:app-three-term-decomposition}
\begin{aligned}
\mathcal{L}(\phi,\psi)
&=
\mathcal{H}(Y\mid H,X)
+
\mathcal{I}_{\mathrm{lost}}(\phi)
+
\mathcal{R}(\phi,\psi),
\end{aligned}
\end{equation}
where
\begin{align}
\label{eq:app-information-loss}
\mathcal{I}_{\mathrm{lost}}(\phi)
&=I(Y;H\mid M,X),\\
\label{eq:app-readout-mismatch}
\mathcal{R}(\phi,\psi)
&=
\mathbb{E}_{M,X}
D_{\mathrm{KL}}\!\left(
p_M(\cdot\mid M,X)
\,\Vert\,
q_{\psi}(\cdot\mid M,X)
\right).
\end{align}
Both terms are nonnegative. Moreover,
\begin{equation}
\label{eq:app-information-loss-kl}
\mathcal{I}_{\mathrm{lost}}(\phi)
=
\mathbb{E}_{H,X}
D_{\mathrm{KL}}\!\left(
p_H(\cdot\mid H,X)
\,\Vert\,
p_M(\cdot\mid f_{\phi}(H),X)
\right).
\end{equation}
\end{theorem}

\begin{proof}
Conditioning on \((M,X)\) and adding and subtracting the log probability under
$p_M$ gives,
\begin{align}
\mathcal{L}(\phi,\psi)
&=
\mathbb{E}_{M,X}
\sum_{y\in\mathcal{Y}}
p_M(y\mid M,X)
\bigl[-\log q_{\psi}(y\mid M,X)\bigr]
\nonumber\\
&=
\mathcal{H}(Y\mid M,X)
\nonumber\\
&\quad+
\mathbb{E}_{M,X}
\sum_{y\in\mathcal{Y}}
p_M(y\mid M,X)
\log\frac{p_M(y\mid M,X)}{q_{\psi}(y\mid M,X)}
\nonumber\\
&=
\mathcal{H}(Y\mid M,X)
+
\mathcal{R}(\phi,\psi).
\label{eq:app-cross-entropy-identity}
\end{align}
Because $M=f_{\phi}(H)$ is deterministic, it holds that
\begin{equation}
\label{eq:app-deterministic-memory-entropy}
\mathcal{H}(Y\mid H,M,X)
=
\mathcal{H}(Y\mid H,X).
\end{equation}
Consequently, the conditional mutual-information identity gives
\begin{align}
I(Y;H\mid M,X)
&=
\mathcal{H}(Y\mid M,X)
-
\mathcal{H}(Y\mid H,M,X)
\nonumber\\
&=
\mathcal{H}(Y\mid M,X)
-
\mathcal{H}(Y\mid H,X).
\label{eq:app-entropy-information-identity}
\end{align}
Equivalently,
\begin{equation}
\label{eq:app-memory-entropy-split}
\mathcal{H}(Y\mid M,X)
=
\mathcal{H}(Y\mid H,X)
+
I(Y;H\mid M,X).
\end{equation}
Substitution into Eq.~\eqref{eq:app-cross-entropy-identity} yields
Eq.~\eqref{eq:app-three-term-decomposition}.

To obtain the conditional-KL representation, expand its definition:
\begin{align}
&\mathbb{E}_{H,X}
D_{\mathrm{KL}}\!\left(
p_H(\cdot\mid H,X)
\,\Vert\,
p_M(\cdot\mid f_{\phi}(H),X)
\right)
\nonumber\\
&\qquad=
\mathbb{E}_{H,X,Y}
\left[
\log\frac{p_H(Y\mid H,X)}{p_M(Y\mid M,X)}
\right]
\nonumber\\
&\qquad=
\mathcal{H}(Y\mid M,X)
-
\mathcal{H}(Y\mid H,X)
\nonumber\\
&\qquad=
\mathcal{I}_{\mathrm{lost}}(\phi).
\end{align}
Since conditional mutual information and KL divergence are nonnegative,
$I_{\mathrm{lost}}(\phi)\ge 0$ and $R(\phi,\psi)\ge 0$, completing the proof.
\end{proof}

\subsection{Scaling Experiment Manifest}
\label{app:scaling-manifest}

This section reports the quantities needed to reconstruct the scaling
comparisons in Section~\ref{sec:scaling}. We use
$N_{\rm ref}=1.0$B valid supervised candidate targets and supervision
fractions $\eta\in\{0.20,0.25,0.30\}$. The measured numbers of valid
targets per encoded trajectory are approximately
$\nu=\{64,255,2.055\mathrm{K}\}$ for the sparse-, medium-, and
rich-history cohorts, respectively. The resulting data volumes are:

\begin{table}[h]
\centering
\small
\begin{tabular}{lccc}
\toprule
 & Sparse & Medium & Rich \\
\midrule
History length & 64--252 & 253--996 & 997--4096 \\
$\nu$ (targets/trajectory) & 64 & 255 & 2.055K \\
Trajectories, $\eta=0.20$ & 3.12M & 0.78M & 97K \\
Trajectories, $\eta=0.25$ & 3.91M & 0.98M & 122K \\
Trajectories, $\eta=0.30$ & 4.69M & 1.18M & 146K \\
\bottomrule
\end{tabular}
\caption{Data manifest for the scaling experiments. Each supervision
level contains $200$M, $250$M, or $300$M valid targets, respectively.}
\label{tab:scaling-data-manifest}
\end{table}

The encoder sweep uses depths $\{1,2,3,4,6,8\}$. Raw encoder compute
$C_E$ (MFLOPs/trajectory) is

\[
\begin{array}{c|rrrrrr}
\text{depth} & 1 & 2 & 3 & 4 & 6 & 8\\
\hline
\text{sparse} & 60 & 119 & 179 & 238 & 358 & 477\\
\text{medium} & 300 & 600 & 899 & 1199 & 1799 & 2398\\
\text{rich} & 2294 & 4587 & 6881 & 9174 & 13761 & 18348
\end{array}
\]

Using the measured numbers of valid supervised targets per encoded trajectory,
\[
\nu_{\rm sparse}\approx64,\qquad
\nu_{\rm medium}\approx255,\qquad
\nu_{\rm rich}\approx2.055\mathrm{K},
\]
the corresponding amortized encoder-compute grids are

\[
\bar C_E^{\rm sparse}
\approx
\{0.94,1.86,2.80,3.72,5.59,7.45\},
\]

\[
\bar C_E^{\rm medium}
\approx
\{1.18,2.35,3.53,4.70,7.05,9.40\},
\]

\[
\bar C_E^{\rm rich}
\approx
\{1.12,2.23,3.35,4.46,6.70,8.93\}
\quad\text{MFLOPs/target}.
\]

Fusion width is swept from 128 to 1024, corresponding to measured
candidate-conditioned compute

\[
C_F
\approx
\{0.40,0.70,1.22,2.61,4.50,7.00,9.80,17.20\}
\quad\text{MFLOPs/candidate}.
\]
while the depth sweep uses $1$--$5$ reader blocks with
$C_F\approx\{0.4,0.7,1.0,1.3,1.6\}$ MFLOPs/candidate.
Memory sweeps evaluate progressively finer local--global memories,
reported both by logical resolution and by physical BF16 snapshot
budgets spanning approximately $8$--$256$ KiB/user. All unvaried
components are held fixed within each sweep.

All configurations within a $(h,\eta)$ regime use the same sampled
training set and optimization recipe. They are trained for one pass over
the sampled valid targets with a logical batch of $8192$ targets,
corresponding to approximately $24.4$K, $30.5$K, and $36.6$K logical
updates for $\eta=0.20$, $0.25$, and $0.30$, respectively. The primary
quality metric in Section~\ref{sec:scaling-analysis} is held-out click AUC. Because each supervision regime is trained for one pass over its sampled
targets, increasing $\eta$ increases both the number of distinct supervised
targets and the number of optimizer updates. The resulting trends therefore
characterize the one-pass training operating point and should not be
interpreted as isolating supervision diversity from optimization duration.

Repeated selected configurations are reported over three independent runs; the full allocation grid is evaluated once per configuration.
For repeated configurations, uncertainty is reported as the standard deviation across three independent training seeds. Repeating the memory-budget selection with
$\epsilon\in\{2,3,5\}$ AUC basis points changes the selected memory
configuration by at most one adjacent evaluated level and does not alter
the qualitative allocation trend toward fusion with increasing
supervision. Accordingly, the primary analysis uses $\epsilon=3$ basis points.
\section{Post-Training Details}
\label{app:posttraining}

Section~\ref{sec:post-training} emphasizes the readout geometries used
to specialize the shared TransX behavioral representation. This appendix
specifies the corresponding objective construction, sampling rules, and
task-specific execution details.

\subsection{Retrieval Objective and Negative Construction}
\label{app:retrieval-posttraining}

For a positive candidate $j^+$ exposed in event $(u,t)$, valid mature
non-positive candidates from the same event form the hard-negative set
\[
\mathcal N_{\mathrm{hard}}^{(u,t)}
=
\left\{
j:
j\neq j^+,\;
j\in\mathcal V_t^{(u)},\;
y_{t,j}^{(u)}=0
\right\}.
\label{eq:retrieval-hard-negatives}
\]
Because these candidates share the same behavioral memory and serving
context as the positive item, they provide naturally difficult
comparisons.

Candidate embeddings associated with other user--event examples in the
same logical batch provide an additional in-batch pool,
\[
\mathcal N_{\mathrm{batch}}^{(u,t)}
=
\left\{
j':
j' \text{ belongs to another user--event in the batch}
\right\},
\]
after removing known positives and duplicate items. For notational
simplicity, with one positive item per event,
\[
\mathcal C_t^{(u)}
=
\{j^+\}
\cup
\mathcal N_{\mathrm{hard}}^{(u,t)}
\cup
\mathcal N_{\mathrm{batch}}^{(u,t)}.
\]

Using the retrieval score $\exp s_{t,i}^{(u)}$ defined for item $i$ at event $(u,t)$, the InfoNCE loss is
\[
\mathcal L_{\mathrm{ret}}^{(u,t,j^+)}
=
-
\log
\frac{
\exp s_{t,j^+}^{(u)}
}{
\sum_{j\in\mathcal C_t^{(u)}}
\exp s_{t,j}^{(u)}
}.
\label{eq:retrieval-infonce}
\]

The event packing used by core training therefore supplies both the
contextual behavioral representation and within-event hard negatives,
while item representations already available elsewhere in the logical
batch enlarge the contrastive pool. Behavioral encoding and memory
materialization are unchanged; the task-specific reader, retrieval
projection, and item-side representation define the new capability.

The retrieval reader and item index must use compatible representation
versions. Changes to retrieval-specific reader projections therefore
induce a new retrieval KV/index compatibility version even when the
underlying behavioral memory remains unchanged.

\subsection{Preference-Ranking Objective and Distributed Reduction}
\label{app:ranking-posttraining}

For event $(u,t)$, let $\mathcal V_t^{(u)}$ contain candidates with
valid, mature feedback. A declared partial ordering $\succ_y$ over
response outcomes induces
\[
\mathcal P_t^{(u)}
=
\left\{
(i,j):
i,j\in\mathcal V_t^{(u)},
\;
y_{t,i}^{(u)}
\succ_y
y_{t,j}^{(u)}
\right\}.
\label{eq:ranking-preference-set}
\]
Tied or incomparable outcomes do not generate a pair.

For each $(i,j)\in\mathcal P_t^{(u)}$, the Bradley--Terry probability is
\[
P_\theta
\left(
i\succ_t j
\mid
M_t^{(u)}
\right)
=
\sigma
\left(
\frac{
a_{t,i}^{(u)}-a_{t,j}^{(u)}
}{
\tau_{\mathrm{rank}}
}
\right).
\]
Rather than enumerate all pairs, we sample a bounded subset
\[
\widetilde{\mathcal P}_t^{(u)}
\subseteq
\mathcal P_t^{(u)}.
\]
With optional positive pair weights $w_{ij}$, the event-normalized loss
is
\[
\mathcal L_{\mathrm{BT}}^{(u,t)}
=
-
\frac{
\sum_{(i,j)\in\widetilde{\mathcal P}_t^{(u)}}
w_{ij}
\log
\sigma
\left(
\frac{
a_{t,i}^{(u)}-a_{t,j}^{(u)}
}{
\tau_{\mathrm{rank}}
}
\right)
}{
\sum_{(i,j)\in\widetilde{\mathcal P}_t^{(u)}} w_{ij}
}.
\label{eq:weighted-bt}
\]
Event normalization prevents events with more candidates or eligible
pairs from receiving disproportionate weight. Events containing no
comparable pair contribute no Bradley--Terry term.

\paragraph{Distributed pair reduction.}
Preference ranking preserves the same encoded trajectory, event memory,
and candidate fanout as core training. When both candidates of a sampled
pair reside on the same fanout worker, the pairwise loss can be evaluated
locally. Otherwise, the corresponding logits are gathered to the event
owner or another event-local reduction point, where the pairwise loss is
formed and the resulting logit gradients are returned to candidate
owners.

Thus the additional communication introduced by preference
specialization is event local. Candidate gradients still aggregate
through the same projected KV and memory state before re-entering the
shared behavioral trajectory.

\subsection{Generative Preference Optimization}
\label{app:generative-preference}

Let $\xi_t$ denote the task-specific conditioning state derived from the
behavioral memory and serving context, and let $(G^+,G^-)$ denote a
preferred/rejected output pair. The trainable policy $\pi_\theta$ is
initialized from a generative checkpoint, while $\pi_{\mathrm{ref}}$
is the corresponding frozen pre-preference policy.

Define
\[
\Delta_\theta(\xi,G^+,G^-)
=
\log
\frac{
\pi_\theta(G^+\mid \xi)
}{
\pi_{\mathrm{ref}}(G^+\mid \xi)
}
-
\log
\frac{
\pi_\theta(G^-\mid \xi)
}{
\pi_{\mathrm{ref}}(G^-\mid \xi)
}.
\label{eq:dpo-delta}
\]
For $P_t$ valid pairs associated with event $t$, the DPO objective~\citep{rafailov2023dpo} is
\begin{equation}
\mathcal L_{\mathrm{DPO}}^{(u,t)}
=
-
\frac{1}{P_t}
\sum_{r=1}^{P_t}
\log
\sigma
\left(
\beta
\Delta_\theta
\left(
\xi_t^{(u)},
G_{t,r}^{+},
G_{t,r}^{-}
\right)
\right).
\label{eq:generative-dpo}
\end{equation}
Preference optimization may be combined with autoregressive likelihood
training,
\[
\mathcal L_{\mathrm{gen}}
=
\mathcal L_{\mathrm{AR}}
+
\lambda_{\mathrm{DPO}}
\mathcal L_{\mathrm{DPO}},
\label{eq:combined-generative-loss}
\]
to retain valid generation while improving preference separation.

All chosen and rejected outputs associated with one event condition on
the same cutoff-valid behavioral state, so preference multiplicity
increases readout work rather than behavioral-history encoding.

\subsection{Generative Retrieval Details}
\label{app:genret-posttraining}

For generative retrieval, each catalog item $j$ is represented by a
fixed Semantic-ID sequence
\[
G_j
=
\operatorname{SID}(j)
=
(g_{j,1},\ldots,g_{j,K}).
\]
The decoder factorizes item probability autoregressively:
\[
\log
\pi_\theta
\left(
G_j
\mid
M_t^{(u)},\xi_t^{(u)}
\right)
=
\sum_{k=1}^{K}
\log
\pi_\theta
\left(
g_{j,k}
\mid
M_t^{(u)},\xi_t^{(u)},g_{j,<k}
\right).
\label{eq:genret-likelihood}
\]

Each reader layer projects the behavioral memory once,
\[
K_{t}^{(\ell)}
=
M_t^{(u)}W_{K,\ell}^{\mathrm{GR}},
\qquad
V_{t}^{(\ell)}
=
M_t^{(u)}W_{V,\ell}^{\mathrm{GR}},
\]
and the resulting cross-attention KV is shared across decoding steps and
hypotheses.

An event-level item preference $i\succ_t j$ induces the sequence-level
pair
\[
G^+
=
\operatorname{SID}(i),
\qquad
G^-
=
\operatorname{SID}(j).
\]
For preferred item $i$, negative candidates are drawn from
\[
\mathcal N_t(i)
=
\mathcal N_{t}^{\mathrm{exposed}}(i)
\cup
\mathcal N_{t}^{\mathrm{hard}}(i),
\]
where the first term contains lower-preference items exposed in the same
event and the second contains catalog-valid hard negatives selected by
the retrieval preference scorer. Synthetic negatives are retained only
when the teacher establishes the configured preference margin.

Malformed or nonexistent Semantic IDs are handled by constrained
decoding rather than being treated as ordinary preference negatives.
With one negative per positive the pair objective follows
Eq.~\eqref{eq:generative-dpo}; multiple negatives may instead be handled
through a multi-negative or listwise objective.

\subsection{Generative Slate-Ranking Details}
\label{app:genrank-posttraining}

For event $t$, candidate-conditioned TransX representations
$\{z_{t,i}\}$ and base relevance logits $\{b_{t,i}\}$ are computed once
for the available candidate set. Let
\[
S=(j_1,\ldots,j_m)
\]
denote an ordered slate. At decoding step $k$, the slate decoder
summarizes the previously selected items into state $h_{t,k}$ and scores
each remaining candidate by
\[
s_{t,k,i}
=
b_{t,i}
+
\frac{
(W_q h_{t,k})^\top(W_c z_{t,i})
}{
\sqrt{d_{\text{slate}}}
}.
\label{eq:genrank-appendix-score}
\]
Previously selected candidates are masked from later positions. The
slate policy therefore factorizes as
\[
\pi_\theta(S\mid \xi_t)
=
\prod_{k=1}^{m}
\pi_\theta
\left(
j_k
\mid
\xi_t,j_{<k}
\right),
\label{eq:slate-factorization}
\]
where $\xi_t$ includes the fixed TransX candidate representations and
serving context.

\paragraph{Preference-pair construction.}
Logged serving events generally provide one realized slate rather than
counterfactual preferred/rejected slate pairs. We therefore sample
candidate slates
\[
S^{(1)},\ldots,S^{(M)}
\sim
\pi_{\mathrm{sample}}(\cdot\mid \xi_t)
\]
and score them using a frozen slate-aware preference function
$R_\omega(\xi_t,S)$. The preference teacher $R_\omega$ is trained on a separate mature-feedback
split and frozen before slate-pair construction. Teacher-preference accuracy
therefore measures agreement with $R_\omega$, not independent recommendation
quality, which is reported separately on held-out evaluation data. 
Matched-length pairs are
retained only when
\[
R_\omega(\xi_t,S^+)
-
R_\omega(\xi_t,S^-)
\ge
\delta_{\mathrm{pref}}.
\label{eq:slate-preference-margin}
\]
Controlled ordering or selection corruptions of a high-quality slate may
supply additional hard negatives when they remain valid and satisfy the
same preference-margin criterion.

The retained slate pairs are optimized using the preference objective in
Eq.~\eqref{eq:generative-dpo}. Candidate representations and behavioral
memory remain fixed across the chosen/rejected likelihood evaluations,
so the additional preference multiplicity is confined to the
autoregressive slate readout.
\section{Model Configurations and Analysis}
\label{app:model-configurations}

\paragraph{Model configurations.}
We use the scaling analysis in Section~\ref{sec:scaling-analysis} to select
the common TransX backbone and then instantiate the task-specific readouts
required by each post-training objective. The TransX backbone configuration is outlined in Table~\ref{tab:transx-backbone-config}.

\begin{table}[hbt]
\centering
\small
\setlength{\tabcolsep}{5pt}
\renewcommand{\arraystretch}{1.12}
\begin{tabularx}{\linewidth}{@{}p{0.27\linewidth}X@{}}
\toprule
\textbf{Component} & \textbf{Configuration} \\
\midrule

Behavior tokenizer
&
MLP $512$--$1024$--$1024$
\\

Behavior encoder
&
4 Transformer blocks; $d_{\mathrm{model}}=1024$; 32 attention heads;
FFN $1024$--$4096$--$1024$
\\

Reader memory
&
32 local behavior states plus up to 32 learned global summaries.
Base block sizes are 64, 32, and 16 states for rich-, medium-, and
sparse-history users, respectively; each block is reduced to one summary
state by learned attention pooling.
\\

Candidate encoder
&
MLP $512$--$1024$--$1024$
\\

Cross-attention fusion &
2 reader blocks; reader width $d_{\rm model}=1024$;
$H_q=32$, $G=4$ KV heads, $d_h=32$; FFN 1024--4096--1024
\\

Sparse embedding table & Roughly $300$M parameters in total (including both hashed and non-hashed ID) \\
\bottomrule
\end{tabularx}
\caption{Shared TransX backbone configuration.}
\label{tab:transx-backbone-config}
\end{table}

The task-specific configurations are as follows:
\begin{itemize}

    \item \texttt{TransX-Base}.
    The prediction readout contains two stacked DCNv2 cross layers~\citep{wang2021dcn} with
    rank $256$, followed by MLP of dimensions
    $2048$--$1024$--$512$--$256$ and a parallel Deep\&Wide tower with the same MLP dimensions. The resulting representation is mapped to the binary click and conversion outputs.

    \item \texttt{TransX-Retrieval}.
    Retrieval replaces the candidate-wise prediction readout with a
    context-level retrieval representation. A context-query encoder with
    dimensions $512$--$256$--$128$ attends to the shared behavior memory using
    $128$-dimensional key and value projections, producing a
    $128$-dimensional user embedding. The item tower reuses the corresponding
    base item tokenizer and applies a projection head to the same
    $128$-dimensional retrieval space.

    \item \texttt{TransX-Ranking}.
    \texttt{TransX-Ranking} reuses the \texttt{TransX-Base} candidate encoder, candidate-to-memory fusion, and interaction network. Post-training replaces the binary response output with a scalar ranking head and changes the supervision to the Bradley–Terry preference objective; this isolates the effect of preference specialization without changing the serving architecture.

    \item \texttt{TransX-GenRet}.
    Generative retrieval retains the shared TransX behavioral memory but replaces the discriminative readout with an autoregressive decoder over item semantic IDs (SIDs). At decoding step $k$, the evolving generative query cross-attends to the same reader memory, allowing the projected memory KV to be reused throughout generation. The decoder predicts one SID token per step until a complete item identifier is produced. We construct SIDs hierarchically from pretrained item embeddings using balanced $k$-means. The first three levels each use a vocabulary of size $2{,}048$; a fourth hash-based token is appended to disambiguate items that collide under the three-level semantic code. The autoregressive decoder is trained independently rather than initialized from the \texttt{TransX-Base} candidate readout.

    \item \texttt{TransX-GenRank}.
    Generative ranking first computes the candidate-conditioned TransX
    representations and base relevance logits for the complete candidate set.
    These representations are reused throughout slate generation. At decoding
    step $k$, a lightweight autoregressive slate decoder summarizes the
    previously selected items into a position-conditioned slate state
    $h_k$. Each remaining candidate is scored by adding a
    slate-conditioned residual to its base relevance logit,
    \[
    s_{k,i}
    =
    b_i
    +
    \frac{(W_q h_k)^\top(W_c z_i)}{\sqrt{d_\text{slate}}},
    \]
    where $z_i$ is the precomputed candidate-conditioned TransX
    representation. Previously selected items are masked from subsequent
    steps. This formulation preserves the strong pointwise relevance learned
    by \texttt{TransX-Base} while concentrating the autoregressive capacity on
    dependencies among slate selections. Candidate representations and
    behavioral memory are computed once and reused across the maximum
    eight decoding steps.
    
\end{itemize}

\paragraph{Post-training settings.}
Table~\ref{tab:posttraining-settings} records the concrete instantiation
used for the experiments reported in this work. Unless otherwise stated,
all post-training variants are initialized from TransX-Base and freeze the
shared behavior tokenizer, causal encoder, and behavior-memory construction;
only the capability-specific readout and output components are optimized.

\begin{table*}[t]
\centering
\footnotesize
\setlength{\tabcolsep}{5pt}
\renewcommand{\arraystretch}{1.18}
\begin{tabularx}{\textwidth}{
    p{2.0cm}
    p{3.3cm}
    p{3.0cm}
    >{\raggedright\arraybackslash}X
}
\toprule
Variant
& Trainable components
& Objective
& Experimental instantiation \\
\midrule

\textbf{TransX-Retrieval}
&
Context-query reader, retrieval KV projection, and item projection/head
&
InfoNCE
&
For each positive item, negatives comprise mature non-positive candidates
from the same event and deduplicated in-batch items from other user--event
examples, after removing known positives.
$\tau_{\mathrm{ret}}=0.02$;
same-event negative cap $=8$ per positive.
\\

\addlinespace

\textbf{TransX-Ranking}
&
Candidate-conditioned ranking readout and scalar ranking head
&
Event-normalized Bradley--Terry
&
Pairs are formed among candidates with mature feedback within the same
event; tied or incomparable outcomes are omitted and at most
$16$ pairs are sampled per event.
$\tau_{\mathrm{rank}}=1.0$;
$w_{ij}=\text{uniform}$.
\\

\addlinespace

\textbf{TransX-GenRet}
&
Generative retrieval reader, task-specific KV projection, and
autoregressive SID decoder
&
$\mathcal L_{\mathrm{AR}}
+\lambda_{\mathrm{DPO}}\mathcal L_{\mathrm{DPO}}$
&
Preferred items are paired with lower-preference same-event items and
catalog-valid hard negatives; synthetic negatives are retained only when
the frozen preference scorer satisfies the configured margin.
$\lambda_{\mathrm{DPO}}=0.2$,
$\beta=0.1$;
decoding: prefix-Trie constraint beam search with width=16.
\\

\addlinespace

\textbf{TransX-GenRank}
&
Autoregressive slate decoder and slate-conditioned scoring parameters
&
$\mathcal L_{\mathrm{AR}}
+\lambda_{\mathrm{DPO}}\mathcal L_{\mathrm{DPO}}$
&
For each event, 8 candidate slates are sampled and
scored by the frozen teacher $R_\omega$. Matched-length pairs are retained
when
$R_\omega(\xi_t,S^+) - R_\omega(\xi_t,S^-)
\ge \delta_{\mathrm{pref}}$.
$\lambda_{\mathrm{DPO}}=0.3$,
$\beta=0.15$,
$\delta_{\mathrm{pref}}=0.22$;
slate-sampling temperature $=0.7$;
maximum slate length $N_{\text{slate}}=8$.
\\

\bottomrule
\end{tabularx}

\caption{
Concrete post-training settings used for the reported TransX specialization
experiments. The shared behavior tokenizer, causal encoder, and
behavior-memory construction are frozen for all variants unless otherwise
noted.
}
\label{tab:posttraining-settings}
\end{table*}

\paragraph{\texttt{TransX-GenRet} and \texttt{TransX-GenRank} analysis.}
Figure~\ref{fig:DPO-results} shows that both generative specializations benefit from the \texttt{TransX-Base} backbone. For \texttt{TransX-GenRet}, post-training converges faster and consistently improves held-out preference quality, Recall@$K$, NDCG@$K$, valid-SID rate, and generation diversity relative to matched training from random initialization. \texttt{TransX-GenRank} shows the same pattern, with faster DPO convergence and stronger held-out preference separation. Together, these results indicate that the behavioral representation learned by \texttt{TransX-Base} transfers effectively to both autoregressive retrieval and slate generation, allowing post-training to focus on task-specific decoding rather than relearning user behavior from scratch.


Despite effective DPO optimization, the two generative specializations
exhibit a consistent trade-off relative to their discriminative
counterparts: greater structured-output expressiveness comes with higher
sensitivity to representation design, supervision quality, and
autoregressive decoding.
\begin{itemize}
    \item  Where \texttt{TransX-GenRet} improves over \texttt{TransX-Retrieval}, the gains are concentrated among long-tail users and items, suggesting that hierarchical SID decoding can capture item structure not fully expressed by a single embedding-space similarity. However, invalid SIDs still occur after convergence and directly reduce Recall@$K$ without constrained decoding; generated results are also more concentrated on popular items.
    \item \texttt{TransX-GenRank} is broadly comparable on average with the conventional learning-to-rank variant, but exhibits higher variance and weaker quality at some later slate positions as early autoregressive decisions propagate forward. Its performance is also more sensitive to the coverage and fidelity of slate-level preference supervision, with additional risk from exposure bias and overfitting to teacher or preference signals.
\end{itemize}

\begin{table}[hbt]
\centering
\small
\setlength{\tabcolsep}{5pt}
\renewcommand{\arraystretch}{1.12}
\begin{tabularx}{\linewidth}{@{}p{0.19\linewidth}p{0.35\linewidth}X@{}}
\toprule
\textbf{Model} & \textbf{Primary advantage} & \textbf{Primary trade-off} \\
\midrule

\texttt{TransX-GenRet}
&
Retrieval as conditional sequence generation exploits hierarchical structure in SIDs and shows potential on long-tail users/items.
&
Quality depends on SID construction, collisions, decoding strategy,
negative sampling, and diversity; invalid SIDs and popularity
concentration remain important failure modes.
\\ \hline

\texttt{TransX-GenRank}
&
Autoregressive slate generation explicitly models dependencies among slate positions and has potential to improve slate-wise recommendation quality.
&
Sequential decisions introduce error propagation and exposure bias, while
quality depends more strongly on the coverage and fidelity of slate-level
preference supervision.
\\

\bottomrule
\end{tabularx}
\caption{
Observed behaviors and trade-offs of the generative TransX specializations.
}
\label{tab:generative-tradeoffs}
\end{table}

Finally, these generative experiments are intended to establish that the shared TransX behavioral representation supports autoregressive retrieval and slate-generation readouts and that preference optimization improves the corresponding held-out generative diagnostics. They are not intended as matched-budget evidence that the generative variants outperform their discriminative counterparts.
\end{document}